\newif\iffullversion
\newif\ifdraft
\newif\ifanonymous
\fullversiontrue
\documentclass{article}

\usepackage[utf8]{inputenc}
\usepackage[T1]{fontenc}
\usepackage{paralist}
\usepackage{amsmath,amssymb}
\usepackage{paralist,stmaryrd,mathpartir}
\usepackage{tocloft}
\usepackage[framemethod=tikz]{mdframed}
\usepackage{declmath}
\usepackage{makeidx}\makeindex
\usepackage[a4paper]{geometry}
\usepackage{version}
\usepackage[pdfborderstyle={/S/U/W 0.3}]{hyperref}
\usepackage{circuits}
\usepackage[users=Q]{uchanges}
\usepackage[backend=bibtex,maxbibnames=100]{biblatex}
\usepackage{xr}

\newcommand\qrhlautoref[1]{\cite[\autoref{qrhlarxiv:#1}]{qrhl-arxiv-v2}}

\input{macros}

\usetikzlibrary{decorations.pathreplacing, decorations.pathmorphing, calc}

\let\norulelemmalink1

\newcommand\crossWire[2]{
  \draw[quantumWire] (\getWireCoord{#1}) to (\getWireCoord{#2} -| \currentXPos);
}

\tikzset{killWire/.style={%
    gate={#1},
    draw=none,
    fill,
    inner sep=0pt,
    minimum width=.5mm, minimum height=3mm
}}

\tikzset{boxAround/.style={
    draw, dotted,
  }}

\makeatletter

\tikzset{boxAroundLabeled/.code={
      \pgfkeysalso{boxAround}
      \pgfkeysalso{
        append after command={node[anchor=south west, inner xsep=0pt, inner ysep=1pt,
                                   name={\@previous@nodeName-label},
                                   at={(\@previous@nodeName.north west)}] {\tiny #1}}}
      \appto\tikz@atend@node{\xdef\@previous@nodeName{\tikz@fig@name}}%
    }
  }

\newcommand{\mathcolorbox}[2]{\mathchoice%
  {\colorbox{#1}{$\displaystyle#2$}}%
  {\colorbox{#1}{$\textstyle#2$}}%
  {\colorbox{#1}{$\scriptstyle#2$}}%
  {\colorbox{#1}{$\scriptscriptstyle#2$}}}%

\def\symbolindexmarkhighlight#1{%
  \setlength{\fboxsep}{2pt}%
  {\TextOrMath{\colorbox{gray!20}{#1}}{\mathcolorbox{gray!20}{#1}}}}

\iffullversion
\newcommand\fullshort[2]{#1}
\includeversion{fullversion}
\excludeversion{shortversion}
\else
\newcommand\fullshort[2]{#2}
\excludeversion{fullversion}
\includeversion{shortversion}
\fi

\newcommand\notanonymous{\ifanonymous \TODOQ{This is not anonymous}\fi}

\newcommand\fullonly[1]{\fullshort{#1}{}}

\newcounter{claimstep}
\newtheorem{claim}{Claim}[claimstep]

\mdfdefinestyle{theoremstyle}{%
  linecolor=black,
  backgroundcolor=black!10!white,
  frametitlebelowskip=0pt,
  innerleftmargin=5pt,
  innerrightmargin=5pt,
}

\mdtheorem[style=theoremstyle]{definition}{Definition}
\mdtheorem[style=theoremstyle]{lemma}{Lemma}
\mdtheorem[style=theoremstyle]{corollary}{Corollary}

\AtBeginDocument{
  \ifx\PreviewMacro\undefined\else
  \PreviewEnvironment*{figure}
  \PreviewEnvironment*{lemma}
  \PreviewMacro*\footnote
  \PreviewMacro[{[]{}{}{}}]\RULE
  \PreviewMacro[{[]{}{}{}}]\ERULE
  \PreviewMacro*\symbolindexmark
  \PreviewMacro*\symbolindexmarkhighlight
  \fi}

\bibliography{localvars}

\begin{document}

\title{Local Variables and Quantum Relational Hoare Logic}
\author{Dominique Unruh\\\small University of Tartu}

\maketitle

\ifdraft
\begin{center}
  \bfseries\Huge\fboxsep=10pt
  \framebox{THIS IS A DRAFT}
\end{center}
\fi

\begin{abstract}
  We add local
  variables to quantum relational Hoare logic (Unruh, POPL 2019). We
  derive reasoning rules for supporting local variables (including an
  improved ``adversary rule''). We extended the \texttt{qrhl-tool} for
  computer-aided verification of qRHL to support local variables and
  our new reasoning rules.
\end{abstract}

\tableofcontents

\section{Introduction}

In this work, we add local variables to the programming language
underlying the quantum relational Hoare logic (qRHL) from \cite{qrhl},
develop some reasoning rules related to this change, and added support
for our extensions to the \texttt{qrhl-tool} \cite{qrhl-tool} that
enables computer-verified reasoning in qRHL.

qRHL is a logic that
allows us to establish pre- and postconditions of pairs of quantum
programs, thereby reasoning about the relationship between those two
programs. (E.g., in the simplest case, establish that they do the same
thing.) qRHL was designed with security proofs for quantum
cryptography in mind, following the example of probabilistic
relational Hoare logic (pRHL) \cite{certicrypt} using in the EasyCrypt tool
for classical security proofs.

To understand the motivation for and challenges in adding local
variables, we first explain a bit of the background and motivation behind qRHL:

\paragraph{Post-quantum security.} %
Quantum computers have long been known to be a potential threat to
cryptographic protocols, in particular public key encryption.
Shor's algorithm \cite{Shor:1994:Algorithms} allows us to efficiently
solve the integer factorization and discrete logarithm problems, thus
breaking RSA and ElGamal and variants thereof. 
This breaks all commonly
used public key encryption and signature schemes.
Of course, as of today, there are no quantum computers that even come
close to being able to execute Shor's algorithm on reasonable problem
sizes.
Yet, there is constant progress towards larger and more powerful
quantum computers (see, e.g., the recent breakthrough by Google
\cite{arute19supremacy}). 
In light of this, it is likely that quantum computers will be able to
break today's public key encryption and signature schemes\fullonly{ (and possibly other
kinds of cryptosystems)} in the foreseeable future.
Since the development, standardization, and industrial deployment of a
cryptosystem can take many years, we need to develop and analyze
future post-quantum secure protocols already today.
One important step in this direction is the NIST post-quantum
competition \cite{nist-pqc} that will select a few post-quantum public-key
encryption and signature schemes for industrial standardization.

\paragraph{Verifying classical cryptography using pRHL.}
Cryptographic security proofs tend to be complex, and, due to their
complexity, error prone.
Small mistakes in a proof can be difficult to notice and may
invalidate the whole proof.
For example, the proof of the
OAEP construction \cite{BeRo_94} went through a number of
fixes \cite{JC:Shoup02,C:FOPS01,JC:FOPS04} until it was finally
formally proven in \cite{RSA:BGLZ11} after years of industrial use.
The
PRF/PRP switching lemma was a standard textbook example for many
years before it was shown that the standard proof is flawed
\cite{EC:BelRog06}.
And more recently, an attack on the ISO standardized blockcipher mode
OCB2 \cite{ocb2-iso} was found \cite{inoue19ocb2}, even though OCB2 was
believed to be proven secure by \cite{rogaway04tweakable}.

While a rigorous and well-structured proof style
(e.g., using sequences of games as advocated in \cite{EC:BelRog06,EPRINT:Shoup04})
can reduce the potential for hidden errors and imprecisions,
it is still very hard to write a proof that is 100\% correct.
And especially if a mistake in a proof happens in a step that seems very intuitive,
it is quite likely that the mistake will also not be spotted by a reader.

To avoid this, formal (computer-aided) verification can be
employed. Typically, a formal version of the sequences-of-games
approach is used. In this approach, roughly speaking, the security of
a cryptographic scheme is represented by the probability that a
certain event happens in a certain program (encoding both the
adversary and the scheme), and then this game is rewritten
step-by-step, and on each step, it is shown that the old and new game
stand in some relationship, until a final game is reached for which
determining the probability of the event of interest is trivial to
bound.

A number of frameworks/tools use this approach for verifying classical
cryptography\fullshort{:
CryptoVerif \cite{cryptoverif},
CertiCrypt \cite{certicrypt}, EasyCrypt \cite{easycrypt}, FCF
\cite{FCF}, CryptHOL \cite{crypthol}, and Verypto
\cite{verypto}}{, e.g., EasyCrypt \cite{easycrypt}}. 
\fullshort{CryptoVerif tries to automatically determine a sequence of games by
using a set of fixed rewriting rules for games.
This has the advantage of reducing user effort, but it also means that
the framework is more limited in terms of what game transformations
are possible.
In contrast, the other frameworks require}{EasyCrypt requires} the user to explicitly
specify the games that constitute the security proof (as is done in a
pen-and-paper proof),
and to additionally provide justification for the fact that two
consecutive games are indeed related as claimed.
This justification will often be considerably more detailed than in a
pen-and-paper proof where the fact that two slightly different games
are equivalent will often be declared to be obvious.

\fullshort{One}{Their} approach for proving the relationship of consecutive games is to
give a proof in relational Hoare logic. 
Relational Hoare logic is a logic that allows us to express the
relationship between two programs by specifying a relational
precondition and a relational postcondition. 
A relational Hoare judgment of the form $\rhl\PA\bc\bd\PB$ intuitively
means that
if the variables of the programs $\bc$ and $\bd$ are related as
described by the precondition $\PA$ before execution, and we execute
$\bc$ and $\bd$,
then afterwards their variables will be related as described by $\PB$.
A very simple example would be $\rhl{x_1\leq x_2}{\assign
  x{x+1}}{\assign x{x+1}}{x_1\leq x_2}$.
This means that is the variable $x$ in the left program is
smaller-equal than in right one,
and both programs increase $x$, then $x$ in the left program will
still be smaller-equal than in the right one.
As this example shows, relational Hoare logic can express more
complex relationships than simple equivalence of two games.
This makes the approach very powerful.
To reason about cryptography, one needs a variant of relational Hoare
logic that supports probabilistic programs. Such a probabilistic
relational Hoare logic (pRHL) was developed for this purpose by
Barthe, Grégoire, and Zanella Béguelin \cite{certicrypt}.
\fullshort{Both CertiCrypt \cite{certicrypt} and its popular successor EasyCrypt
use}{EasyCrypt uses} pRHL for proving the relationship between cryptographic games.

\paragraph{Verifying quantum cryptography using qRHL.}
If we wish to follow the EasyCrypt approach to verify security proofs
of quantum cryptographic schemes (be it actual quantum protocols, or
merely post-quantum secure schemes that withstand quantum attacks), we
cannot use pRHL but need a logic that allows us to reason about
quantum programs, i.e., programs that can operate on quantum
data. Such a logic was proposed in \cite{qrhl}, namely quantum
relational Hoare logic (qRHL). Inspired by qPRHL, this logic allows us
to write judgments of the form $\rhl\PA\bc\bd\PB$ which mean,
informally, that if the predicate $\PA$ is satisfied by a pair of
quantum memories $M_1,M_2$, and we execute the quantum programs
$\bc,\bd$ on those memories, then $M_1,M_2$ satisfy $\PB$ afterwards.
Since the quantum memories do not contain classical values, the
predicates $\PA,\PB$ are not predicates in the classical sense. We
will make this more formal later, for now it is sufficient to
understand that those predicate can express conditions both about the
classical and quantum variables in the memories $M_1,M_2$, analogously
to what is done in qRHL. (E.g., state that they are equal.)

It was argued in \cite{qrhl} that qRHL is suitable for reasoning about
cryptography in the quantum setting. To demonstrate this, they
developed a tool for computer-aided verification of proofs in qRHL
(called the \texttt{qrhl-tool} henceforth), and did several example
verifications, such as a verification of quantum teleportation and one
of the post-quantum security of a very simple encryption scheme.
However, those were toy examples only, and did not shed light on the
scalability of the approach.  To resolve this issue,
\cite{pqfo-verify} attempted formal computer-aided verification of a
non-trivial post-quantum cryptographic proof of a state-of-the-art
construction of an encryption scheme (a variant of the
Fujisaki-Okamoto transform \cite{IEICE2000:FujisakiO} analyzed in
\cite{hovelmanns20generic}).  This verification was performed using
the \texttt{qrhl-tool}.  The upshot of that case study was that, in
principle, qRHL is suitable for analyzing more complex cryptographic
schemes, but several limitations were identified. One of them
concerned the absence of support for local variables in qRHL.  It
turned out that without support for local variables, formalizing the
whole proof was at least very difficult.\footnote{%
  Changes in one subproof tended to need a refactoring of most other
  subproofs (affecting the variables the other subproofs talked
  about).  And this refactoring then required new changes in other
  subproofs.  It was not clear whether this process would end
  eventually (without the additions to the logic introduced in the
  present work). On the other hand, seen separately, each subproof
  seemed easy to finish. (Which is why the toy examples from
  \cite{qrhl} did not uncover this difficulty.)}  These difficulties
prompted the developments in the present paper, both theoretical ones
and extensions of \texttt{qrhl-tool} (which in turn were used to
finish the proofs in \cite{pqfo-verify}).

\paragraph{Local variables.} We first explain the reasons why
\cite{qrhl} did not include local variables in the definition of the
languages. By a local variable we mean a variable such that any
changes of their values are limited to a specific scope. That is, a
read/write operation on a local variable has no observable effect
outside the scope of the local variable. At a first glance, it may
seem that it is easily possible to simulate local variables even if
the language does not have builtin support for them. Namely, any
procedure or program fragment that is supposed to use a local variable
can use a global variable instead, as long as we ensure that no such
``pseudo-local'' variable is used in more than one program
fragment.\footnote{For example, we could adopt a disciplined naming
  strategy that prefixes local variables with the names of the
  procedures they are used in, and to initialize all local variables
  before use. This would make sure that local variables are never
  accessed outside their intended scopes. (In the presence of
  recursion this would not work because recursive invocations would
  access the same variables at the same time. But the language of qRHL
  does not support recursion anyway.)} While this approach is less
convenient than having built-in support for local variables, it would
seem sufficient at least for handling small programs. (In the case
study \cite{pqfo-verify}, the programs tend to be a few dozen lines at
most, even including the subprocedures they invoke.) For this reason,
qRHL as defined in \cite{qrhl} (and thus also \texttt{qrhl-tool}) did
not include support for local variables in order to keep the language
and logic simple.

However, as the case study \cite{pqfo-verify} revealed, this argument
is not correct when quantum variables are involved. (It probably works
well for classical programs, i.e., for pRHL.) To understand why, we
first look at an extremely simple example how reasoning with
pseudo-local variables would work with classical programs (i.e., in
pRHL). Say $\xx$ is a global variable and $\yy$ is a pseudo-local
variable (i.e., that is never accessed by other programs).  Let
$\bc_1:=\bd_1:=\paren{\assign\yy{f(\xx)}}$ be a program that assigned
a function of $\xx$ to $\yy$..  We want to express in pRHL that the
two programs have the same observable behavior.  We express this as
$\rhl{\xx_1=\xx_2}{\bc_1}{\bd_1}{\xx_1=\xx_2}$, i.e., if $\xx$ is the
same before execution of $\bc_1$ or $\bd_1$, $\xx$ is the same afterwards.
(This is very trivial since $\bc_1=\bd_1$, but as we will see, in the
quantum setting, even this seemingly trivial case leads to problems.)
One way to prove this is to show that
$\rhl{\xx_1=\xx_2}{\bc_1}{\bd_1}{(\xx_1,\yy_1)=(\xx_2,\yy_2)}$ holds (we do
not need to include $\yy_1,\yy_2$ in the precondition since $\yy$ is
overwritten in both programs), and then use the fact that
$(\xx_1,\yy_1)=(\xx_2,\yy_2)\implies \xx_1=\xx_2$ to conclude
$\rhl{\xx_1=\xx_2}{\bc_1}{\bd_1}{\xx_1=\xx_2}$.

Now consider an analogous
example involving quantum variables.  Say $\qq$ is a global quantum
variable and $\rr$ is pseudo-local.  Let
$\bc_2:=\bd_2:=\paren{\Qinit\rr{\basis{}0};\ \Qapply U{\qq\rr}}$. (That
is, $\rr$ is initialized with a fixed state $\basis{}0$, and then the
unitary $U$ is jointly applied to $\qq,\rr$. E.g., $U$ could be a
CNOT.)  Again, we want to show that $\bc$ and $\bd$ have the same
observable behavior (given access only to $\qq$).  This can be
expressed in qRHL
$\rhl{\qq_1\quanteq\qq_2}{\bc_2}{\bd_2}{\qq_1\quanteq\qq_2}$.  Here $\quanteq$
is the quantum equality introduced in \cite{qrhl}, intuitively it
expresses that two variables (or two tuples of variables) have the
same value. We try to follow the same approach as in the classical
case.  Using the reasoning rules from \cite{qrhl}, it is easy to show
that $\rhl{\qq_1\quanteq\qq_2}{\bc_2}{\bd_2}{\qq_1\rr_1\quanteq\qq_2\rr_2}$
holds.  Then, if $\qq_1\rr_1\quanteq\qq_2\rr_2$ (meaning that $\qq\rr$
in memory $M_1$ jointly are equal in content to $\qq\rr$ in $M_2$)
would imply $\qq_1\quanteq\qq_2$, we could conclude
$\rhl{\qq_1\quanteq\qq_2}{\bc_2}{\bd_2}{\qq_1\quanteq\qq_2}$ Unfortunately,
the quantum equality is more peculiar than the classical one.
$\qq_1\rr_1\quanteq\qq_2\rr_2$ is not the same as
$\qq_1\quanteq\qq_2 \cap \rr_1\quanteq\rr_2$ (intersection $\cap$ is
the analogue of conjuction $\land$ for quantum predicates) and does
not imply $\qq_1\quanteq\qq_2$.\footnote{The converse holds:
  $\qq_1\quanteq\qq_2 \cap \rr_1\quanteq\rr_2$ implies
  $\qq_1\rr_1\quanteq\qq_2\rr_2$.}  This is because a quantum equality
$\QQ_1\quanteq\QQ_2$ not only implies that $\QQ_1$ and $\QQ_2$ have
the same content but that also that $\QQ_1$ and $\QQ_2$ are not
entangled with any other variables. Thus $\qq_1\quanteq\qq_2$ would
imply that $\qq_1,\qq_2$ are not entangled with anything else but that
is not implied by $\qq_1\rr_1\quanteq\qq_2\rr_2$ (and in fact does not
even hold after running $\bc,\bd$).

This issue means that even though $\rr$ has no relevance outside of
$\bc,\bd$, we have to carry information about $\rr$ in our
postconditions. The effect of this is that local variables ``spread''
through the invariants used in other parts of the proof as described
previously, making it very hard to find consistent invariants and
breaking the modularity of proofs.

Can this problem be resolved? Instead of $\bc,\bd$ as defined above,
we could define them as:
$\bc_3:=\bd_3:=\paren{\Qinit\rr{\basis{}0};\ \Qapply U{\qq\rr};\
  \Qinit\rr{\basis{}0}}$. (Or stated more generally, initialize any
pseudo-local variable initialized before use, overwrite it at the end
of its scope.) We have that
\begin{equation}\label{eq:remove.r}
\rhl{\qq_1\rr_1\quanteq\qq_2\rr_2}{\Qinit\rr{\basis{}0}}{\Qinit\rr{\basis{}0}}{\qq_1\quanteq\qq_2},
\end{equation}
in other words, if we overwrite a quantum variables occurring in a
quantum equality, that variable can be removed from the quantum
equality. Then
$\rhl{\qq_1\quanteq\qq_2}{\bc_3}{\bd_3}{\rr_1\quanteq\rr_2}$ follows
immediate from
$\rhl{\qq_1\quanteq\qq_2}{\bc_2}{\bd_2}{\qq_1\rr_1\quanteq\qq_2\rr_2}$
(from the previous paragraph) and \eqref{eq:remove.r} by the
\rulerefx{Seq} rule from \cite{qrhl}. Judgment \eqref{eq:remove.r} cannot be
proven using the rules from \cite{qrhl}.\footnote{%
  At least it is not obvious which rules to use. We have no formal
  proof that it does not follow from a nontrivial combination of the
  rules there.}. One of the results of the present work is a reasoning
rule \rulerefx{JointQInitEq} from which \eqref{eq:remove.r} is an
immediate consequence.

Proving \ruleref{JointQInitEq} would probably be enough to have
rudimentary support for pseudo-local variables (when following all the
guidelines mentioned above about keeping names separate, and
initializing and overwriting). However, it seems quite inconvenient to
do so in a larger project. Furthermore, if local variables are not
explicitly declared as such, they will show up in, e.g., the set of
free variables of a program. For example, the \rulerefx{Adversary}
rule from \cite{qrhl} allows us to reason about program fragments as a
black box (i.e., without needing to look at their concrete
implementation, more about that later) but it depends on the set of
free variables of a program. Since the rule would not recognize that
some of the free variables are pseudo-local, the pseudo-local
variables would creep back into the pre-/postconditions produced by
the adversary rule.

In light of those challenges, it seems that for making qRHL and
\texttt{qrhl-tool} usable for larger projects, built-in support for
local variables is a high priority. This is what we set out to do in
the present work.

\subsection{Our contribution}
\label{sec:isabelle-proofs}

We add local variables to the programming language underlying qRHL and
prove sound reasoning rules to work with local variables. Furthermore,
we extend the \texttt{qrhl-tool} to support reasoning with local
variables. In more detail:
\begin{itemize}
\item We extend the language by a construct for declaring local
  variables (Sections \ref{sec:syn.prog} and~\ref{sec:sem.progs}). If $\vv$ is a variable (classical or quantum) and $\bc$
  is a program containing variable $\vv$, then $\local\vv\bc$ is the
  program where $\vv$ is local. That is, the value of $\vv$ is saved
  before executing $\bc$ and restored afterwards. (Of course, if $\vv$
  is quantum, storing does not mean making a copy.)

  Based on this, we derive a number of laws for denotational
  equivalence of programs involving $\localkw$ (such as invariance
  under $\alpha$-renaming, commutativity of nested
  $\localkw$-statements, moving of $\localkw$-statements,
  adding/removing initializations of local variables, etc.).  Closely
  related, we also introduce some laws concerning the initialization
  of variables (e.g., when an initialization has no effect because the
  variable is overwritten). The latter laws are not directly related
  to local variables but turn out to come up over and over while
  deriving our theory of local variables. (\autoref{sec:sem.lemmas})

\item Basic reasoning rules for qRHL statements: We provide sound
  reasoning rules for qRHL for remove local variable declarations and
  to rename variables. We need to remove local variable declarations
  to be able to break down a qRHL judgment into judgments about more
  elementary programs. E.g., we show judgments of the form
  $\rhl{\PA}{\bc_1;\bc_2}{\bd_1;\bd_2}{\PC}$ by by showing judgments
  $\rhl{\PA}{\bc_1}{\bd_1}\PB$ and $\rhl{\PB}{\bc_1}{\bd_1}\PC$ and
  then using the \rulerefx{Seq} rule. To do the same with a goal of
  the form
  $\rhl{\PA}{\local\vv{\bc_1;\bc_2}}{\local\vv{\bd_1;\bd_2}}{\PC}$ we
  first need to remove the $\localkw$-declaration, last but not least
  because we may want to refer to $\vv$ in $\PB$. Very roughly
  speaking, the rule says that to prove
  $\rhl{\PA}{\local\vv{\bc_1;\bc_2}}{\local\vv{\bd_1;\bd_2}}{\PC}$ it
  is sufficient to prove
  $\rhl{\PA}{\bc_1;\bc_2}{\bd_1;\bd_2}{\PC}$. (Interestingly, the
  converse does not hold.)

  These rules are given in \autoref{sec:basic.var.rules}.
  
\item As explained above, the rules of qRHL from \cite{qrhl} do not
  allow us to derive that
  \begin{equation}
    \rhl{\qq_1\rr_1\quanteq\qq_2\rr_2}{\Qinit{\rr}{\basis{}{0}}}
    {\Qinit{\rr}{\basis{}{0}}}{\qq_1\quanteq\qq_2}.
    \label{eq:qqrr.eq}
  \end{equation}
  That is, we cannot get rid of variables that occur in a quantum
  equality, even if these variables are overwritten (which is
  essentially the same as erasing them). This is because the rules for
  quantum initialization in \cite{qrhl} (\rulerefx{QInit1}\textsc{/2})
  are one-sided rules. That means they consider only an initialization
  (e.g., $\Qinit\rr{\basis{}0}$) in the left or the right program but
  not both simultaneously.  To derive \eqref{eq:qqrr.eq}, though, we need a
  rule that operates on both initializations simultaneously
  (intuitively, to make sure the entanglement between $\qq$ and $\rr$
  is handled in a synchronized fashion on the left and right
  side.\footnote{%
    The need for two-sided rules is not a new observation. Even in the
    classical pRHL \cite{certicrypt}, we have a two-sided rule for
    probabilistic sampling that ``synchronizes'' the random choices on
    the left and right side. This rule cannot be emulated using two
    applications of the one-sided rule for samplings. Similarly, qRHL
    \cite{qrhl} has a two-sided rule for measurements, synchronizing
    the measurement outcomes. However, for assignments, there is no
    two-sided rule in pRHL because there seems to be nothing that this
    rule could achieve that cannot be achieved with two consecutive
    applications of the one-sided rule. Thus it comes as a bit of a
    surprise that the quantum analogue to an assignment does need a
    two-sided rule.}
  We prove such a rule (\rulerefx{JointQInitEq}).

  As a consequence, we also prove a two-sided rule for removing local
  variables from qRHL judgments (\rulerefx{JointRemoveLocal}). Put
  simply, we show that to show a judgment such as
  $\rhl{\qq_1\quanteq\qq_2} {\local\rr\bc}{\local\rr\bd}
  {\qq_1\quanteq\qq_2}$, it is sufficient to show
  $\rhl{\qq_1\rr_1\quanteq\qq_2\rr_2} {\bc}{\bd}
  {\qq_1\rr_1\quanteq\qq_2\rr_2}$. The fact that $\rr$ is included in
  the quantum equality (with the one-sided rules
  \rulerefx{RemoveLocal1}\textsc{/2} it would not be) makes this
  judgment easier to prove.
  $\rhl{\qq_1\quanteq\qq_2} {\bc}{\bd} {\qq_1\quanteq\qq_2}$ would
  only be provable if $\bc,\bd$ do not create any entanglement between
  $\qq$ and $\rr$.

  The \ruleref{JointRemoveLocal} in turn is crucial in the derivation
  of the \rulerefx{Adversary} rule (see below).

  As a simple corollary of \rulerefx{JointQInitEq}, we also get a
  strengthening of the \rulerefx{QrhlElimEq} rule from \cite{qrhl}
  that allows us to relate qRHL judgments and (in)equalities of
  probabilities involving programs, we call the new rule
  \rulerefx{QrhlElimEqNew}.
  
  (The three new rules are presented in \autoref{sec:two-sided.init}.)

\item Variable changing: The contributions described above already go
  a long way towards making it possible to work with local variables
  in qRHL proofs. However, we still cannot have modular proofs (in the
  sense that one part of the proof does not have to depend on which
  local variables occur in another part of the proof). Consider the
  following example: Say, we want to prove a qRHL judgment of the form
  \begin{equation}
    \rhl{\qq_1\quanteq\qq_2}{\local\rr{\bc_0;\bc}}{\local\rr{\bc_0;\bd}}{\qq_1\quanteq\qq_2}.
    \label{eq:bc0.bc}
  \end{equation}
  Say the programs $\bc,\bd$ are complex subroutines that we wish to
  handle in a different subproof. Since $\bc_0$ might entangle $\qq$
  and $\rr$, proving \eqref{eq:bc0.bc}, we might end up having to
  prove the subgoal
  $X :=
  \rhl{\qq_1\rr_1\quanteq\qq_2\rr_2}{\bc}{\bd}{\qq_1\rr_1\quanteq\qq_2\rr_2}$. This
  breaks the modularity of the overall proof because now our analysis
  of $\bc,\bd$ needs to know which local variables (namely, $\rr$) are
  used in a different part of the overall proof (namely, the analysis
  of ${\local\rr{\bc_0;\bc}}$ and ${\local\rr{\bc_0;\bd}} $). Even
  worse, if $\bc,\bd$ appear in different place where different local
  variables are used, we may have to prove several different variants
  of $X$, all differing only in which local variable(s) are included
  in the quantum equality. What we want to do it to prove a single
  theorem about $\bc,\bd$ not mentioning $\rr$, say
  $X_0 := \rhl{\qq_1\quanteq\qq_2}{\bc}{\bd}{\qq_1\quanteq\qq_2}$, and
  to be able to derive $X$ from it whenever needed. Unfortunately, we
  do not know whether $X_0$ implies $X$. However, we do prove a rule
  (\rulerefx{EqVarChange}, \autoref{sec:varchange}) that allows us to derive $X$ from a theorem
  of the form
  $X_1 :=
  \rhl{\qq_1\qq_{\mathit{aux},1}\quanteq\qq_2\qq_{\mathit{aux},2}}{\bc}{\bd}{\qq_1\qq_{\mathit{aux},1}\quanteq\qq_2\qq_{\mathit{aux},2}}$
  where $\qq_\mathit{aux}$ is an auxiliary variable that is never used
  anywhere. (Basically, $X_1$ says that equality is preserved even in
  a larger context.) It might seem as if we can derive $X$ from $X_1$
  simply by renaming $\qq_{\mathit{aux}}$ into $\rr$ (using our rules
  for renaming variables), but that is not possible because
  $\qq_{\mathit{aux}}$ and $\rr$ might not have the same
  type. Requiring $\qq_\mathtt{aux}$ and $\rr$ to have the same type
  would break the modularity of the proof again, and furthermore there
  might be more than just one local variable, while our theorem $X_1$
  always uses the same single auxiliary variable $\qq_\mathit{aux}$.

  The \ruleref{EqVarChange} is also crucial in the derivation of the
  \rulerefx{Adversary} rule (see below).

\item Adversary rule: Proofs in qRHL (and in other Hoare logics) are
  often performed by deriving a judgment about the whole program from
  judgements about the individual statements in that program. However,
  in a cryptographic context, this is not always possible. We often
  need to reason about unknown fragments of code, namely whenever we
  reason about the behavior of an adversary attacking the
  cryptographic scheme. (From a logical perspective, an adversary is
  simply a program whose precise code is not known.)  Of course, if we
  do not know the code of a program $\bc$, we cannot say much about
  the pre- and postconditions. However, what we do now is, informally,
  that if the same program $\bc$ is used on the left and right side,
  and the variables of both instances of $\bc$ have the same value,
  then both instances will behave the same. That is,
  $\rhl{\QQ_1\quanteq\QQ_2}\bc\bc{\QQ_1\quanteq\QQ_2}$ if $\QQ$
  contains all free variables of $\bc$. Or, in a more general
  situation, we have
  $\rhl{\QQ_1\quanteq\QQ_2}{C[s]}{C[s']}{\QQ_1\quanteq\QQ_2}$ if
  $\rhl{\QQ_1\quanteq\QQ_2}{s}{s'}{\QQ_1\quanteq\QQ_2}$.  This would
  be used in a situation where the adversary is represented by an
  unknown context $C$, and that invokes some known procedure $s$ (or
  $s'$), e.g., $s,s'$ might be some real/fake encryption oracle.
  And since $s,s'$ are known, we can manually prove 
  $\rhl{\QQ_1\quanteq\QQ_2}{s}{s'}{\QQ_1\quanteq\QQ_2}$.

  Situations like the examples above (where unknown but identical code
  occurs on both sides) are handled by an adversary rule. In the
  classical setting, an adversary rule was already introduced in pRHL
  \cite{certicrypt}. Also in qRHL \cite{qrhl}, we have an adversary rule
  \texttt{Adversary}. However, the rule presented there has several
  drawbacks in our setting:
  \begin{compactitem}
  \item In the presence of local variables, its proof does not apply
    any more. This is because the proof is by induction over the
    structure of the adversary/context $C$. But the introduction of
    local variables means that there is another case that would need
    to be covered in the induction (namely,
    $C=\local\vv{C'}$). Dealing with local variables makes the rule
    and the induction more complex because we need to make sure the
    rule correctly handles cases where a variable of $s$ is local in
    $C$ (and thus also local in $C[s]$).
  \item The adversary rule from \cite{qrhl} requires the quantum
    equality $\QQ_1\quanteq\QQ_2$ to be the same in the precondition
    and postcondition of
    $\rhl{\QQ_1\quanteq\QQ_2}{C[s]}{C[s']}{\QQ_1\quanteq\QQ_2}$, and
    in the subgoal
    $\rhl{\QQ_1\quanteq\QQ_2}{s}{s'}{\QQ_1\quanteq\QQ_2}$. However,
    this is unnecessarily restrictive. E.g., if $C$ has local
    variables, those might occur in the subgoal but not in the
    pre-/postcondition. Or if $C$ initializes certain variables before
    use, then they can be omitted from the precondition but not from
    the postcondition.
  \end{compactitem}

  We present a new \ruleref{Adversary} that solves the these
  problems. Our rule is considerably more fine-grained than the
  original rule in that allows us to include different variable sets
  in pre-/postconditions and subgoals, and that it takes into account
  various kinds of overwritten, local, and read-only variables.

  The proof of the adversary rule relies in particular on the rules
  \rulerefx{JointRemoveLocal} and \rulerefx{EqVarChange} to maintain
  the induction hypothesis even below $\localkw$-statements.
\item New/rewritten tactics: In theory, all we need in order to do
  proofs in qRHL are the rules introduced above and in \cite{qrhl}. In
  practice, however, manually doing proofs is too cumbersome and
  error-prone. Instead, \cite{qrhl} introduced the \texttt{qrhl-tool}
  that allows to develop and check qRHL proofs interactively on the
  computer. To use the new rules we introduce in this work, we
  implemented a number of new tactics: \texttt{rename} for renaming
  variables (\autopageref{page:tactic:rename}), \texttt{local remove} for removing local
  variables (\autopageref{page:tactic:remove}),\footnote{We have not implemented the two-sided removal
    via \ruleref{JointRemoveLocal}, but that rule is implicitly
    present in the adversary rule.}
  \texttt{local up} for moving local variables to the top of a program (\autopageref{page:tactic:local-up}),
  \texttt{conseq qrhl} for changing
  variables in a quantum equality using \ruleref{EqVarChange} (\autopageref{page:tactic:conseq}),
  \texttt{equal} implementing the adversary rule (\autopageref{page:tactic:equal}, this tactic existed
  before but we completely rewrote it based on our new
  \rulerefx{Adversary} rule). We also strengthened the tactic \texttt{byqrhl} that introduces qRHL subgoals in
  the first place, using the new \ruleref{QrhlElimEqNew} (\autopageref{page:tactic:byqrhl}).

  In this paper, we only briefly sketch what those tactics do. For
  details, see the user manual of \texttt{qrhl-tool}, version 0.5.
\end{itemize}

Some of the results in this paper are shown in Isabelle/HOL
\cite{isabelle}. This concerns especially results which involve
inductions with many side conditions (such proofs are particular error
prone when done by hand).  Those proofs are not proofs from first
princples and/or based on the semantics of the language. For this, we
would need developments in operator theory that are not yet available
in Isabelle/HOL. Instead, we axiomatize the language and semantics,
and base all proofs on an explicit list of axioms in the file
\texttt{Assumptions.thy}. Those are either facts shown in \cite{qrhl},
in manual proofs in this paper, or that are elementary.  This approach
gives us a good trade-off -- avoiding errors in proofs that involve
many technical conditions, but at the same time avoiding the extreme
effort of formalizing everything in Isabelle/HOL.   The Isabelle/HOL formalization consists of 4315
lines of code.
The Isabelle theory
files for Isabelle/HOL (version Isabelle-2020) are available here
\cite{isabelle-thys}. 

\section{Preliminaries}

We introduce the notation used in this work. See also the symbol index at the end of this paper.

\paragraph{Variables.}
A \emph{program variable}%
\index{program variable}%
\index{variable!program} $\xx$
(short: variable) is an identifier annotated with a
\symbolindexmark{\typev}set $\typev \xx\neq\varnothing$,
and with a flag that determined whether the variable is
\emph{quantum}%
\index{quantum variable}%
\index{variable!quantum}%
\index{program variable!quantum} or \emph{classical}.%
\index{classical variable}%
\index{variable!classical}%
\index{program variable!classical} (In our semantics, for classical variables $\xx$ the type
$\typev\xx$
will be the set of all values a classical variable can store. Quantum
variables $\qq$ can store superpositions of values in $\typev\qq$.)

We will usually denote classical variables
with \symbolindexmark\xx{$\xx,\yy$}
and quantum variables with \symbolindexmark\qq{$\qq$}.
Given a set $V$
of variables, we write \symbolindexmark\cl{$\cl V$}
for the classical variables in $V$
and \symbolindexmark\qu{$\qu V$}
for the quantum variables in $V$.

Given a set $V$ of variables, we write \symbolindexmark\types{$\types V$} for the set of all
functions $f$ on $V$ with $f(\xx)\in \typev\xx$ for all $\xx\in V$. (I.e.,
the dependent product $\types V=\prod_{\xx\in V}\typev\xx$.)

Intuitively, $\types V$
is the set of all memories that assign a classical value to each variable in~$V$.

Given a list $V=(\xx_1,\dots,\xx_n)$
of variables, \symbolindexmarkonly\typel$\symbolindexmarkhighlight{\typel V}:=\typev{\xx_1}\times\dots\times\typev{\xx_n}$.
Note that if $V$
is a list with distinct elements, and $V'$
is the set of those elements, then $\typel V$
and $\types{V'}$
are still not the same set, but their elements can be identified
canonically. Roughly speaking, for a list $V$,
the components of $m\in\typel V$
are indexed by natural numbers (and are therefore independent of the
names of the variables in $V$),
while for a set $V$,
the components of $m\in\types V$ are indexed by variable names.

Given disjoint sets $\VV,\WW$
of variables, we write $\VV\WW$ for the union (instead of $\VV\cup \WW$).

Expressions (i.e., formulas that depend on some classical variables
$\XX$) $e$ are always assumed to have finitely many variables $\fv e$.
If $m$ is an assignment of values to classical variables, we write
\symbolindexmark\denotee{$\denotee em$} for $e$ evaluated on $m$. We
write \symbolindexmark\typee{$\typee e$} for the type of $e$, i.e., the set of all possible
values of $e$.

An important concept in the formalization of qRHL are indexed
variables, i.e., for every variable $\vv$ there are two distinct
variables $\vv_1,\vv_2$.  In \cite{qrhl}, there are explicit
operations $\operatorname{idx}_1$, $\operatorname{idx}_2$ that replace
all variables by indexed variables in a list/set of variables or in an
expression. We use a more compact notation and simply index the
list/set/expression. I.e., if $\VV$ is a list/set of variables,
$\VV_1$ refers to $\VV$ with every variable $\vv$ replaced by
$\vv_1$. And $e_1$ is the expression $e$ with every $\vv$ substituted
by $\vv_1$. (In \cite{qrhl} this would be $\operatorname{idx}_1\VV$,
$\operatorname{idx}_1e$.)  Similarly, given a quantum predicate $\PA$
(defined later in \autoref{sec:qrhl}), $\PA_1$ and $\PA_2$ are quantum
predicates with all variables $\vv$ replaced by $\vv_1,\vv_2$,
respectively.

Let \symbolindexmark\VVall{$\VVall$} be the set of all variables (not
including indexed variables).

We make some assumptions about the set $\VVall$ of all
variables. (Those assumptions were not made in \cite{qrhl}.)  Namely,
for any variable $\vv\in\VVall$, there exist infinitely many
$\ww\in\VVall$ that are compatible with $\vv$. (I.e.  $\vv$ and $\ww$
are either both quantum or both classical, and $\typev\vv=\typev\ww$.)
Furthermore, we assume that there is at least one quantum variable
$\qq$ with $\abs{\typev\qq}=\aleph_0$. (Note, we only assume that
those variables exist, not that they are actually used in any given
program.)

\paragraph{Linear algebra.}
We write \symbolindexmark\elltwo{$\elltwo X$} for the Hilbert space with basis
$\{\basis{}x\}_{x\in X}$.  For a set of quantum variables $\QQ$, we
write \symbolindexmark\elltwov{$\elltwov\QQ$} for $\elltwo{\types\QQ}$, i.e., the space of all
states those quantum variables can take.

Given a vector $\psi$, we define
\symbolindexmarkonly\proj$\symbolindexmarkhighlight{\proj\psi}:=\psi\adj\psi$.
Given a bounded operator $A$, we define
\symbolindexmarkonly\toE$\symbolindexmarkhighlight{\toE A}(\rho) :=
A\rho\adj A$.

A \index{cq-operator}\emph{cq-operator} is a positive trace-class
operator over a set $\VV$ of variables of the form
$\sum_m \proj{\basis{\cl\VV}m} \otimes \rho_m$ for positive trace-class operators
$\rho_m$ over $\qu\VV$. I.e., a cq-operator is basically a density operator
that is classical in the classical variables of $\qu\VV$ (except that we do not
require that the trace is $=1$ or $\leq1$).

A \index{superoperator}\emph{superoperator} is a completely positive
map $\calE$ that maps trace-class operators to trace-class operators
such that $\exists B.\forall\rho.\tr\calE(\rho) \leq B\tr\rho$.

Subspaces always mean topologically closed subspaces. For a subspace
$A$, let \symbolindexmark\orth{$\orth A$} be the orthogonal complement.

CPTPM means completely positive trace preserving map, while CPTRM
means completely positive trace reducing map (i.e., for positive
input, the trace of the output is smaller-equal the trace of the
input).

For disjoint $\RR,\SSS\subseteq \QQ$, let \symbolindexmark\swap{$\swap\RR\SSS$} be the
unitary operator on $\elltwov\QQ$ that swap the subsystems $\RR$ and
$\SSS$.

Let \symbolindexmark\suppo{$\suppo A$} denote the support of an operator $A$. (Formally, the
image of the smallest projector $P$ such that $PAP=A$.)

For $\RR\subseteq \QQ$ and a trace-class operator over $\elltwov\QQ$, let
\symbolindexmarkonly\partr$\symbolindexmarkhighlight{\partr{}{\RR}}\rho$
denote the partial trace of $\rho$ that traces out $\RR$. That is, 
$\partr{}{\RR}\rho$ is a trace-class operator over $\elltwov{\QQ\setminus\RR}$.
Sometimes, we annotate $\partr{}{\RR}$ with the set of remaining variables,
i.e., \symbolindexmark\partr{$\partr{\SSS}{\RR}$} if $\SSS=\QQ\setminus\RR$.
If $\QQ$ consists of indexed variables,
we write $\tr{}1$ short for $\tr{}{\QQ^1}$ where $\QQ^1\subseteq\QQ$ consists only of the 1-indexed variables.
Analogously $\tr{}2$.

For a bounded operator $A$, let \symbolindexmark\adj{$\adj A$} denote
the adjoint of $A$. (I.e., the conjugate transpose, often also written
$A^\dagger$.)

\cite{qrhl} also explicitly writes the canonical isomorphisms
$U_{\mathit{vars},\QQ}$ between different isomorphic spaces related to the
variables $\QQ$. (Namely $\elltwov{\QQ}$ and $\elltwo{\types\QQ}$.) We
omit those isomorphisms in our notation. In particular, if $\psi\in\elltwov{\QQ}$,
and $\PA\subseteq\elltwov{\QQ_1}$,
then the expression $\psi\in\PA$ is well-typed and understood to mean
$\adj{U_{\mathit{vars},\QQ_1}}U_{\mathit{vars},\QQ}\psi\in\elltwov{\QQ_2}$.

\paragraph{Distributions.} Probability distributions are always
discrete distributions (i.e., the $\sigma$-algebra of all measurable
spaces is the powerset). A subprobability distribution is like a
probability distribution except that the total probability may be
$\leq1$.  For a (sub)probability distribution $\mu$ over $X$, let
\symbolindexmarkonly\suppd$\symbolindexmarkhighlight\suppd\mu\subseteq X$ be the support of $X$, i.e., the set of values
with nonzero probability.  For a (sub)probability distribution $\mu$
over $X\times Y$, let
\symbolindexmarkonly\marginal$\symbolindexmarkhighlight{\marginal1\mu},
\symbolindexmarkhighlight{\marginal2\mu}$ be the first/second
marginal (i.e., (sub)probability distributions over $X$ and $Y$,
respectively).

\section{Language of programs}

\subsection{Syntax}
\label{sec:syn.prog}

We recap the syntax from \cite{qrhl}, and add one more statement to
it, for declaring local variables. Everything else is unchanged.

We will typically denote programs with
\symbolindexmark\bc{$\bc$} or \symbolindexmarkhighlight{$\bd$}.

Quantum variables are written \symbolindexmark\qq{$\qq,\rr$},
classical variables \symbolindexmark\xx{$\xx,\yy$}, an arbitrary
variables \symbolindexmark\vv{$\vv,\ww$}. Sets/lists of variables are
$\QQ,\RR,\SSS$ or $\XX,\YY$ or $\VV,\WW$.

\symbolindexmarkonly{\Qinit}%
\symbolindexmarkonly{\Qapply}%
\symbolindexmarkonly{\Qmeasure}%
\symbolindexmarkonly{\assign}%
\symbolindexmarkonly{\sample}%
\symbolindexmarkonly{\Skip}%
\symbolindexmarkonly{\while}%
\symbolindexmarkonly{\local}%
\begin{align*}
\bc,\bd := {} &\Skip && \text{(no operation)} \\
& \assign \XX e && \text{(classical assignment)} \\
& \sample \XX e && \text{(classical sampling)}\\
& \langif{e}{\bc}{\bd} && \text{(conditional)}\\
& \while{e}{\bc} && \text{(loop)}\\
& \seq{\bc}{\bd} && \text{(sequential composition)} \\
& \Qinit\QQ{e} && \text{(initialization of quantum registers)}\\
& \Qapply{e}\QQ && \text{(quantum application)} \\
& \Qmeasure{\XX}\QQ{e}
                     && \text{(measurement)} \\
& \local \vv \bc && \text{(local variables)}                    
\end{align*}

In the sampling statement, $e$
evaluates to a distribution.  In the initialization of quantum
registers, $e$
evaluates to a pure quantum state, $\qq_1\dots\qq_n$
are jointly initialized to that state. In the quantum application, $e$
evaluates to an isometry that is applied to
$\qq_1\dots\qq_n$.
In the measurement, $e$
evaluates to a projective measurement, the outcome is
stored in $\xx$.
(Recall that an expression $e$
  can be an arbitrarily complex mathematical formula in the classical
  variables. So, e.g., an expression that describes an isometry could
  be something as simple as just $H$
  (here $H$ denotes the Hadamard transform),
  or something more complex such as, e.g., $H^{\xx}$, meaning $H$ is applied if $\xx=1$.)

  The new statement in this syntax (relative to \cite{qrhl}) is
  $\local\vv\bc$. Intuitively, this means that $\vv$ is a local
  variable in $\bc$. More specifically (but still informally), at the
  beginning of $\local\vv\bc$, the current state of $\vv$ is stored
  (think of a stack), $\vv$ is initialized with a default value, $\bc$
  is executed, and the original state of $\vv$ is restored.

  Note that $\local\vv\bc$ binds weaker than $\bc;\bd$. I.e.,
  $\local\vv{\bc;\bc}$ means $\local\vv{\paren{\bc;\bd}}$, not $\paren{\local\vv\bd};\bc$.
  
A program is \emph{well-typed}%
\index{well-typed (program)}%
\index{program!well-typed} according to the following rules:
\begin{compactitem}
\item $\assign\XX e$ is well-typed iff $\typee e\subseteq\typel\XX$,
\item $\sample \XX e$ is well-typed iff $\typee e$ is a subset of
the subprobability distributions on $\typel\XX$.
\item $\langif{e}{\bc}{\bd}$ is well-typed iff $\typee e\subseteq\{\true,\false\}$ and $\bc,\bd$ are well-typed.
\item $\while{e}{\bc}$ is well-typed iff $\typee e\subseteq\{\true,\false\}$ and $\bc$ is well-typed.
\item $\seq{\bc}{\bd}$ is well-typed iff $\bc$ and $\bd$ are well-typed.
\item $\Qinit{Q}{e}$
  is well-typed iff
  $\typee e\subseteq\elltwo{\typel Q}$, and
   $\norm\psi = 1$
  for all $\psi\in\typee e$. 
\item $\Qapply{e}\QQ$ is well-typed iff $\typee e$ is a subset of
  the set of isometries on $\elltwo{\typel\QQ}$.
\item $\Qmeasure{\XX}\QQ{e}$
  is well-typed iff
  $\typee e$
  is a subset of the set of all projective measurements
  on $\typel\QQ$ with outcomes in $\typel\XX$.
\item $\local\vv\bc$ is well-typed iff $\bc$ is well-typed.
\end{compactitem}
In this paper, we will only consider well-typed programs. That is,
``program'' implicitly means ``well-typed program'', and all
derivation rules hold under the implicit assumption that the programs
in premises and conclusions are well-typed.

We also consider contexts in this work. A context follows the above
grammar, with the additional symbol \symbolindexmark\hole{$\hole i$} where $i$ is a natural
number. A context $C$ can be instantiated as $C[\bc_1,\dots,\bc_n]$,
This means that every occurrence of $\hole i$ is replaced by $\bc_i$.
(With no special treatment of local-variables. E.g., if $C=\local\vv{\hole 1}$,
then $C[\bc]=\local\vv\bc$ even if $\bc$ contains $\vv$.)

\subsection{Semantics of programs}
\label{sec:sem.progs}

First, we recap the semantics of the language as defined in \cite{qrhl}.

Given a program $\bc$
(with $\fv\bc\subseteq\VV$), we define its semantics
\symbolindexmark{\denotc}{$\denotc\bc$}
as a cq-superoperator
that maps trace-class
cq-operators over $\VV$ onto trace-class cq-operators over $\VV$.
In the following, let
$\rho$ be a trace-class cq-operator over $\VV$,
$m\in\types{\cl \VV}$ (i.e., an assignment of values to
classical variables), and
$\rho_m$ a positive trace-class operator over $\qu \VV$. Note that specifying $\denotc\bc$ on
operators of the form $\pointstate{\cl\VV}m\tensor\rho_m$ specifies $\denotc\bc$
on all $\rho$, since $\rho$ can be written as an
infinite sum of $\pointstate{\cl\VV}m\tensor\rho_m$.

Then the semantics of the language were defined as follows in \cite{qrhl}:

\symbolindexmarkonly\langif
\symbolindexmarkonly\while
\symbolindexmarkonly\seq
\symbolindexmarkonly\assign
\symbolindexmarkonly\sample
\symbolindexmarkonly\Qinit
\symbolindexmarkonly\Qapply
\symbolindexmarkonly\Qmeasure
  \begin{align*}
  \denotc\Skip(\rho) &:= \rho \\
  \denotc{\assign\xx e}\pb\paren{\pointstate{\cl\VV}m\tensor\rho_m} &:=
                                                    \pointstate{\cl\VV}{\upd m\xx{ \denotee{e}{m}}}\tensor \rho_m \\
  \denotc{\sample\xx e}\pb\paren{\pointstate{\cl\VV}m\tensor\rho_m} &:=
                                                          \sum_{z\in\typev\xx}{\denotee{e}m}(z) \cdot \pb\pointstate{\cl\VV}{\upd m\xx z}\tensor \rho_m
  \\
    \denotc{\langif{e}{\bc}{\bd}}(\rho) &:=
    \denotc{\bc}( \restricte {e}(\rho) ) +
    \denotc{\bd}( \restricte {\lnot e}(\rho) )
                                           \\
  \denotc{\while e\bc}(\rho)
  &:=\sum_{i=0}^\infty
  \restricte {\lnot e}\bigl((\denotc\bc\circ\restricte {e})^i(\rho)\bigr)
     \\
    \denotc{\seq{\bc_1}{\bc_2}}&:=\denotc{\bc_2}\circ\denotc{\bc_1}
                                 \\
  \denotc{\Qinit Q{e}}\pb\paren{\pointstate{\cl\VV}m\tensor\rho_m}
                     &:= \pb\pointstate{\cl\VV}m \tensor
                       \partr{}{Q}{\rho_m}\otimes\pb\proj{
                       \denotee em}
  \\
  \denotc{\Qapply{e}Q}\pb\paren{\pointstate{\cl\VV}m\tensor\rho_m}
                     &:= \pb\pointstate{\cl\VV}m \tensor \denotee{e}m \rho_m\adj{\paren{\denotee{e}m}}\\
  \denotc{\Qmeasure{\xx}{Q}e}\pb\paren{\pointstate{\cl\VV}m\tensor\rho_m}
                     &:=
                       \sum_{z\in\typev\xx} \!\!\!   \pb
                       \pointstate{\cl\VV}{\upd m\xx z} \tensor
                       (\denotee{e}m(z))\rho_m(\denotee{e}m(z))
\end{align*}

Here \symbolindexmarkonly\restricte$\symbolindexmarkhighlight{\restricte e}(\rho)$ is the cq-density operator $\rho$ restricted to the
parts where the expression $e$ holds. Formally, $\restricte e$ is the cq-superoperator on $V$ such that
\[\restricte e(\pointstate{\cl{V}}{m} \tensor \rho_m) :=
  \begin{cases}
    \pointstate{\cl{V}}m \tensor \rho_m & (\denotee{e}m=\true) \\
    0 & (\text{otherwise})
  \end{cases}
\]

\paragraph{Local variables.} It remains to give semantics to
statements of the form $\local\vv\bc$ as these did not occur in
\cite{qrhl}.

For every variable $\vv$,
we assume a fixed element \symbolindexmarkonly\initial$\symbolindexmarkhighlight{\initial\vv}\in\typev\vv$ (the \emph{default value}%
\index{default value}).
Let \symbolindexmarkonly\rhoinit$\symbolindexmarkhighlight{\rhoinit\vv}:=\pb\proj{\basis{\vv}{\initial\vv}}$.

In the following definition, for any variable $\vv$, let $\vv'$ denote
another (so far unused) variable of the same type, with the same
default value, and $\vv'$ is quantum/classical iff $\vv$ is. Then, for
any superoperator $\calE$,

\begin{equation}
  \label{eq:semantics.local}
  \pb\LOCAL\vv\calE(\rho)
  :=
  \partr{}\vv\
  \toE{\swapcop{\vv}}\circ
  {\pb\paren{\calE\otimes\idv{\vv'}}
      \circ \toE{\swapcop{\vv}}
  {\paren{\rho\otimes\rhoinit\vv}}}
\end{equation}

\begin{equation*}
\begin{tikzpicture}
  \initializeCircuit;
  \newWires{rest,v,v2};
  \stepForward{5mm};
  \labelWire[\tiny $\VVall\setminus\vv$]{rest};
  \labelWire[\tiny $\vv$]{v};
  \stepForward{11.5mm};
  \node[wireInput={v2}] (init-v2) {\footnotesize $\psiinit\vv$};
  \stepForward{2mm};
  \labelWire[\tiny $\vv'$]{v2};
  \stepForward{2mm};
  \drawWires{v,v2};
  \stepForward{2mm};
  \crossWire{v}{v2};
  \crossWire{v2}{v};
  \skipWires{v,v2};
  \stepForward{2mm};
  \node[gate={rest,v}] (c) {$\calE$};
  \stepForward{2mm};
  \drawWires{v,v2};
  \stepForward{2mm};
  \crossWire{v}{v2};
  \crossWire{v2}{v};
  \skipWires{v,v2};
  \stepForward{5mm};
  \node[killWire=v2] (kill-v2) {};
  \drawWire{v2}; \node[boxAroundLabeled={\tiny ${\LOCAL\vv\calE}$},
  fit=(c)(init-v2)(kill-v2)(\getWireCoord{v2})] (dotted) {};
  \stepForward{5mm};
  \drawWires{rest,v};
\end{tikzpicture}
\end{equation*}

Or equivalently:
\begin{equation*}
  {\LOCAL\vv\bc} :=
  \calF
  \otimes
  \id_{\vv}
  \qquad
  \text{where}
  \qquad
  \calF(\rho) := \partr{}\vv \calE(\rho \otimes \rhoinit\vv)
  \quad
  \text{for all trace-class operators $\rho$ over $\VVall\setminus\vv$}
\end{equation*}

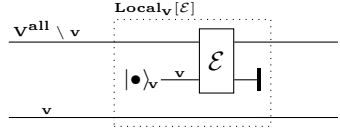
\[
\begin{tikzpicture}
  \initializeCircuit;
  \newWires{rest,v,v2}
  \stepForward{5mm};
  \labelWire[\tiny $\VVall\setminus\vv$]{rest};
  \labelWire[\tiny $\vv$]{v2};
  \stepForward{15mm};
  \skipWire{v};
  \node[wireInput={v}] (init-v) {\footnotesize $\psiinit\vv$};
  \stepForward{2.5mm};
  \labelWire[\tiny $\vv$]{v};
  \stepForward{2.5mm};
  \node[gate={rest,v}] (c) {$\calE$};
%  \node[gate={rest,v}] (c) {\includegraphics[width=10mm]{/tmp/kr.png}};
  %
  \stepForward{3mm};
  \node[killWire=v] (kill-v) {};
  \drawWire{v2}; \node[boxAroundLabeled={\tiny $\LOCAL\vv\calE$},
                       fit=(c)(init-v)(kill-v)(\getWireCoord{v2})] (dotted) {};
  \stepForward{10mm};
  \drawWires{rest,v2};
\end{tikzpicture}
\]

And then we can define $\denotc{\local\vv\bc}=\pb\LOCAL\vv{\denotc\bc}$.

We write
\symbolindexmarkonly\deneq$\bc\symbolindexmarkhighlight\deneq\bd$ to
denote denotational equivalence, i.e., $\denotc\bc=\denotc\bd$.

\medskip

Given the semantics, we can define the probability that a certain
condition holds after execution of a program, using the following
definition from \cite{qrhl}:

\begin{definition}\label{def:prafter}
  Fix a program $\bc$, an expression $e$ with
  $\typee e=\{\true,\false\}$, and some trace-class cq-operator $\rho$
  over $\VVall$.  Then
  \symbolindexmarkonly\prafter$\symbolindexmarkhighlight{\prafter
    e\bc\rho}:=\sum_{m\text{ s.t.\ }\denotee em=\true}\tr\rho_m$ where
  $\denotc\bc(\rho)=:\sum_m\proj{\basis{}m}\otimes\rho_m$ for
  trace-class operators $\rho_m$ over $\qu{\paren\VVall}$.
\end{definition}

\subsection{Variable sets}

Given a context (or program) $C$, we define a number of sets of
variables such as the set of free variables. These will be used
throughout the paper in various rules, most crucially in the
\rulerefx{Adversary} rule. Those sets are:
\begin{compactitem}
\item \symbolindexmark\fv{$\fv C$}: All free variables in $C$.
\item \symbolindexmark\inner{$\inner C$}: All variables $\vv$ such that $C$ contains a hole
  under a $\locala\vv$. (Those are the
  variables that will be shadowed if we substitute a program into a
  hole of $C$.)
\item \symbolindexmark\covered{$\covered C$}: All variables $\vv$ such every hole is under a
  $\locala\vv$. (Those are the variables which, if a program that is
  substituted into a hole of $C$ contains them, will still not be
  visible outside $C$.)
\item \symbolindexmark\overwr{$\overwr C$}: All variables that are overwritten in $C$.  I.e.,
  written before they are used for the first time. (Thus the content
  of those variables before execution of $C$ does not matter.)
\item \symbolindexmark\written{$\written C$}: All variables that are written (i.e., classical
  variables on the lhs of an assignment or sampling, and all free
  quantum variables).
\end{compactitem}

The precise recursive definitions follow.  All those variables sets
are also formally defined in Isabelle/HOL in the theory
\texttt{Basic\_Definitions}.

\begin{alignat*}{2}
  & \fv{\hole i}           &&:= \varnothing \\
  & \fv{\assign\XX e}      &&:= \XX \cup \fv e \\
  & \fv{\sample\XX e}      &&:= \XX \cup \fv e \\
  & \fv{\local\vv C}       &&:= \fv C \setminus \{\vv\} \\
  & \fv{\Qinit\QQ e}       && := \QQ \cup \fv e \\
  & \fv{\Qapply\QQ e}      && := \QQ \cup \fv e \\
  & \fv{\Qmeasure\XX\QQ e} && := \QQ \cup \XX \cup \fv e \\
  & \fv{C;C'}              && := \fv C \cup \fv{C'} \\
  & \fv{\langif e{C}{C'}}  && := \fv e \cup \fv C \cup \fv{C'} \\
  & \fv{\while e{C}}       && := \fv e \cup \fv C \\
  & \fv\Skip               && := \varnothing
\end{alignat*}

\begin{alignat*}3
  & \inner{\hole i}         && := \varnothing \\
  & \inner C                && := \varnothing           && \ (\text{if $C$ is a program}) \\
  & \inner{\local\vv C}     && := \inner C \cup \{\vv\} && \ (\text{if $C$ is not a program}) \\
  & \inner{\langif e C{C'}} && := \inner C \cup \inner {C'} \\
  & \inner{\while e C}      && := \inner C \\
  & \inner{C;C'}            && := \inner C \cup \inner{C'} \\
\end{alignat*}

\begin{alignat*}3
  & \covered{\hole i}        && := \varnothing \\
  & \covered{C;C'}           && := \covered C \cap \covered{C'} \\
  & \covered{\langif eC{C'}} && := \covered C \cap \covered{C'} \\
  & \covered{\while eC}      && := \covered C \\
  & \covered{\local\vv C}    && := \covered C \cup \{\vv\} \\
  & \covered C               && := \VVall && (\text{if $C$ is a program})
\end{alignat*}

\begin{alignat*}2
  & \overwr{\hole i}           &&  := \varnothing \\
  & \overwr{\assign\XX e}      &&  := \XX \setminus \fv e \\
  & \overwr{\sample\XX e}      &&  := \XX \setminus \fv e \\
  & \overwr{\Qinit\QQ e}       &&  := \QQ \\
  & \overwr{\Qapply\QQ e}      &&  := \varnothing \\
  & \overwr{\Qmeasure\XX\QQ e} &&  := \XX \setminus \fv e \\
  & \overwr{\langif eC{C'}}    &&  := \paren{\overwr C \cap \overwr {C'}} \setminus \fv e \\
  & \overwr{\while eC}         &&  := \varnothing \\
  & \overwr{\local\vv C}       &&  := \overwr C \setminus \{\vv\} \\
  & \overwr{C;C'}              &&  := \overwr C \cup \pB\paren{
                                    \pb\paren{\overwr{C'} \setminus \fv C} \cap \covered{C}}  \\
  & \overwr\Skip               &&  := \varnothing
\end{alignat*}

\begin{alignat*}2
  & \written{\hole i}           &&  := \varnothing \\
  & \written{\assign\XX e}      &&  := \XX \\
  & \written{\sample\XX e}      &&  := \XX \\
  & \written{\local\vv C}       &&  := \written C \setminus \{\vv\} \\
  & \written{\Qinit\QQ e}       &&  := \QQ \\
  & \written{\Qapply\QQ e}      &&  := \QQ \\
  & \written{\Qmeasure\XX\QQ e} &&  := \XX \cup \QQ \\
  & \written{\langif eC{C'}}    &&  := \written C \cup \written{C'} \\
  & \written{\while eC}         &&  := \written C \\
  & \written{\Skip}             &&  := \varnothing \\
  & \written{C;C'}              &&  := \written C \cup \written {C'}
\end{alignat*}

\subsection{Substitutions}

A \index{variable substitution}%
\index{substitution!variable}%
\emph{variable substitution} $\sigma$ is a function from variables to
variables such that $\vv$ and $\sigma(\vv)$ are compatible, i.e.
$\vv$ and $\sigma(\vv)$ are either both quantum or both classical,
and $\typev\vv=\typev{\sigma(\vv)}$.

Given a variable substitution $\vv$ and a program/context $\bc$,
\symbolindexmark\subst{$\subst\bc\sigma$} denotes the result of replacing every non-local variable $\vv$ in $\bc$ by $\sigma(\vv)$.
In contrast, \symbolindexmark\fullsubst{$\fullsubst \bc\sigma$} replaces every variable $\vv$ by $\sigma(\vv)$.
(E.g., if $\sigma(\vv)=\ww$, then $(\assign\vv 1;\local\vv{\assign\vv 1})\sigma=
(\assign\ww 1;\local\vv{\assign\vv 1})$
but
 $\fullsubst{(\assign\vv 1;\local\vv{\assign\vv 1})}\sigma=
(\assign\ww 1;\local\ww{\assign\ww 1})$.)

Renaming variables using a substitution may lead to conflicts with
existing local variables. The following inductive predicate
\symbolindexmark\noconflict{$\noconflict\cdot\cdot$} ensures that this
does not happen.

\begin{ruleblock}
  \RULEY{
    \noconflict\sigma{\bc}\\
    \noconflict\sigma{\bd}
  }{
    \noconflict\sigma{\langif e\bc\bd}
  }
  \RULEY{
    \noconflict\sigma{\bc}\\
    \noconflict\sigma{\bd}
  }{
    \noconflict\sigma{\bc;\bd}
  }
  \RULEY{
    \noconflict\sigma{\bc}
  }{
    \noconflict\sigma{\while e\bc}
  }
  \RULEY{~}{\noconflict\sigma{\assign\XX e}}
  \RULEY{~}{\noconflict\sigma{\sample\XX e}}
  \RULEY{~}{\noconflict\sigma{\Qinit\QQ e}}
  \RULEY{~}{\noconflict\sigma{\Qapply\QQ e}}
  \RULEY{~}{\noconflict\sigma{\Qmeasure\XX\QQ e}}
  \RULEY{~}{\noconflict\sigma\Skip}
  \RULEY{
    \noconflict{\sigma(\vv:=\vv)}\bc
    \\
    \vv \notin \sigma(\fv\bc \cap \dom \sigma )
  }{
    \noconflict\sigma{\local\vv\bc}
  }
\end{ruleblock}

Here $\dom\sigma:=\{\vv:\sigma(\vv)\neq\vv\}$.

In the Isabelle theories, the substitution $\subst\bc\sigma$ is formalized
as \verb|Basic_Definitions.subst_vars|,
the substitution $\fullsubst\bc\sigma$ as \verb|Basic_Definitions.full_subst_vars|,
and $\noconflict\sigma\bc$ as
\verb|Basic_Definitions.no_conflict|.

\section{Quantum relational Hoare logic}
\label{sec:qrhl}

In this section, we recap the relevant definitions of qRHL from
\cite{qrhl}. We slightly rewrite the definitions to make them
compatible with our notational conventions.

A \emph{quantum predicate}%
\index{quantum predicate}%
\index{predicate!quantum} $\PA$ over variables $\VV$ is, formally, an
expression with variables in $\cl\VV$ that evaluates to a subspace of
$\elltwov{\qu\VV}$.  Intuitively, a memory (with classical and quantum
variables) statisfies $\PA$ iff the quantum part of the memory lies in
$\PA$, when we instantiate the variables of $\PA$ with the classical
variables of the memory.  For pre-/postconditions in qRHL we use
quantum predicates over $\VVall_1\VVall_2$.  If such a memory is
represented as a density operator $\rho$, we say ``$\rho$
\emph{satisfies} $\PA$''%
\index{satisfy} if this holds. A formal definition is given in
\qrhlautoref{def:satisfy}.  Following \cite{ghosts}, we only consider
quantum predicates that depend on a finite number of variables, then
\symbolindexmark\fv{$\fv\PA$}, the set of free classical and quantum
variables of $\PA$, is well-defined (see \cite{ghosts} for details).

For detailed discussion of quantum predicates, we refer to
\cite{qrhl}. Here we only recall the most important constructions of
quantum predicates:

Intersection $\cap$ of quantum predicates is the analogue of
conjuction $\land$ of classical predicates.  Sum $+$ of spaces is the
analogue to disjunction $\vee$.  $\PA\subseteq\PB$ intuitively means
that $\PA$ implies $\PB$. (Note that $\PA\subseteq\PB$ is not a
quantum predicate, just a mathematical proposition.)

Given a classical predicate $P$ (i.e., a Boolean formula depending
only on classical variables), we can construct a quantum predicate
\symbolindexmark\CL{$\CL P$}. $\CL P$ is defined to be the whole space
is $P$ is true, and to be the $0$-space if $P$ is false. This way, a
state $\rho$ satisfies $\CL P$ if the classical variables of $\rho$
satisfy $P$.

Furthermore, \cite{qrhl} introduces the notation \symbolindexmark\lift{$\lift S\QQ$} to
denote the predicate that encodes the fact that $\QQ$ has a value in
$\QQ$ ($S$ must be a subspace of $\elltwov\QQ$). $\lift A\QQ$ can also
be used for operators $A$ to emphasize that $A$ operates on
$\elltwov\QQ$. And \symbolindexmark\spaceat{$\spaceat\PA\psi$} is an operation specific to the
\rulerefx{QInit1} rule, we omit the definition here. See \cite{qrhl}
for details. We write \symbolindexmark\quanteq{$\QQ\eqstate\psi$} to mean
$\lift{\SPAN{\{\psi\}}}\QQ$, i.e., the quantum predicate that says
that $\QQ$ is in state $\psi$.

\paragraph{Quantum equality.} One very important quantum predicate
(that can be combined with other predicates, e.g., using $\cap$ and
$+$) is the quantum equality. If $\QQ,\RR$ are disjoint lists of
quantum variables, then $\QQ\quanteq\RR$ intuitively means that $\QQ$
and $\RR$ have the same content. Formally,
\symbolindexmark{\quanteq}{$\QQ\quanteq\RR$} is the space of all
vectors that are invariant under $\swap\QQ\RR$, i.e., the unitary that
swaps registers $\QQ$ and $\RR$ in a quantum state. Intuitively, this
makes sense: two variables have the same content if exchanging them
does not change the overall state of the system. Though not formally
required, $\QQ$ will always contain $1$-indexed variables, and $\RR$
will contain $2$-indexed variables (or vice versa). That way, we can
use a quantum equality in a pre-/postcondition in a qRHL judgment to
state that the quantum variables of two programs are ``equal''.  For
example $\qq_1\rr_1\quanteq \qq_2\rr_2$ means that $\qq,\rr$ jointly
have the same content in the left and right memory.

There is an extended form of the quantum equality,
\symbolindexmark\quanteq{$U\QQ\quanteq V\RR$} where $U,V$ are
unitaries (or more generally, bounded operators, but then the
intuitive meaning of the quantum equality gets lots). Intuitively,
$U\QQ\quanteq V\RR$ means that the variables in $\QQ$, when we apply
$U$, have the same content as the variables in $\RR$, when we apply
$V$. For example
$\id\, \qq_1\rr_1 \quanteq \mathsf{CNOT} \, \qq_2\rr_2$ means that
$\qq\rr$ on the left is what you get from $\qq\rr$ on the right after
a CNOT. We refer to \qrhlautoref{def:quanteq} for the formal definition.

An important fact about the quantum equality is that
$\QQ\QQ'\quanteq\RR\RR'$ is not equivalent to
$\QQ\quanteq\RR \cap \QQ'\quanteq\RR'$, we merely have
$\QQ\quanteq\RR \cap \QQ'\quanteq\RR' \subseteq \QQ\quanteq\RR \cap
\QQ'\quanteq\RR'$.  This makes it harder to work with the quantum
equality than the classical equality. For useful laws about the
quantum equality, see \cite{qrhl}.

\paragraph{qRHL judgments.}
In qRHL, we want to express that given a precondition $\PA$ (on a pair
of memories, i.e., a quantum predicate on $\VVall_1\VVall_2$), when
executing the programs $\bc,\bd$, the postcondition $\PB$ holds, in
short $\rhl\PA\bc\bd\PB$. However, this simplified description is
somewhat misleading. We do not simply execute $\bc,\bd$ in parallel on
an initial state consisting of two memories satisfying $\PA$ and look
whether the final state satisfies $\PB$. The reason is that if we did
that, even simple fact such as
$\rhl{\CL\true}{\sample\xx{\mathcal U}}{\sample\xx{\mathcal
    U}}{\CL{\xx_1=\xx_2}}$ would not hold (where $\mathcal U$ is the uniform
distribution on $\{0,1\}$). This is because executing the left and
right program in parallel would only with probability 1/2 result in
the same bit $\xx$. This phenomenon already occurred in the classical
setting (pRHL, \cite{certicrypt}). Thus we use a more complex
definition that ``synchronizes'' probabilistic choices between the
left and right program. For a detailed justification of the definition of qRHL
see \cite{qrhl}. We simply state it here:

\begin{definition}[Quantum relational Hoare judgments]\label{def:rhl}%
  Let $\bc,\bd$ be programs.  Let $A,B$ be quantum predicates over
  $\VVall_1\VVall_2$.
  
  Then \symbolindexmark\rhl{$\rhl A{\bc}{\bd}B$} holds iff for all
  separable $\rho$ that satisfy $A$, we have that there exists a
  separable $\rho'$ that satisfies $B$ such that
  $\partr{}2\rho' = \denotc{\bc}\pb\paren{\partr{}2\rho}$ and
  $\partr{}1\rho' = \denotc{\bd}\pb\paren{\partr{}1\rho}$.
\end{definition}

\subsection{Rules of qRHL}

Most rules proven in \cite{qrhl} still hold (with the same proof) in
our setting (even though the definition of the language has
changed). This is because the proof of these rules are
``semantic''. By this, we mean that, if a program $\bc$ is
all-quantified in a rule, the proof makes no assumptions about the
code of $\bc$, and instead only refers to its semantics
$\denotc\bc$. Thus the exactly same proofs work when more statements
are added to the language. (But not if the definition of existing
statements is changed.) A notable exception is the \texttt{Adversary}
rule from \cite{qrhl} which does not apply any more since it is proven
by induction of the structure of programs.

In Figures \ref{fig:rules.general}, \ref{fig:rules.stmts},
\ref{fig:rules.quantum}, we state the rules from \cite{qrhl} that
still hold, using our more compact notation (in particular, we omit
the types of the various variables and expressions, and we omit
explicitly stated canonical isomorphisms between various
spaces).\footnote{We also use \autoref{lemma:fv} to justify replacing
assumptions of the form ``$\bc$ is $\VV$-local'' by
``$\fv\bc\subseteq\VV$''.} We omit the rather lengthy rules
\texttt{Trans} and \texttt{JointMeasure} that also still hold for
brevity, see \cite{qrhl}.

\begin{figure}[tp]
  \begin{ruleblock}
    \RULE{Sym}{
      \pb\rhl{\subst \PA\sigma}\bd\bc
      {\subst \PB\sigma}\\
      \forall\vv. \sigma(\vv_1):=\vv_2, \sigma(\vv_2):=\vv_1
    }{
      \rhl \PA\bc\bd \PB
    }
    \RULE{Conseq}{\PA\subseteq \PA'\\\PB'\subseteq \PB
      \\\rhl{\PA'}\bc\bd {\PB'}}{\rhl{\PA}\bc\bd \PB}
    \RULE{Seq}{
      \rhl{\PA}{\bc_1}{\bc_2}{\PB}
      \\
      \rhl{\PB}{\bd_1}{\bd_2}{\PC}
    }{
      \rhl{\PA}{\seq{\bc_1}{\bd_1}}{\seq{\bc_2}{\bd_2}}{\PC}
    }
    \RULE{Case}{
      \forall z.\
      \pb\rhl{\CL{e=z}\cap \PA}\bc\bd{\PB}
    }{
      \rhl{\PA}\bc\bd{\PB}
    }
    \RULE{Equal}{
      \fv\bc\subseteq\XX\QQ
    }{
      \pb\rhl
      {\CL{\XX_1=\XX_2}\cap (\QQ_1\quanteq \QQ_2)}
      \bc\bc
      {\CL{\XX_1=\XX_2}\cap (\QQ_1\quanteq \QQ_2)}
    }
    \RULE{Frame}{
      \fv\PR\subseteq  \VV_1\VV'_2\\
      \text{$\fv\bc\cap \VV$ and $\fv\bd\cap \VV'$ are classical}\\
      \text{$\bc$ is $(\fv\bc\cap \VV)$-readonly}
      \\
      \text{$\bd$ is $(\fv\bd\cap \VV')$-readonly}\\
      \rhl{\PA}\bc\bd \PB
    }{
      \rhl{\PA\cap \PR}\bc\bd{\PB\cap \PR} }
    \RULE{QrhlElim}{
      \rho\text{ is separable}\\
      \rho\text{ satisfies }\PA \\
      \rho_1:=\partr{V_1}{V_2}\rho \\
      \rho_2:=\partr{V_2}{V_1}\rho \\
      \pb\rhl \PA\bc\bd {\CL{e_1 \Rightarrow f_2}}
    }{
      \pb\prafter e\bc{\rho_1}
      \leq
      \pb\prafter f\bd{\rho_2}
      \rlap{\hskip1.5cm\parbox{4cm}{\footnotesize (also holds for
          $=,\Leftrightarrow$
          and\\\strut\,
          $\geq,\Leftarrow$
          instead of $\leq,\Rightarrow$)}}
    }
    \RULE{QrhlElimEq}{
      \rho\text{ satisfies }A\\
      \fv\bc,\fv\bd \subseteq \XX\QQ \\
      \qu{\fv\PA}\subseteq \QQ \\
      \pb\rhl {\CL{\XX_1=\XX_2}\cap (\QQ_1\quanteq \QQ_2)\cap \PA_1\cap \PA_2}{\bc}{\bd} {\CL{e_1 \Rightarrow f_2}}
    }{
      \pb\prafter e\bc{\rho}
      \leq
      \prafter f\bd{\rho}
      \newcommand\alsoholds{\footnotesize (also holds for
      $=,\Leftrightarrow$
      and\\\strut\,
      $\geq,\Leftarrow$
      instead of $\leq,\Rightarrow$)}
    \rlap{\hskip23mm\parbox{4cm}{\alsoholds}}%
  }
  \RULE{TransSimple}{
    \XX_{p}:=\cl{\fv p},\
    \QQ_{p}:=\qu{\fv p}\text{ for }p=\bc,\bd,\be
    \\
    \rhl{\CL{\XX_{c1}=\XX_{d2}}\cap (\QQ_{c1}\quanteq \QQ_{d2})}\bc\bd
    {\CL{\XX_{c1}=\XX_{d2}}\cap (\QQ_{c1}\quanteq \QQ_{d2})}
    \\
    \rhl{\CL{\XX_{d1}=\XX_{e2}}\cap (\QQ_{d1}\quanteq \QQ_{e2})}\bd\be
    {\CL{\XX_{d1}=\XX_{e2}}\cap (\QQ_{d1}\quanteq \QQ_{e2})}
  }{
    \rhl{\CL{\XX_{c1}=\XX_{e2}}\cap (\QQ_{c1}\quanteq \QQ_{e2})}\bc\be
    {\CL{\XX_{c1}=\XX_{e2}}\cap (\QQ_{c1}\quanteq \QQ_{e2})}
  }
\end{ruleblock}
\caption{Rules for qRHL (general rules).  
  }
  \label{fig:rules.general}
\end{figure}

\begin{figure*}[t]
  \begin{ruleblock}
    \RULE{Skip}{ }{\rhl{\PA}\Skip\Skip{\PA}}
    \RULE{Assign1}{ }{
      \pb\rhl{
        \substi \PB{e_1/\xx_1}
      }{\assign{\xx} e}{\Skip}{\PB}
    }
    \RULE{Sample1}{
      A:=\pB\paren{\CL{e_1\text{ is total}} \cap 
        \bigcap\nolimits_{z\in\suppd e_1}\substi B{z/\xx_1}}
    }{
      \rhl{
        A
      }{
        \sample{\xx} e
      }{\Skip} B
    }
    \RULE{JointSample}{
      A
      :=
      \pB\paren{\pb\CL{\marginal1f = e_1
          \land
          \marginal2f = e'_2}
        \cap
        \bigcap_{(z,z')\in\suppd f}
        \substi \PB{z/\xx_1,z'/\yy_2}}
    }{
      \rhl{\PA}
      {\sample{\xx}{e}}
      {\sample{\yy}{e'}}
      {\PB}
    }
    \RULE{If1}{
      \pb\rhl{\CL{e_1}\cap \PA}\bc\Skip \PB
      \\
      \pb\rhl{\CL{\lnot e_1}\cap \PA}\bd\Skip \PB
    }{
      \rhl{\PA}{\langif e\bc\bd}{\Skip}\PB
    }
    \RULE{JointIf}{
      \PA \subseteq \CL{e_1=e'_2}\\
      \pb\rhl{\CL{e_1\land e'_2}\cap \PA}{\bc}{\bc'}\PB
      \\
      \pb\rhl{\CL{\lnot e_1\land\lnot e'_2}\cap \PA}{\bd}{\bd'}\PB\\
    }{
      \rhl{\PA}{\langif{e}{\bc}{\bd}}{\langif{e'}{\bc'}{\bd'}}\PB
    }
    \RULE{While1}{
      \pb\rhl{\CL{e_1}\cap \PA}\bc\Skip \PA
      \\
      \PA\subseteq \PB_1
      \\
      (\while{e}{\bc})\text{ is total on }\PB
    }{
      \pb \rhl{\PA}{\while e\bc}\Skip {\CL{\lnot e_1}\cap \PA}
    }
    % `
    \RULE{JointWhile}{
      \PA\subseteq \CL{e_1=e'_2}\\
      \rhl{\CL{e_1\land e'_2}\cap A}\bc\bd \PA
    }{
      \rhl{\PA}{\while{e}\bc}{\while{e'}{\bd}}{\CL{\lnot e_1\land\lnot e'_2}\cap \PA}
    }
  \end{ruleblock}
  \caption{Rules for qRHL (related to individual classical statements).  
    For the rules \rulerefx{Assign1}, \rulerefx{Sample1},
    \rulerefx{If1}, and \rulerefx{While1}, there is also an analogous
    symmetric rule that we do not list explicitly.}
  \label{fig:rules.stmts}
\end{figure*}

\begin{figure*}[t]
  \begin{ruleblock}
    \RULE{QInit1}{}{
      \pb\rhl{\spaceat \PA{\lift{e_1}{\QQ_1}}}
      {\Qinit{\QQ}{e}}\Skip \PA
    }
    \RULE{QApply1}{}{
      \rhl{\adj{(\lift{e_1}{\QQ_1})}\cdot(\PB\cap \im (\lift{e_1}{\QQ_1}))}{\Qapply{e}{\QQ}}\Skip{\PB}
    }
    \RULE{Measure1}{
      \PA:=\Bigl(\CL{e_1\text{ is a total measurement}}
      \cap
      \bigcap\nolimits_{z} \bigl((\substi B{z/\xx_1}\cap \im (\lift{e_1(z)}{\QQ_1})) + \orth{ (\lift{e_1(z)}{\QQ_1})} \bigr)\Bigr)
    }{
      \rhl{A}{\Qmeasure{\xx}{Q}{e}}\Skip B
    }
    \RULE{JointMeasureSimple}{
      f_{1z} := \lift{e_1(z)}{\QQ_1}\\
      f'_{2z} := \lift{e'_2(z)}{\QQ_2'}\\ 
      \\
      \PA := \CL{e_1=e'_2}\cap(\QQ_1\quanteq \QQ_2')\cap \bigcap_{z}
      (\substi B{z/\xx_1,z/\yy_2}
      \cap \im f_{1z} \cap \im f'_{2z})
      + \orth{(\im f_{1z})} + \orth{(\im f'_{2z})}
    }{
      \rhl{\PA}
      {\Qmeasure{\xx}{\QQ}{e}}{\Qmeasure{\yy}{\QQ'}{e'}}{\PB}
    }
    % 
    % \RULE{JointMeasure}{
    %   \typee f\subseteq\powerset{\typev\xx\times\typev\yy}\\
    %   \typee{u_1}\subseteq\iso{\typel{Q_1},Z}\\
    %   \typee{u_2}\subseteq\iso{\typel{Q_2},Z}\\
    %   Q_1' := \idx1Q_1\\
    %   Q_2' := \idx2Q_2\\
    %   e'_{1x} := \lift{\idx1e_1(x)}{Q_1'}\\ %      \Uvarnames{Q_1'}\pb\paren{\idx1e_1(x)}\adj{\Uvarnames{Q_1'}} \tensor \idv{\qu{V_1}\qu{V_2}\setminus Q_1'}\\
    %   e'_{2y} := \lift{\idx2e_2(y)}{Q_2'}\\ % \Uvarnames{Q_2'}\pb\paren{\idx2e_2(y)}\adj{\Uvarnames{Q_2'}} \tensor \idv{\qu{V_1}\qu{V_2}\setminus Q_2'}\\
    %   % u'_1 := u_1\adj{\Uvarnames{Q_1'}}\\
    %   % u'_2 := u_2\adj{\Uvarnames{Q_2'}}\\
    %   C_f := \pB\CL{
    %     \pb\paren{\forall x.\idx1e_1(x)\neq0\implies \pb\abs{\{y:(x,y)\in f\}}=1}
    %     \land
    %     \pb\paren{\forall y.\idx2e_2(y)\neq0\implies \pb\abs{\{x:(x,y)\in f\}}=1}
    %   }
    %   \\
    %   C_e := \pB\CL{\forall(x,y)\in f.\ u_1\pb\paren{\idx1e_1(x)}\adj{u_1}=u_2\pb\paren{\idx2e_2(y)}\adj{u_2}}
    %   \\
    %   A := \bigcap_{(x,y)\in f}
    %   (\substi B{x/\xx_1,y/\yy_2}
    %   \cap \im e'_{1x} \cap \im e'_{2y})
    %   + \orth{(\im e'_{1x})} + \orth{(\im e'_{2y})}
    % }{
    %   \pb\rhl{C_f\cap C_e\cap A\cap(u_1Q_1'\quanteq u_2Q_2')}
    %   {\Qmeasure{\xx}{Q_1}{e_1}}{\Qmeasure{\yy}{Q_2}{e_2}}{B}
    % }
    % 
  \end{ruleblock}
  \caption{Rules for qRHL (related to individual quantum statements).
    For the rules \rulerefx{Measure1}, \rulerefx{QApply1},
    and \rulerefx{QInit1}, there is also an analogous
    symmetric rule that we do not list explicitly.}
  \label{fig:rules.quantum}
\end{figure*}

\section{Semantics-related lemmas}
\label{sec:sem.lemmas}

\begin{lemma}\label{lemma:fv}
  If $\fv\bc\subseteq \VV$, then there is a $\calE$ on
  $\VV$ such that $\denotc\bc=\calE\otimes\id_{\VVall\setminus\VV}$.
\end{lemma}

(This lemma was already stated in \cite{qrhl} but the proof was not ``semantic''.)

\begin{proof}
  We show this by induction on $\bc$.
  Each case is elementary to check, we only show the case $\bc=\local\vv{\bc'}$ here:

  Since $\fv\bc=\fv{\bc'}\setminus\{\vv\}$, we have $\fv{\bc'}\subseteq \VV\cup\{\vv\}$.
  By induction hypothesis, there is an $\calE'$ on $\VV\cup\{\vv\}$ with $\denotc{\bc'}=\calE'\otimes\id_{\VVall\setminus\VV\setminus\{\vv\}}$.

  We have (expressing the various superoperators as circuits for readability):
  \[
    \denotc\bc
    \quad\starrel=\quad
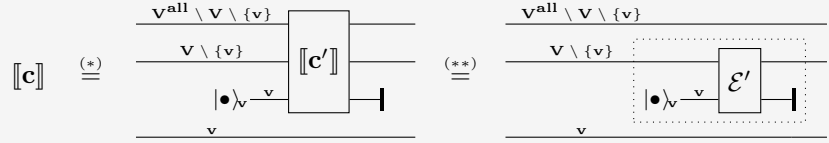
\begin{tikzpicture}[baseline={(0,-.78)}]
  \initializeCircuit;
  \newWires{rest,V,v,v2}
  \stepForward{10mm};
  \labelWire[\tiny $\VVall\setminus\VV\setminus\{\vv\}$]{rest};
  \labelWire[\tiny $\VV\setminus\{\vv\}$]{V};
  \labelWire[\tiny $\vv$]{v2};
  \stepForward{5mm};
  \skipWire{v};
  \node[wireInput={v}] (init-v) {\footnotesize $\psiinit\vv$};
  \stepForward{2.5mm};
  \labelWire[\tiny $\vv$]{v};
  \stepForward{2.5mm};
  \node[gate={rest,V,v}] (c) {$\denotc{\bc'}$};
%  \node[gate={rest,v}] (c) {\includegraphics[width=10mm]{/tmp/kr.png}};
  %
  \stepForward{4mm};
  \node[killWire=v] (kill-v) {};
  \stepForward{4mm};
  \drawWires{rest,V,v2};
\end{tikzpicture}
\quad\starstarrel=\quad
\begin{tikzpicture}[baseline={(0,-.78)}]
  \initializeCircuit;
  \newWires{rest,V,v,v2}
  \stepForward{10mm};
  \labelWire[\tiny $\VVall\setminus\VV\setminus\{\vv\}$]{rest};
  \labelWire[\tiny $\VV\setminus\{\vv\}$]{V};
  \labelWire[\tiny $\vv$]{v2};
  \stepForward{13mm};
  \skipWire{v};
  \node[wireInput={v}] (init-v) {\footnotesize $\psiinit\vv$};
  \stepForward{2.5mm};
  \labelWire[\tiny $\vv$]{v};
  \stepForward{2.5mm};
  \node[gate={V,v}] (E) {$\calE'$};
  \stepForward{4mm};
  \node[killWire=v] (kill-v) {};
  \drawWire{v2}; \node[boxAround, fit=(E)(init-v)(kill-v)] (dotted) {};
  \stepForward{4mm};
  \drawWires{rest,V,v2};
\end{tikzpicture}
\]
Here $(*)$ is by the semantics of the language, and $(**)$ since
$\denotc{\bc'}=\calE'\otimes\id_{\VVall\setminus\VV\setminus\{\vv\}}$.
If $\vv\notin\VV$, let $\calE$ be the dotted box in the rhs. If
$\vv\in\VV$, let $\calE$ be the dotted box together with the $\vv$-wire.
Then $\calE$ is a superoperator on $\VV$, and $\denotc\bc=\calE\otimes\id_{\VVall\setminus\VV}$.
\end{proof}

\begin{lemma}\label{lemma:swap}
  If $\fv\bc\cap\fv\bd=\varnothing$,
  then ${\bc;\bd} \deneq {\bd;\bc}$.
\end{lemma}

\begin{proof}
  Let $\VV:=\fv\bc$ and $\WW:=\VVall\setminus\fv\bd$.  By
  \autoref{lemma:fv}, there exists $\calE_{\bc}$ on $\VV$ such that
  $\denotc\bc=\calE_{\bc}\otimes\id_{\WW}$. And there exists
  $\calE_{\bd}$ on $\WW$ such that
  $\denotc\bd=\id_{\VV}\otimes\calE_{\bd}$.
  Thus
  \[
    \denotc{\bc;\bd} =
    \denotc\bd \circ \denotc\bc =
    (\id_{\VV} \otimes \calE_{\bd}) \circ     (\calE_{\bc} \otimes \id_{\WW}) =
    \calE_{\bc} \otimes \calE_{\bd} =
    (\calE_{\bc} \otimes \id_{\WW}) \circ (\id_{\VV} \otimes \calE_{\bd}) =
    \denotc\bc \circ \denotc\bd =
    \denotc{\bd;\bc}.
  \]
\end{proof}

\begin{lemma}\lemmalabel{hgdfaysdgfyasdgfasdfh}
  \begin{compactenum}[(i)]
  \item\itlabel{lemma:local.idem} ${\pb\local \vv{\local \vv\bc}} \deneq {\local \vv\bc}$.
    (And $\pb\LOCAL\vv{\LOCAL\vv\calE}=\LOCAL\vv\calE$.)
  \item\itlabel{lemma:local.swap} ${\pb\local \vv{\local {\ww}\bc}} \deneq {\pb\local {\ww}{\local {\vv}\bc}}$.
    (And $\pb\LOCAL\vv{\LOCAL\ww\calE}=\pb\LOCAL\ww{\LOCAL\vv\calE}$.)
  \end{compactenum}
\end{lemma}

\begin{proof}
  In this proof, we show the claims in parentheses, involving
  $\LOCAL\vv\calE$ etc. The claims involving $\locala\vv$ etc.~are an immediate consequence.
  
  (We show only the claims in terms of $\local\vv\dots$. The claims in terms of $\LOCAL\vv\dots$ are shown analogously.)
  
  By definition of $\LOCAL\vv\dots,\LOCAL\ww\dots$, the lhs and rhs of \eqref{lemma:local.idem} are
  described by the following circuits:
  \[
    \newcommand\bl{-.78}
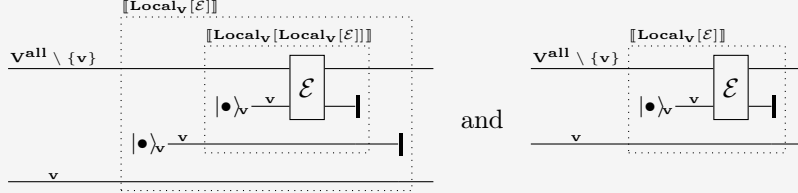
  \begin{tikzpicture}[baseline={(0,\bl)}]
  \initializeCircuit;
  \newWires{rest,v0,v,v2}
  \stepForward{6mm};
  \labelWire[\tiny $\VVall\setminus\{\vv\}$]{rest};
  \labelWire[\tiny $\vv$]{v2};
  \stepForward{15mm};
  \skipWire{v};
  \node[wireInput={v}] (init-v) {\footnotesize $\psiinit\vv$};
  \stepForward{2mm};
  \labelWire[\tiny $\vv$]{v};
  \stepForward{9mm};
  \skipWire{v0};
  \node[wireInput={v0}] (init-v0) {\footnotesize $\psiinit\vv$};
  \stepForward{2.5mm};
  \labelWire[\tiny $\vv$]{v0};
  \stepForward{2.5mm};
  \node[gate={rest,v0}] (c) {$\calE$};
  \stepForward{4mm};
  \node[killWire=v0] (kill-v0) {};
  \drawWire{v}; \node[fit=(c)(init-v0)(kill-v0)(\getWireCoord{v}),
                      boxAroundLabeled=\tiny $\denotc{\LOCAL\vv{\LOCAL\vv\calE}}$] (inner-dotted) {};
  \stepForward{5mm};
  \node[killWire=v] (kill-v) {};
  \drawWire{v2}; \node[fit=(c)(init-v)(kill-v)(\getWireCoord{v2})(inner-dotted-label)(inner-dotted),
                       boxAroundLabeled={\tiny $\denotc{\LOCAL\vv\calE}$}] (dotted) {};
  \stepForward{4mm};
  \drawWires{rest,v2};
\end{tikzpicture}
\quad\text{and}\quad
\begin{tikzpicture}[baseline={(0,\bl)}]
  \initializeCircuit;
  \newWires{rest,v,v2}
  \stepForward{6mm};
  \labelWire[\tiny $\VVall\setminus\{\vv\}$]{rest};
  \labelWire[\tiny $\vv$]{v2};
  \stepForward{13mm};
  \skipWire{v};
  \node[wireInput={v}] (init-v) {\footnotesize $\psiinit\vv$};
  \stepForward{2.5mm};
  \labelWire[\tiny $\vv$]{v};
  \stepForward{2.5mm};
  \node[gate={rest,v}] (c) {$\calE$};
  \stepForward{3mm};
  \node[killWire=v] (kill-v) {};
  \drawWire{v2}; \node[boxAroundLabeled={\tiny $\denotc{\LOCAL\vv\calE}$},
                       fit=(c)(init-v)(kill-v)(\getWireCoord{v2})] (dotted) {};
  \stepForward{4mm};
  \drawWires{rest,v2};
\end{tikzpicture}
\]
The only difference is the third wire in the lhs which is initialized
with $\psiinit\vv$ and then discarded again, is the same as the
identity. Thus the lhs and rhs are equal, \eqref{lemma:local.idem}
follows.

\medskip

By definition of the semantics of the language, the lhs and rhs of \eqref{lemma:local.swap} are
described by the following circuits:
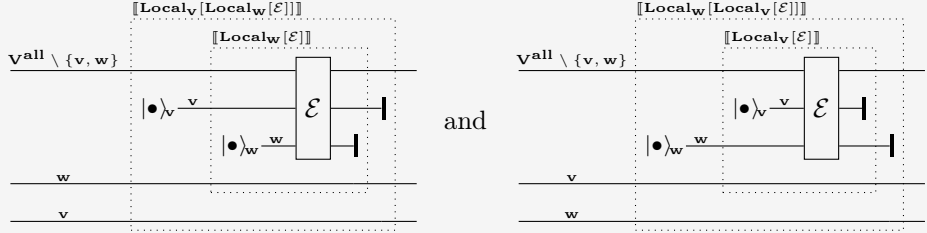
\[
  \def\bl{-.78}
  \begin{tikzpicture}[baseline={(0,\bl)}]
    \initializeCircuit;
    \newWires{rest,v,w,w2,v2}
    \stepForward{7mm};
    \labelWire[\tiny $\VVall\setminus\{\vv,\ww\}$]{rest};
    \labelWire[\tiny $\vv$]{v2};
    \labelWire[\tiny $\ww$]{w2};
    \stepForward{15mm};
    \skipWire{v};
    \node[wireInput={v}] (init-v) {\footnotesize $\psiinit\vv$};
    \stepForward{2mm};
    \labelWire[\tiny $\vv$]{v};
    \stepForward{9mm};
    \skipWire{w};
    \node[wireInput={w}] (init-w) {\footnotesize $\psiinit\ww$};
    \stepForward{2mm};
    \labelWire[\tiny $\ww$]{w};
    \stepForward{2.5mm};
    \node[gate={rest,v,w}] (c) {$\calE$};
    \stepForward{3mm};
    \node[killWire=w] (kill-w) {};
    \drawWire{w2}; \node[boxAroundLabeled={\tiny $\denotc{\LOCAL\ww\calE}$},
                         fit=(c)(init-w)(kill-w)(\getWireCoord{w2})] (dotted-w) {};
    \stepForward{3mm};
    \node[killWire=v] (kill-v) {};
    \drawWire{v2}; \node[boxAroundLabeled={\tiny $\denotc{\LOCAL\vv{\LOCAL\ww\calE}}$},
                         fit=(c)(init-v)(kill-v)(\getWireCoord{v2})(dotted-w-label)] (dotted-v) {};
    \stepForward{4mm};
    \drawWires{rest,v2,w2};
  \end{tikzpicture}
  \quad\text{and}\quad
  \begin{tikzpicture}[baseline={(0,\bl)}]
    \initializeCircuit;
    \newWires{rest,v,w,v2,w2}
    \stepForward{7mm};
    \labelWire[\tiny $\VVall\setminus\{\vv,\ww\}$]{rest};
    \labelWire[\tiny $\vv$]{v2};
    \labelWire[\tiny $\ww$]{w2};
    \stepForward{15mm};
    \skipWire{w};
    \node[wireInput={w}] (init-w) {\footnotesize $\psiinit\ww$};
    \stepForward{2mm};
    \labelWire[\tiny $\ww$]{w};
    \stepForward{9mm};
    \skipWire{v};
    \node[wireInput={v}] (init-v) {\footnotesize $\psiinit\vv$};
    \stepForward{2mm};
    \labelWire[\tiny $\vv$]{v};
    \stepForward{2.5mm};
    \node[gate={rest,v,w}] (c) {$\calE$};
    \stepForward{3mm};
    \node[killWire=v] (kill-v) {};
    \drawWire{v2}; \node[boxAroundLabeled={\tiny $\denotc{\LOCAL\vv\calE}$},
                         fit=(c)(init-v)(kill-v)(\getWireCoord{v2})] (dotted-v) {};
    \stepForward{3mm};
    \node[killWire=w] (kill-w) {};
    \drawWire{w2}; \node[boxAroundLabeled={\tiny $\denotc{\LOCAL\ww{\LOCAL\vv\calE}}$},
                         fit=(c)(init-w)(kill-w)(\getWireCoord{w2})(dotted-v-label)] (dotted-w) {};
    \stepForward{4mm};
    \drawWires{rest,v2,w2};
  \end{tikzpicture}
\]
The only difference is the order in which the last two wires are
drawn which has no semantic meaning.  Thus the lhs and rhs are
equal, \eqref{lemma:local.swap} follows.
\end{proof}

This lemma implies that the order in which variables are declared
local does not matter.  This motivates the following shorthand: For a
finite $\VV$
we introduce the following shorthand:
$\localp\VV\bc := \localp{\vv_1}{\dots; \ \local{\vv_n}\bc}$
where $\vv_1,\dots,\vv_n$
are the elements of $\VV$ in arbitrary order.

Similarly, we define $\LOCAL\VV\calE := \LOCAL{\vv_1}{\LOCAL{\vv_2}{\dots\LOCAL{\vv_n}\calE}}$.

As an immediate consequence of the definition and \lemmaref{lemma:local.swap}, we get
\begin{lemma}\lemmalabel{locals.simplefacts}
  \begin{compactenum}[(i)]
  \item $\local\varnothing\bc = \bc$. (Also $\LOCAL\varnothing\calE=\calE$.)
  \item\itlabel{lemma:local.merge} ${\pb\local \VV{\local {\VV'}\bc}} \deneq {\pb\local {\VV\cup\VV'}\bc}$. (Also $\LOCAL\VV{\LOCAL{\VV'}\calE} = \LOCAL{\VV\cup\VV'}\calE$.)
  \item\itlabel{lemma:local.LOCAL} $\denotc{\local\VV\bc} = \pb\LOCAL\VV{\denotc\bc}$.
  \end{compactenum}
\end{lemma}

\begin{lemma}\label{lemma:local.sum}
  Let $I$ be a set and $\calE_i$ ($i\in I$) be superoperators.
  Assume that $\sum_{i\in I}\calE_i$ converges.
  % Assume that there is a finite set $\VV$ such that
  % $\calE_i=\calE_i'\otimes\id$ for some $\calE_i$ on $\VV$.\TODOQ{check if this is OK where we use this lemma}
  Then
  $\sum_{i\in I}\LOCAL\VV{\calE_i} = \LOCAL\VV{\sum_{i\in I}\calE_i}$
  (and the lhs converges).\footnote{%
    Note that the converse does not hold: If
    $\sum_{i\in I}\LOCAL\VV{\calE_i}$ converges,
    $\sum_{i\in I}\calE_i$ does not necessarily converge.  For
    example, let $\calE_i(\rho):=\proj\psi \rho\proj\psi$ where $\psi$
    is a normalized vector orthogonal to $\psiinit\VV$. Then
    $\LOCAL\VV{\calE_i} = 0$ and thus
    $\sum_{i\in I}\LOCAL\VV{\calE_i}$ converges trivially.  But
    $\sum_{i\in I}\calE_i(\proj\psi)= \sum_{i\in I}\proj\psi$ diverges
    (assuming $I$ is infinite), and thus $\sum_{i\in I}\calE_i$ does
    not converge pointwise.  } Here convergence is pointwise
  convergence with respect to the trace-norm.
\end{lemma}

\begin{proof}
  In this proof, unless mentioned otherwise, convergence of trace-class operators is with respect to trace-norm,
  and convergence of superoperators is pointwise  with respect to trace-norm.
  Whenever we write an equality, we mean that equality holds whenever the sums in lhs and rhs converge, and
  that the lhs converges if the rhs does.

  \begin{claim}\label{claim:sum.id}
    $\sum_i (\calE_i\otimes\id) = (\sum_i\calE_i) \otimes \id$.
  \end{claim}

  \begin{claimproof}
    Let $\calL':=\sum_{i\in I}\calE_i$. By assumption, $\calL'$ exists
    and is trace bounded. Let $B$ such that
    $\tr\calL'(\rho)\leq B\tr\rho$ for all positive $\rho$.  For
    finite $F$ and positive $\rho$, we have
    $\sum_{i\in F}\tr(\calE_i\otimes\id)(\rho) =
    \sum_{i\in F}\calE_i(\partr{}2\rho)\leq
    \calL'(\partr{}2\rho)$. Thus the sum
    $\sum_{i\in F}\tr(\calE_i\otimes\id)(\rho)$ is bounded (as a function
    of finite $F$). Furthermore, since $\calE_i$ is completely
    positive, $(\calE_i\otimes\id)(\rho)$ is positive.  Thus
    $\sum_{i\in F}(\calE_i\otimes\id)(\rho)$ is bounded and
    increasing, hence it converges. Thus the limit
    $\calL''(\rho) := \sum_{i\in I}(\calE_i\otimes\id)(\rho)$ exists
    for positive $\rho$.  Since every trace class $\rho$ is a linear
    combination of four positive $\rho$, the limit also exists for
    arbitrary~$\rho$.

    We are left to show that $(\calL'\otimes\id)=\calL''$. Assume this is not the
    case. Since the set of all trace class operators is spanned by
    operators $\sigma\otimes\tau$ with unit trace, this implies that
    there are $\sigma,\tau$ with unit trace such that
    $(\calL'\otimes\id)(\sigma\otimes\tau) \neq \calL''(\sigma\otimes\tau)$. Let
    $\delta := \norm{ (\calL'\otimes\id)(\sigma\otimes\tau) -
      \calL''(\sigma\otimes\tau) }_{\tr}$.  Since $\calL'(\sigma)$ is
    the limit of $\sum_i\calE_i(\sigma)$, for sufficiently large
    finite $F$,
    \begin{equation}
      \label{eq:largeF1}
      \pB\norm{\sum_{i\in F}\calE_i(\sigma) - \calL'(\sigma)}_{\tr} \leq \delta/3.
    \end{equation}
    And since $\calL''(\sigma\otimes\tau)$ is the limit of
    $\sum_i(\calE_i\otimes\id)(\sigma\otimes\tau)$, for sufficiently
    large finite $F$,
    \begin{equation}
      \label{eq:largeF2}
      \pB\norm{\sum_{i\in F}(\calE_i\otimes\id)(\sigma\otimes\tau) - \calL''(\sigma\otimes\tau)}_{\tr} \leq \delta/3.
    \end{equation}
    Fix an $F$ such that both \eqref{eq:largeF1} and
    \eqref{eq:largeF2} hold.

    Furthermore,
    \begin{multline*}
      \pB\norm{\sum_{i\in F}\calE_i(\sigma) - \calL'(\sigma)}_{\tr}
      =
      \pB\norm{
        \pB\paren{\sum_{i\in F}\calE_i(\sigma) - \calL'(\sigma)}\otimes\tau}_{\tr}
      \\=
      \pB\norm{\sum_{i\in F}(\calE_i\otimes\id)(\sigma\otimes\tau) - (\calL'\otimes\id)(\sigma\otimes\tau)}_{\tr}.
    \end{multline*}
    With \eqref{eq:largeF1}, this implies
    \[
      \pB\norm{\sum_{i\in F}(\calE_i\otimes\id)(\sigma\otimes\tau) - (\calL'\otimes\id)(\sigma\otimes\tau)}_{\tr}
      \leq \delta/3.
    \]
    With \eqref{eq:largeF2} and the triangle inequality,
    we get $\norm{(\calL'\otimes\id)(\sigma\otimes\tau)
      - \calL''(\sigma\otimes\tau)} \leq 2\delta/3$,
    in contradiction to the definition of $\delta$.
    Thus  $(\calL'\otimes\id)=\calL''$.
  \end{claimproof}

  \begin{claim}\label{claim:local.sum1}
    $\sum_{i}\LOCAL\vv{\calE_i}
    = \LOCAL\vv{\sum_{i}\calE_i}$.
  \end{claim}

  (Note: the index of $\LOCAL\vv{\dots}$ is a single variable $\vv$, not $\VV$.)
  
  \begin{claimproof}
    We have:
  \begin{align*}
    \sum_{i}\LOCAL\vv{\calE_i}(\rho)
    &\eqrefrel{eq:semantics.local}=
      \sum_i
      \partr{}\vv
      \pB\paren{
      \toE{\swapcop{\vv}}
      \pB\paren{
      {\pb\paren{\calE_i\otimes\idv{\vv'}}
      \pb\paren{
      \toE{\swapcop{\vv}}
      {\paren{\rho\otimes\rhoinit\vv}}}}}}
      \\
    &\starrel=
      \partr{}\vv
      \pB\paren{
      \toE{\swapcop{\vv}}
      \pB\paren{
      \sum_i
      {\pb\paren{\calE_i\otimes\idv{\vv'}}
      \pb\paren{
      \toE{\swapcop{\vv}}
      {\paren{\rho\otimes\rhoinit\vv}}}}}}
      \\
    &\starstarrel=
      \partr{}\vv
      \pB\paren{
      \toE{\swapcop{\vv}}
      \pB\paren{
      {\pb\paren{      \sum_i \paren{\calE_i\otimes\idv{\vv'}}}
      \pb\paren{
      \toE{\swapcop{\vv}}
      {\paren{\rho\otimes\rhoinit\vv}}}}}}
      \\
      &\txtrel{Cl.~\ref{claim:sum.id}}=\,
      \partr{}\vv
      \pB\paren{
      \toE{\swapcop{\vv}}
      \pB\paren{
      {{      \pb\paren{\paren{\textstyle\sum_i \calE_i}\otimes\idv{\vv'}}}
      \pb\paren{
      \toE{\swapcop{\vv}}
      {\paren{\rho\otimes\rhoinit\vv}}}}}}
    \\
    &\eqrefrel{eq:semantics.local}=
      \pb\LOCAL\vv{\textstyle\sum_i\calE_i}
  \end{align*}
  Here each equality means that the lhs converges if the rhs converges.
  And $(*)$ follows because $\partr{}\vv$ and $\toE{\swapcop{\vv}}$ are trace-preserving.
  And $(**)$ follows by definition of pointwise convergence.
\end{claimproof}

The lemma follows by induction over $\VV$ with \autoref{claim:local.sum1}.
\end{proof}

For simpler notation, we write \symbolindexmark\init{$\init\vv$}
for $\Qinit\vv{\basis{}{\initial\vv}}$
or $\assign\vv{\initial\vv}$,
depending on whether $\vv$ is quantum or classical.
For a finite set $\VV$, let $\init\VV$ denote $\init{\vv_1};\dots;\init{\vv_n}$
where $\vv_1,\dots,\vv_n$ are the elements of $\VV$ in some arbitrary order.
(The order does not matter due to \autoref{lemma:swap}.)

\begin{lemma}\label{lemma:init.sem}
  For finite $\VV=\{\vv_1,\dots,\vv_n\}$, let
  \symbolindexmark\Einit{$\Einit\VV$} be the superoperator
  $\rho\mapsto \proj{\psiinit\VV}\otimes\tr\rho$, where
  $\psiinit\VV := \psiinit{\vv_1}\otimes\dots\otimes \psiinit{\vv_n}$.

  Then
  $\denotc{\init\VV} = \Einit\VV \otimes \id_{\VVall\setminus\VV}$.
\end{lemma}

\begin{proof}
  For $\VV=\{\qq\}$ or $\VV=\{\xx\}$, this follows from the definition
  of $\init\qq$ and $\init\xx$, as well as the semantics of assignment
  and quantum initialization. By definition of
  $\Einit\VV$,
  $\Einit\VV \otimes \id_{\VVall\setminus\VV}
  \otimes
  \Einit\WW \otimes \id_{\VVall\setminus\WW}
  =
  \Einit{\VV\WW} \otimes \id_{\VVall\setminus\VV\WW}.
  $
  The lemma then follows by induction.  
\end{proof}

\begin{lemma}\lemmalabel{dfsjdfjsdfhuiyers}
  \begin{compactenum}[(i)]
  \item\itlabel{lemma:unused} ${\local\VV\bc} \deneq \bc$ if $\VV\cap\fv\bc=\varnothing$.
    
  \item\itlabel{lemma:add.init.begin}
    ${\local\VV\bc}
    \deneq {\local\VV{(\init{\VV'};\bc)}}$
    if $\VV'\subseteq\VV$.
  \item\itlabel{lemma:add.init.end}
    ${\local\VV\bc}
    \deneq {\local\VV{(\bc;\init\VV')}}$
    if $\VV'\subseteq\VV$.
  \end{compactenum}
\end{lemma}

\begin{proof}
  We first show \eqref{lemma:unused}.  By \autoref{lemma:fv}, there is
  an $\calE$ on $\fv\bc$ such that
  $\denotc\bc=\calE\otimes\id_{\VVall\setminus\fv\bc}$.  Thus we can
  represent $\denotc\bc$ and $\denotc{\local\VV\bc}$ by the following
  circuits:
  
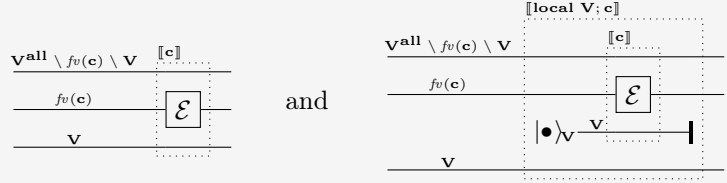
\[
    \begin{tikzpicture}[baseline={(0,-.5)}]
      \initializeCircuit;
      \newWires{rest,fvc,V};
      \stepForward{8mm}
      \labelWire[\tiny $\VVall\setminus\fv\bc\setminus\VV$]{rest};
      \labelWire[\tiny $\fv\bc$]{fvc};
      \labelWire[\tiny $\VV$]{V};
      \stepForward{12mm};
      \node[gate=fvc] (c) {$\calE$};
      \drawWires{rest,V};
      \node[boxAroundLabeled={\tiny$\denotc\bc$}, fit=(c)(\getWireCoord{rest})(\getWireCoord{V})] (c-box) {};
      \stepForward{4mm};
      \drawWires{rest,fvc,V};
    \end{tikzpicture}
    \qquad\text{and}\qquad
    \begin{tikzpicture}[baseline={(0,-.7)}]
      \initializeCircuit;
      \newWires{rest,fvc,V,V2};
      \stepForward{8mm}
      \labelWire[\tiny $\VVall\setminus\fv\bc\setminus\VV$]{rest};
      \labelWire[\tiny $\fv\bc$]{fvc};
      \labelWire[\tiny $\VV$]{V2};
      \stepForward{17mm};
      \node[wireInput=V] (init-V) {\small $\psiinit\VV$};
      \stepForward{2.5mm};
      \labelWire[\tiny $\VV$]{V};
      \stepForward{2.5mm};
      \node[gate=fvc] (c) {$\calE$};
      \drawWires{rest,V};
      \node[boxAroundLabeled={\tiny$\denotc\bc$}, fit=(c)(\getWireCoord{rest})(\getWireCoord{V})] (c-box) {};
      \stepForward{5mm};
      \node[killWire=V] (kill-V) {};
      \drawWires{rest,V2};
      \node[boxAroundLabeled={\tiny$\denotc{\local\VV\bc}$}, fit=(c-box-label)(init-V)(kill-V)(\getWireCoord{rest})(\getWireCoord{V2})] (c-box) {};
      \stepForward{4mm};
      \drawWires{rest,fvc,V2};
    \end{tikzpicture}
  \]
  The only difference is the third wire that is created and discarded on the rhs.
  This is equal to the identity, thus the two circuits are identical and we have
  $\bc\deneq\local\VV\bc$. This shows \eqref{lemma:unused}.

  \medskip

  We now show \eqref{lemma:add.init.begin}. We show the special case
  ${\local\vv\bc}
  \deneq {\local\vv{\paren{\init{\vv};\bc}}}$. The general case follows by induction.
  The lhs and rhs, as circuits are, respectively:

  \[
    \def\bl{-.8}
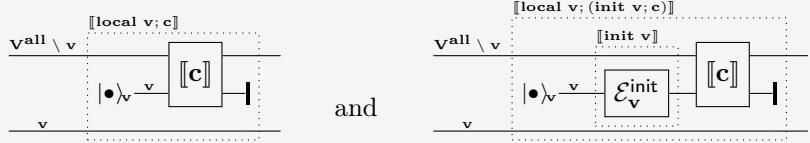
\begin{tikzpicture}[baseline={(0,\bl)}]
  \initializeCircuit;
  \newWires{rest,v,v2}
  \stepForward{4.5mm};
  \labelWire[\tiny $\VVall\setminus\vv$]{rest};
  \labelWire[\tiny $\vv$]{v2};
  \stepForward{12mm};
  \node[wireInput={v}] (init-v) {\footnotesize $\psiinit\vv$};
  \stepForward{2mm};
  \labelWire[\tiny $\vv$]{v};
  \stepForward{2.5mm};
  \node[gate={rest,v}] (c) {$\denotc\bc$};
  \stepForward{3mm};
  \node[killWire=v] (kill-v) {};
  \drawWire{v2}; \node[boxAroundLabeled={\tiny $\denotc{\local\vv\bc}$},
                       fit=(c)(init-v)(kill-v)(\getWireCoord{v2})] (dotted) {};
  \stepForward{4mm};
  \drawWires{rest,v2};
\end{tikzpicture}
\qquad\text{and}\qquad
\begin{tikzpicture}[baseline={(0,\bl)}]
  \initializeCircuit;
  \newWires{rest,v,v2}
  \stepForward{4.5mm};
  \labelWire[\tiny $\VVall\setminus\vv$]{rest};
  \labelWire[\tiny $\vv$]{v2};
  \stepForward{12mm};
  \node[wireInput={v}] (init-v) {\footnotesize $\psiinit\vv$};
  \stepForward{2mm};
  \labelWire[\tiny $\vv$]{v};
  \stepForward{4mm};
  \node[gate=v] (Einit-v) {$\Einit\vv$};
  \drawWire{rest};
  \node[boxAroundLabeled={\tiny $\denotc{\init\vv}$},
        fit=(Einit-v)(\getWireCoord{rest})] (init-box) {};
  \stepForward{3.5mm};
  \node[gate={rest,v}] (c) {$\denotc\bc$};
  \stepForward{3mm};
  \node[killWire=v] (kill-v) {};
  \drawWire{v2}; \node[boxAroundLabeled={\tiny $\denotc{\local\vv{\paren{\init\vv;\bc}}}$},
                       fit=(c)(init-v)(init-box-label)(kill-v)(\getWireCoord{v2})] (dotted) {};
  \stepForward{4mm};
  \drawWires{rest,v2};
\end{tikzpicture}
\]
Here we used
\autoref{lemma:init.sem} for expressing $\denotc{\init\vv}$ in terms of $\Einit\vv$.

It follows immediately from the definition of $\Einit\vv$ that
$\Einit\vv(\proj{\psiinit\vv})=\proj{\psiinit\vv}$. Thus the two
circuits compute the same function. Hence $\local\vv\bc\deneq\local\vv{\paren{\init\vv;\bc}}$.
\eqref{lemma:add.init.begin} follows.

\medskip

  We now show \eqref{lemma:add.init.end}.
  We show the special case
  ${\local\vv\bc}
  \deneq {\local\vv{\paren{\bc;\init\vv}}}$. The general case follows by induction.
  The lhs and rhs, as circuits are, respectively:

  \[
    \def\bl{-.8}
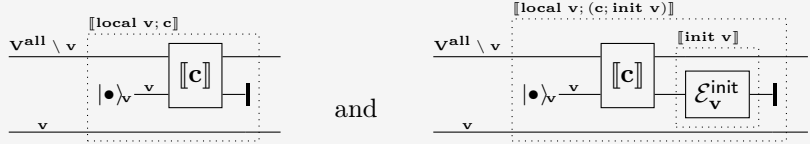
\begin{tikzpicture}[baseline={(0,\bl)}]
  \initializeCircuit;
  \newWires{rest,v,v2}
  \stepForward{4.5mm};
  \labelWire[\tiny $\VVall\setminus\vv$]{rest};
  \labelWire[\tiny $\vv$]{v2};
  \stepForward{12mm};
  \node[wireInput={v}] (init-v) {\footnotesize $\psiinit\vv$};
  \stepForward{2mm};
  \labelWire[\tiny $\vv$]{v};
  \stepForward{2.5mm};
  \node[gate={rest,v}] (c) {$\denotc\bc$};
  \stepForward{3mm};
  \node[killWire=v] (kill-v) {};
  \drawWire{v2}; \node[boxAroundLabeled={\tiny $\denotc{\local\vv\bc}$},
                       fit=(c)(init-v)(kill-v)(\getWireCoord{v2})] (dotted) {};
  \stepForward{4mm};
  \drawWires{rest,v2};
\end{tikzpicture}
\qquad\text{and}\qquad
\begin{tikzpicture}[baseline={(0,\bl)}]
  \initializeCircuit;
  \newWires{rest,v,v2}
  \stepForward{4.5mm};
  \labelWire[\tiny $\VVall\setminus\vv$]{rest};
  \labelWire[\tiny $\vv$]{v2};
  \stepForward{12mm};
  \node[wireInput={v}] (init-v) {\footnotesize $\psiinit\vv$};
  \stepForward{2mm};
  \labelWire[\tiny $\vv$]{v};
  \stepForward{3.5mm};
  \node[gate={rest,v}] (c) {$\denotc\bc$};
  \stepForward{4mm};
  \node[gate=v] (Einit-v) {$\Einit\vv$};
  \drawWire{rest};
  \node[boxAroundLabeled={\tiny $\denotc{\init\vv}$},
        fit=(Einit-v)(\getWireCoord{rest})] (init-box) {};
  \stepForward{3mm};
  \node[killWire=v] (kill-v) {};
  \drawWire{v2}; \node[boxAroundLabeled={\tiny $\denotc{\local\vv{\paren{\bc;\init\vv}}}$},
                       fit=(c)(init-v)(init-box-label)(kill-v)(\getWireCoord{v2})] (dotted) {};
  \stepForward{4mm};
  \drawWires{rest,v2};
\end{tikzpicture}
\]
Here we used
\autoref{lemma:init.sem} for expressing $\denotc{\init\vv}$ in terms of $\Einit\vv$.

  Since $\Einit\vv$ is trace-preserving, ${\tr}\circ{\Einit\vv}=\tr$. I.e.,
  {\footnotesize$
  \begin{tikzpicture}[baseline={(0,-.1)}]
    \initializeCircuit;
    \newWire{v};
    \stepForward{2.5mm};
    \node[gate=v] (E) {$\Einit\vv$};
    \stepForward{2.5mm};
    \node[killWire=v] (k) {};
  \end{tikzpicture}
  $}
and
  {\footnotesize$
  \begin{tikzpicture}[baseline={(0,-.1)}]
    \initializeCircuit;
    \newWire{v};
    \stepForward{2mm};
    \node[killWire=v] (k) {};
  \end{tikzpicture}
  $}
compute the same function. Thus the lhs and rhs compute the same function,
i.e., $\local\bc\deneq\local\vv{\paren{\bc;\init\vv}}$.
\eqref{lemma:add.init.end} follows.
\end{proof}

\begin{lemma}\label{lemma:full_subst_vars}
  Let $\sigma$ be a bijective variable substitution.  Assume
  $\dom\sigma\cap\fv\bc=\varnothing$.  Then
  $\fullsubst\bc\sigma \deneq \bc$.
\end{lemma}

\begin{proof}
  Let $U_\sigma$ be the unitary on $\elltwov\VVall$ defined by
  $U_\sigma\basis{}{m} = \basis{}{m\circ\sigma}$. That is
  $U_\sigma$ reorders the subsystems corresponding to the variables in
  $\VVall$. Let $\calE_\sigma(\rho) := U_\sigma\rho\adj{U_\sigma}$.
  
  Since $\fullsubst\bc\sigma$ simply renames all variables (even the
  local ones), $\denotc{\fullsubst\bc\sigma}$ simply operates on the
  reordered variables, formally
  $\denotc{\fullsubst\bc\sigma}(\rho)=\calE_{\sigma^{-1}} \circ
  \denotc\bc \circ \calE_{\sigma}$.

  Since $\sigma$ is the identity on $\dom\sigma$,
  $\calE_\sigma=\calE_\sigma'\otimes \id$ for some $\calE_{\sigma}'$
  on $\dom\sigma$.  And $\denotc\bc=\calE_{\bc}'\otimes\id$ for some
  $\calE_{\bc}'$ on $\fv\bc$. Since $\dom\sigma$ and $\fv\bc$ are
  disjoint, this implies that $ $$\calE_\sigma$ and
  $\denotc\bc$ commute.

  Thus 
  \begin{equation*}
    \denotc{\fullsubst\bc\sigma}(\rho)
    =\calE_{\sigma^{-1}} \circ \denotc\bc \circ \calE_{\sigma}
    = \calE_{\sigma^{-1}} \circ \calE_{\sigma} \circ \denotc\bc
    = \calE_{\sigma\circ \sigma^{-1}} \circ \denotc\bc
    = \calE_{\id}  \circ \denotc\bc
    = \denotc\bc.
    \mathQED
  \end{equation*}
\end{proof}

\begin{lemma}\label{lemma:rename.locals}
  Let $\sigma$ be a variable substitution that is injective
  on $\VV$ and let $\WW:=\sigma(\VV)$. Assume that $\sigma=\id$ outside $\VV$
  and that $(\fv\bc\setminus\VV)\cap\WW=\varnothing$.
  Assume $\noconflict\sigma\bc$.
  
  Then ${\local\VV\bc} \deneq {\local\WW{\paren{\bc\{\WW/\VV\}}}}$
\end{lemma}

This is shown in Isabelle/HOL, as \verb|rename_locals| in theory
\verb|Rename_Locals.thy|. See \autoref{sec:isabelle-proofs} for
remarks about our Isabelle/HOL development.

{\TODOQ{Proof sketch}}

\begin{lemma}\lemmalabel{dghfasdfgaeyfsdyfs}
  \begin{compactenum}[(i)]
  \item\itlabel{lemma:merge.local}
    % Isabelle axiom local_seq_merge0
    % General case (multiple local variables): locals_seq_merge
    $\localp\VV\bc;\localp\VV\bd
    \deneq
    \local\VV{\paren{\bc;\init\VV;\bd}}$.
  \item \itlabel{lemma:merge.local.fv}
      % General case in Isabelle: locals_seq2
      % One-var case: local_seq2
    If $\VV\cap\fv\bc=\varnothing$, then $\bc;\localp \VV\bd \deneq \local\VV{\paren{\bc;\bd}}$.
  \end{compactenum}
\end{lemma}

\begin{proof}
  We first show \eqref{lemma:merge.local} in the special case $\VV=\vv$. The lhs and rhs are
  depicted by the following circuits (using \autoref{lemma:init.sem} for the semantics of $\init\vv$):
  \begin{equation*}
    \begin{tikzpicture}
      \initializeCircuit;
      \newWires{rest,v,v2}
      \stepForward{5mm};
      \labelWire[\tiny $\VVall\setminus\vv$]{rest};
      \labelWire[\tiny $\vv$]{v2};
      \stepForward{12mm};
      %
      % [local v; c]
      %
      \skipWire{v};
      \node[wireInput={v}] (init-v) {\footnotesize $\psiinit\vv$};
      \stepForward{2.5mm};
      \labelWire[\tiny $\vv$]{v};
      \stepForward{2.5mm};
      \node[gate={rest,v}] (c) {$\denotc\bc$};
      % \node[gate={rest,v}] (c) {\includegraphics[width=10mm]{/tmp/kr.png}};
      % 
      \stepForward{3mm};
      \node[killWire=v] (kill-v) {};
      \drawWire{v2}; \node[boxAroundLabeled={\tiny $\denotc{\local\vv\bc}$},
      fit=(c)(init-v)(kill-v)(\getWireCoord{v2})] (dotted) {};
      \stepForward{10mm};
      \drawWires{rest,v2};
      % 
      % [local v; d]
      % 
      \skipWire{v};
      \node[wireInput={v}] (init-v) {\footnotesize $\psiinit\vv$};
      \stepForward{2.5mm};
      \labelWire[\tiny $\vv$]{v};
      \stepForward{2.5mm};
      \node[gate={rest,v}] (d) {$\denotc\bd$};
      \stepForward{3mm};
      \node[killWire=v] (kill-v) {};
      \drawWire{v2}; \node[boxAroundLabeled={\tiny $\denotc{\local\vv\bd}$},
      fit=(d)(init-v)(kill-v)(\getWireCoord{v2})] (dotted) {};
      \stepForward{5mm};
      \drawWires{rest,v2};
    \end{tikzpicture}
    \qquad
    \begin{tikzpicture}
      \initializeCircuit;
      \newWires{rest,v,v2}
      \stepForward{5mm};
      \labelWire[\tiny $\VVall\setminus\vv$]{rest};
      \labelWire[\tiny $\vv$]{v2};
      \stepForward{12mm};
      \skipWire{v};
      \node[wireInput={v}] (init-v) {\footnotesize $\psiinit\vv$};
      \stepForward{2.5mm};
      \labelWire[\tiny $\vv$]{v};
      \stepForward{2.5mm};
      \node[gate={rest,v}] (c) {$\denotc\bc$};
      % \node[gate={rest,v}] (c) {\includegraphics[width=10mm]{/tmp/kr.png}};
      %
      \stepForward{5mm};
      \node[gate={v}] (Einit) {$\Einit\vv$};
      \stepForward{5mm};
      \node[gate={rest,v}] (d) {$\denotc\bd$};
      \drawWire{v2}; \node[boxAroundLabeled={\tiny $\denotc{\bc;\init\vv;\bd}$},
      fit=(c)(d)(Einit)] (dotted) {};
      \stepForward{3mm};
      \node[killWire=v] (kill-v) {};
      \stepForward{5mm};
      \drawWires{rest,v2};
    \end{tikzpicture}
  \end{equation*}
  By definition (\autoref{lemma:init.sem}),
  $\Einit\vv:\rho\mapsto \proj{\psiinit\vv}\otimes\tr\rho$. Or, as a circuit, $
  \begin{tikzpicture}
    \newWire{v};
    \stepForward{2.5mm};
    \labelWire[\tiny $\vv$]{v};
    \stepForward{2.5mm};
    \node[killWire=v] (kill-v) {};
    \stepForward{7mm};
    \node[wireInput={v}] (init-v) {\footnotesize $\psiinit\vv$};
    \stepForward{3.5mm};
    \drawWire{v};
  \end{tikzpicture}
  $.
  Thus the two circuits are identical, hence \eqref{lemma:merge.local} follows.  

  \medskip

  We prove \eqref{lemma:merge.local.fv} in the special case $\VV=\vv$:
  \begin{multline*}
    \bc;\localp\vv\bd
    \starrel\deneq
    \localp\vv\bc;\localp\vv\bd
    \eqrefrel{lemma:merge.local}\deneq
    \local\vv{\paren{\bc;\init\vv;\bd}}
    \\
    \starstarrel\deneq
    \local\vv{\paren{\init\vv;\bc;\bd}}
    \tristarrel\deneq
    \local\vv{\paren{\bc;\bd}}
  \end{multline*}
  Here $(*)$ uses \lemmaref{lemma:unused},
  $(**)$ uses \lemmaref{lemma:swap},
  and $(*{*}*)$ uses \lemmaref{lemma:add.init.begin}.

  This shows \eqref{lemma:merge.local.fv}.

  \medskip

  The general case of \eqref{lemma:merge.local.fv} is a
  straightforward induction over $\VV$, using the special case for the
  induction step. We did the proof of the general case in Isabelle/HOL
  (\texttt{Helping\_Lemmas.locals\_seq2}), using the special of 
  \eqref{lemma:merge.local.fv} as an axiom.

  \medskip

  The general case of \eqref{lemma:merge.local} is proven by a simple
  induction over $\VV$, using the special case of
  \eqref{lemma:merge.local} and the general case of
  \eqref{lemma:merge.local} for the base case.
  We did the proof of the general case in Isabelle/HOL
  (\texttt{Helping\_Lemmas.locals\_seq\_merge}), using the special of 
  \eqref{lemma:merge.local} as an axiom.
\end{proof}

\begin{lemma}\label{lemma:init.overwr}
  \begin{compactenum}[(i)]
  \item\label{lemma:init.overwr:X}   If $\XX\subseteq\overwr\bc$, then ${\assign\XX e;\bc}\deneq \bc$.
  \item\label{lemma:init.overwr:Q}   If $\QQ\subseteq\overwr\bc$, then ${\Qinit\QQ e;\bc}\deneq \bc$.
  \item\label{lemma:init.overwr:init}   If $\VV\subseteq\overwr\bc$, then ${\init\VV;\bc}\deneq\bc$.
\end{compactenum}
\end{lemma}

\begin{proof}
  We first show a special case of \eqref{lemma:init.overwr:init},
  namely that ${\init\vv;\bc}\deneq\bc$ if $\vv\in\overwr\bc$.
  We show this by induction over the structure of $\bc$.

  We distinguish the following cases:
  \begin{itemize}
  \item Cases $\bc=\Skip$, $\bc=\while e{\bc'}$, $\bc=\Qapply e \QQ$:

    In these cases, $\overwr\bc=\varnothing$. By assumption,
    $\vv\in\overwr\bc$. Thus this case cannot arise.
  \item Case $\bc=\assign\XX e$:

    In this case, $\vv\in\overwr\bc=\XX\setminus\fv e$. Thus $\vv$ is
    a classical variable, $\vv\in\XX$, and $\vv\notin\fv e$.

    We have for all $m,\rho$:
    \begin{align*}
      &\denotc{\init\vv;\bc}\pb\paren{\proj{\basis{\cl{{\VVall}}}m}\otimes\rho} \\
      &= \denotc{\assign\XX e}\circ\denotc{\assign\vv{\initial\vv}}\pb\paren{\proj{\basis{\cl{{\VVall}}}m}\otimes\rho}
      && \text{(sem.~of ;, def.~of $\initkw$)} \\
      &= \denotc{\assign\XX e}\pb\paren{\proj{\basis{\cl{{\VVall}}}{m(\vv:={\initial\vv})}}\otimes\rho}
      && \text{(sem.~of assignment)}\\
      &= {\proj{\basis{\cl{{\VVall}}}{m(\vv:={\initial\vv})(\XX:=\denotee e{m(\vv:=\initial\vv)})}}\otimes\rho}
      && \text{(sem.~of assignment)}\\
      &\starrel= {\proj{\basis{\cl{{\VVall}}}{m(\vv:={\initial\vv})(\XX:=\denotee e{m})}}\otimes\rho} \\
      &\starstarrel= {\proj{\basis{\cl{{\VVall}}}{m(\XX:=\denotee e{m})}}\otimes\rho} \\
      &= \denotc{\assign\XX e}\pb\paren{\proj{\basis{\cl{{\VVall}}}m}\otimes\rho}
        = \denotc\bc\pb\paren{\proj{\basis{\cl{{\VVall}}}m}\otimes\rho}
      && \text{(sem.~ of assignment)}
    \end{align*}
    Here $(*)$ follows since $\vv\notin\fv e$ and thus
    $\denotee e{m(\vv:=\initial\vv)}=\denotee em$.  And $(**)$ follows
    since $\vv\in\XX$ and thus
    $m(\vv:={\initial\vv})(\XX:=\denotee e{m}) = m(\XX:=\denotee
    e{m})$.

    Thus $\denotc{\init\vv;\bc}=\denotc\bc$ on all states of the form
    ${\proj{\basis{\cl{{\VVall}}}m}\otimes\rho}$.
Since states of this form span all cq-states,
    by linearity,
    ${\init\vv;\bc}\deneq\bc$.
  \item Case $\bc=\sample\XX e$:

    In this case, $\vv\in\overwr\bc=\XX\setminus\fv e$. Thus $\vv$ is
    a classical variable, $\vv\in\XX$, and $\vv\notin\fv e$.

    We have for all $m,\rho$:
    \begin{align*}
      &\denotc{\init\vv;\bc}\pb\paren{\proj{\basis{\cl{{\VVall}}}m}\otimes\rho} \\
      &= \denotc{\sample\XX e}\circ\denotc{\assign\vv{\initial\vv}}\pb\paren{\proj{\basis{\cl{{\VVall}}}m}\otimes\rho}
      && \text{(sem.~of ;, def.~of $\initkw$)} \\
      &= \denotc{\sample\XX e}\pb\paren{\proj{\basis{\cl{{\VVall}}}{m(\vv:={\initial\vv})}}\otimes\rho}
      && \text{(sem.~of assignment)}\\
      &= \sum_z \denotee e{m(\vv:=\initial\vv)}(z)\cdot  {\proj{\basis{\cl{{\VVall}}}{m(\vv:={\initial\vv})(\XX:=z)}}\otimes\rho}
      && \text{(sem.~of sample)}\\
      &\starrel= \sum_z \denotee e{m}(z)\cdot {\proj{\basis{\cl{{\VVall}}}{m(\vv:={\initial\vv})(\XX:=z)}}\otimes\rho} \\
      &\starstarrel= \sum_z \denotee e{m}(z)\cdot {\proj{\basis{\cl{{\VVall}}}{m(\XX:=z)}}\otimes\rho} \\
      &= \denotc{\sample\XX e}\pb\paren{\proj{\basis{\cl{{\VVall}}}m}\otimes\rho}
        = \denotc\bc\pb\paren{\proj{\basis{\cl{{\VVall}}}m}\otimes\rho}
      && \text{(sem.~of assignment)}
    \end{align*}
    Here $(*)$ follows since $\vv\notin\fv e$ and thus
    $\denotee e{m(\vv:=\initial\vv)}=\denotee em$.  And $(**)$ follows
    since $\vv\in\XX$ and thus
    $m(\vv:={\initial\vv})(\XX:=\denotee e{m}) = m(\XX:=\denotee
    e{m})$.

    Thus $\denotc{\init\vv;\bc}=\denotc\bc$ on all states of the form
    ${\proj{\basis{\cl{{\VVall}}}m}\otimes\rho}$. By linearity,
    ${\init\vv;\bc}\deneq\bc$.
    
  \item Case $\bc=\Qinit\QQ e$:

    In this case, $\vv\in\overwr\bc=\QQ$. Thus $\vv$ is a quantum variable and $\vv\in\QQ$.
    We have for all $m,\rho$:
    \begin{align*}
      \denotc{\init\vv;\bc}\pb\paren{\proj{\basis{}m}\otimes\rho}
      &=
        \denotc{\Qinit\QQ e}\circ\denotc{\init\vv}
        \pb\paren{\proj{\basis{}m}\otimes\rho}
      && \text{(def.~of $\bc$, semantics of ;)} \\
      &=
      \denotc{\Qinit\QQ e}\pb\paren{
      \proj{\basis{}m}\otimes
      \partr\vv{}\rho \otimes \proj{\psiinit\vv}
        }
        && \text{(\autoref{lemma:init.sem})}
      \\&
      =
      \proj{\basis{}m}\otimes
      \partr{}\QQ \pb\paren{\partr{}\vv\rho \otimes \proj{\psiinit\vv}}
      \otimes \denotee em
        && \text{(sem.~of quant.~init.)}
      \\&
      =
      \proj{\basis{}m}\otimes
      \partr{}\QQ \rho
      \otimes \denotee em
      && \text{($\vv\in\QQ$)}
      \\&
      =
      \denotc{\Qinit\QQ e}\pb\paren{
      \proj{\basis{}m}\otimes
      \rho
      }
        && \text{(sem.~of quant.~init.)}
      \\&
      =
      \denotc\bc\pb\paren{
      \proj{\basis{}m}\otimes
      \rho
      }
        && \text{(definition of $\bc$)}
    \end{align*}
    Thus $\denotc{\init\vv;\bc}=\denotc\bc$ on all states of the form
    ${\proj{\basis{\cl{{\VVall}}}m}\otimes\rho}$. By linearity,
    ${\init\vv;\bc}\deneq\bc$.
  \item Case $\bc=\langif e{\bc'}{\bd'}$:

    Since $\vv\in\overwr\bc=(\overwr{\bc'} \cap \overwr{\bd'})\setminus\fv e$ we have:
    $\vv\in\overwr{\bc'}$, $\vv\in\overwr{\bd'}$, $\vv\notin\fv e$.

    We will show that $\denotc{\init\vv;\bc}\pb\paren{\proj{\basis{\cl{{\VVall}}}m}\otimes\rho}
    = \denotc{\bc}\pb\paren{\proj{\basis{\cl{{\VVall}}}m}\otimes\rho}$ for all $m,\rho$.
    By linearity, this then shows ${\init\vv;\bc}\deneq\bc$.
    Fix $m,\rho$.
    We assume $\denotee e m=\true$. The case
    $\denotee e m=\false$ is shown analogously. 

    We distinguish two cases, depending on whether $\vv$ is a classical or a quantum variable:
    \begin{itemize}
    \item Case $\vv$ is classical:

      Since $\vv\notin\fv e$, we have
      \begin{equation}\label{eq:emv.true}
        \denotee e{m(\vv:=\initial\vv)} = \denotee em = \true
      \end{equation}
      And
      \begin{align*}
        &\denotc{\init\vv;\bc}\pB\paren{\proj{\basis{\cl{{\VVall}}}m}\otimes\rho} \\
        &= \denotc{\langif e{\bc'}{\bd'}}\circ\denotc{\assign\vv{\initial\vv}}\pb\paren{\proj{\basis{\cl{{\VVall}}}m}\otimes\rho}
        && \text{(sem.~of ;, def.~of $\initkw$)} \\
        &= \denotc{\langif e{\bc'}{\bd'}}\pb\paren{\proj{\basis{\cl{{\VVall}}}{m(\vv:=\initial\vv)}}\otimes\rho}
        && \text{(sem.~of assignment)} \\
        &= \denotc{\bc'}\circ\underbrace{\restricte e\pb\paren{\proj{\basis{\cl{{\VVall}}}{m(\vv:=\initial\vv)}}\otimes\rho}}_{\eqrefrel{eq:emv.true}=\ {\proj{\basis{\cl{{\VVall}}}{m(\vv:=\initial\vv)}}\otimes\rho}}
          + \denotc{\bd'}\circ\underbrace{\restricte{\lnot e}\pb\paren{\proj{\basis{\cl{{\VVall}}}{m(\vv:=\initial\vv)}}\otimes\rho}}_{\eqrefrel{eq:emv.true}=\ 0}
          \hspace{-1in}
        \\[-17pt]
        &&& \qquad\qquad\text{(sem.~of if)} \\
        & = \denotc{\bc'}\pb\paren{\proj{\basis{\cl{{\VVall}}}{m(\vv:=\initial\vv)}}\otimes\rho} 
         = \denotc{\init\vv; \bc'}\pb\paren{\proj{\basis{\cl{{\VVall}}}{m}}\otimes\rho}
        && \text{(sem., def.~of init)}
        \\
        &= \denotc{\bc'}\pb\paren{\proj{\basis{\cl{{\VVall}}}{m}}\otimes\rho}
        && \text{(induction hypothesis)} \\
        & = \denotc{\langif e{\bc'}{\bd'}}\pb\paren{\proj{\basis{\cl{{\VVall}}}{m}}\otimes\rho}
        = \denotc\bc\pb\paren{\proj{\basis{\cl{{\VVall}}}{m}}\otimes\rho}
        && \text{(sem.~of if, \eqref{eq:emv.true})}
      \end{align*}
      
    \item Case $\vv$ is quantum:

      Let $\rho' := \partr{}\QQ \rho \otimes \proj{\denotee em}$. We have
      \begin{align*}
        &\denotc{\init\vv;\bc}\pB\paren{\proj{\basis{\cl{{\VVall}}}m}\otimes\rho} \\
        &= \denotc{\langif e{\bc'}{\bd'}}\circ\denotc{\Qinit\vv{\psiinit\vv}}\pb\paren{\proj{\basis{\cl{{\VVall}}}m}\otimes\rho}
        && \text{(sem.~of ;, def.~of $\initkw$)} \\
        &= \denotc{\langif e{\bc'}{\bd'}}\pb\paren{\proj{\basis{\cl{{\VVall}}}{m}}\otimes\rho'}
        && \text{(sem.~of qu.~init.)} \\
        &= \denotc{\bc'}\circ\underbrace{\restricte e\pb\paren{\proj{\basis{\cl{{\VVall}}}{m}}\otimes\rho'}}_{=\ {\proj{\basis{\cl{{\VVall}}}{m}}\otimes\rho'}}
          + \denotc{\bd'}\circ\underbrace{\restricte{\lnot e}\pb\paren{\proj{\basis{\cl{{\VVall}}}{m}}\otimes\rho'}}_{=\ 0}
          \hspace{-1in}
        && \text{(sem.~of if)} \\
        & = \denotc{\bc'}\pb\paren{\proj{\basis{\cl{{\VVall}}}{m}}\otimes\rho'} 
         = \denotc{\init\vv; \bc'}\pb\paren{\proj{\basis{\cl{{\VVall}}}{m}}\otimes\rho'}
        && \text{(sem., def.~of init)}
        \\
        &= \denotc{\bc'}\pb\paren{\proj{\basis{\cl{{\VVall}}}{m}}\otimes\rho'}
        && \text{(induction hypothesis)}
        \\
        & = \denotc{\langif e{\bc'}{\bd'}}\pb\paren{\proj{\basis{\cl{{\VVall}}}{m}}\otimes\rho}
        = \denotc\bc\pb\paren{\proj{\basis{\cl{{\VVall}}}{m}}\otimes\rho}
%        &= \denotc{\bc}\pb\paren{\proj{\basis{\cl{{\VVall}}}m}\otimes\rho}
        && \text{(sem.~of if, \eqref{eq:emv.true})}
      \end{align*}
    \end{itemize}

  \item Case $\bc=\bc';\bd'$:    

    We have
    $\vv\in\overwr\bc = \overwr{\bc'} \cup \pB\paren{
      \pb\paren{\overwr{\bd'}\setminus\fv{\bc'}} \cap \covered{\bc'}
    }$. We thus distinguish the following cases:
    \begin{itemize}
    \item Case $\vv\in\overwr{\bc'}$:

      By induction hypothesis, ${\init\vv;\bc'}\deneq{\bc'}$.
      Thus ${\init\vv;\bc}\deneq{\init\vv;\bc';\bd'}
      \deneq {\bc';\bd'} = \bc$.
      
    \item Case $\vv\in{\overwr{\bd'}\setminus\fv{\bc'}}$:

      Since $\vv\notin\fv{\bc'}$ and $\fv{\init\vv}=\{\vv\}$, we have
      ${\init\vv;\bc'} \deneq {\bc';\init\vv}$ by
      \autoref{lemma:swap}. And since $\vv\in\overwr{\bd'}$, by induction
      hypothesis, we have ${\init\vv;\bd'} \deneq {\bd'}$.
      Thus ${\init\vv;\bc} \deneq {\init\vv;\bc';\bd'} \deneq
      {\bc';\init\vv;\bd'} \deneq {\bc';\bd'} = \bc$.
    \end{itemize}
  \item Case $\bc=\Qmeasure \XX\QQ e$:

    Since $\vv\in\overwr\bc=\XX\setminus\fv e$, we have that $\vv$ is
    classical, $\vv\in\XX$, but $\vv\notin\fv e$.

    We will show that $\denotc{\init\vv;\bc}\pb\paren{\proj{\basis{\cl{{\VVall}}}m}\otimes\rho}
    = \denotc{\bc}\pb\paren{\proj{\basis{\cl{{\VVall}}}m}\otimes\rho}$ for all $m,\rho$.
    By linearity, this then shows ${\init\vv;\bc}\deneq\bc$.
    Fix $m,\rho$.

    Let $P_z := \denotee em(z) \otimes \id_{\qu{{\VVall}}\setminus\QQ}$. Since
    $\vv\notin\fv e$, we also have  $P_z = \denotee e{m(\vv:=\initial\vv)}(z) \otimes \id_{\qu{{\VVall}}\setminus\QQ}$
    
    We have for all $m,\rho$:
    \begin{align*}
      &\denotc{\init\vv;\bc}\pb\paren{\proj{\basis{\cl{{\VVall}}}m}\otimes\rho} \\
      &= \denotc{\Qmeasure\XX\QQ e}\circ\denotc{\assign\vv{\initial\vv}}\pb\paren{\proj{\basis{\cl{{\VVall}}}m}\otimes\rho}
      && \text{(sem.~of ;, def.~of $\initkw$)} \\
      &= \denotc{\Qmeasure\XX\QQ e}\pb\paren{\proj{\basis{\cl{{\VVall}}}{m(\vv:={\initial\vv})}}\otimes\rho}
      && \text{(sem.~of assignment)}\\
      &= \sum_z \proj{\basis{\cl{{\VVall}}}{m(\vv:={\initial\vv})(\XX:=z)}}\otimes P_z\rho P_z
      && \text{(sem.~of measurement)}\\
      &\starrel= \sum_z \proj{\basis{\cl{{\VVall}}}{m(\XX:=z)}}\otimes P_z\rho P_z
      \\
      & = \denotc{\Qmeasure\XX\QQ e}\pb\paren{\proj{\basis{\cl{{\VVall}}}{m}}\otimes\rho}
      && \text{(sem.~of measurement)}\\
      & = \denotc\bc\pb\paren{\proj{\basis{\cl{{\VVall}}}{m}}\otimes\rho}
    \end{align*}
    Here $(*)$ follows
    since $\vv\in\XX$ and thus
    $m(\vv:={\initial\vv})(\XX:=z) = m(\XX:=z)$.
  \item Case $\bc=\local\ww{\bc'}$:

    Since $\vv\in\overwr\bc=\overwr{\bc'}\setminus\{\ww\}$, we have
    $\vv\in\overwr{\bc'}$ but $\vv\neq\ww$.  By induction hypothesis,
    we have $\init\vv;\bc'\deneq\bc'$.  By
    \lemmaref{lemma:merge.local.fv}, we have
    $\init\vv;\local\ww{\bc'} \deneq \local\ww{\init\vv;\bc'}$.
    Thus $\init\vv;\bc
    \deneq
    \init\vv;\local\ww{\bc'} \deneq \local\ww{\init\vv;\bc'}
    \deneq \local\ww\bc' = \bc $.
  \end{itemize}

  Thus we have shown:
  \begin{equation}
    {\init\vv;\bc} \deneq \bc
    \qquad\text{if}\qquad
    \vv\in\overwr\bc.\label{eq:init.v}
  \end{equation}
  
  From this we can conclude \eqref{lemma:init.overwr:init}: Let $\{\vv_1,\dots,\vv_n\}:=\VV\subseteq\overwr\bc$.
  \begin{equation*}
    {\init\VV;\bc}
    \deneq
    {\init{\vv_1};\dots;\init{\vv_n};\bc}
    \deneq
    {\init{\vv_1};\dots;\init{\vv_{n-1}};\bc}
    \deneq
    \dots
    \deneq
    {\bc}.
  \end{equation*}
  This shows \eqref{lemma:init.overwr:init}.

  \medskip

  We show \eqref{lemma:init.overwr:X}. It is easy to check from the
  semantics that $\assign\XX e;\init\XX \deneq \init\XX$. Thus
  \[
    \assign\XX e;\bc
  \eqrefrel{lemma:init.overwr:init}
  \deneq
  \assign\XX e;\init\XX;\bc
  \deneq
  \init\XX;\bc
  \eqrefrel{lemma:init.overwr:init}
  \deneq
  \bc.
\]

\medskip

Finally, we 
show \eqref{lemma:init.overwr:Q}. It is easy to check from the
  semantics that $\assign\QQ e;\init\QQ \deneq \init\QQ$. Thus
  \begin{equation*}
    \Qinit\QQ e;\bc
  \eqrefrel{lemma:init.overwr:init}
  \deneq
  \Qinit\QQ e;\init\QQ;\bc
  \deneq
  \init\QQ;\bc
  \eqrefrel{lemma:init.overwr:init}
  \deneq
  \bc.
  \mathQED
\end{equation*}
\end{proof}

\begin{lemma}\lemmalabel{lemma:move}
  \begin{compactenum}[(i)]
  \item \itlabel{lemma:move.block}
    For $\bc := \localp{\VV_1}{\bc_1};\dots;\localp{\VV_n}{\bc_n}$, we have
    $\bc
    \deneq
      \local{\VV^\uparrow}{
        \paren{\bc'_1;\dots;\bc'_n}
        }$
    where
    $\bc_i' := \init {\VV_i^*}; \local{\VV_i^\downarrow}{\bc_i}$,
    $\WW_1 := \VV^\uparrow$,
    $\WW_{i+1} := \WW_i \cup \VV_i^* - (\fv{\bc_i} - \VV^\downarrow_i)$,
    $\VV_i - \VV_i^\downarrow - \WW_i \subseteq \VV^*_i \subseteq \VV^\uparrow$,
    $\VV^\downarrow_i\subseteq\VV_i\subseteq\VV^\downarrow_i \cup \VV^\uparrow$,
    $\VV^\uparrow \cap \fv\bc = \varnothing$.\footnote{The greedy instantiation is
      $\VV^\uparrow := \pb\paren{\bigcup_i \VV_i} - \fv\bc$,
      $\VV^\downarrow_i := \VV_i - \VV^\uparrow$,
      $\VV^*_i := \VV_i - \VV^\downarrow_i - \WW_i$.}

  \item \itlabel{lemma:move.if}
    Assume
  $\VV^\uparrow \cup \VV^\downarrow_\mathit{then} \supseteq \VV_\mathit{then} \supseteq \VV^\downarrow_\mathit{then}$, 
  $\VV^\uparrow \cup \VV^\downarrow_\mathit{else} \supseteq \VV_\mathit{else} \supseteq \VV^\downarrow_\mathit{else}$,
  $\VV^\uparrow \cap \fv e = \varnothing$,
  $\VV^\uparrow \cap (\fv{\bc_\mathit{then}} \setminus \VV_\mathit{then})=\varnothing$,
  $\VV^\uparrow \cap (\fv{\bc_\mathit{else}} \setminus \VV_\mathit{else})=\varnothing$.\footnote{The greedy instantiation
    is $\VV^\uparrow := (\VV_\mathit{then}\cup \VV_\mathit{else})
    \setminus \pb\paren{\fv{\bc_\mathit{then}} \setminus \VV_\mathit{then}}
    \setminus \pb\paren{\fv{\bc_\mathit{else}} \setminus \VV_\mathit{else}}
    \setminus \fv e$,
    $\VV_\mathit{then}^\downarrow
    := \VV_\mathit{then} \setminus \VV^\uparrow$
    $\VV_\mathit{else}^\downarrow
    := \VV_\mathit{else} \setminus \VV^\uparrow$. }
  Then
  \begin{multline*}
    {\langif e{\{\local{\VV_\mathit{then}}\bc_\mathit{then}\}}
    {\{\local{\VV_\mathit{else}}\bc_\mathit{else}\}}}
  \\
  \deneq
  {
    \local{\VV^\uparrow}{
      \langif e{\{\local{\VV^\downarrow_\mathit{then}}\bc_\mathit{then}\}}
      {\{\local{\VV^\downarrow_\mathit{else}}\bc_\mathit{else}\}}}}
\end{multline*}

  \item\itlabel{lemma:move.while}
    Assume $\VV^\uparrow \subseteq \VV \setminus \fv e$ and $\VV^\downarrow := \VV \setminus \VV^\uparrow$.\footnote{The greedy instantiation is $\VV^\uparrow := \VV \setminus \fv e$.}
      Then
      \[
        {\while e{\{ \local \VV \bc \}}}
        \deneq {\local{\VV^\uparrow}{\while e{\{\local{\VV^\downarrow}\init{\VV^\uparrow};\bc\}}}}
      \]
\end{compactenum}
\end{lemma}

This lemma allows us to move local variable declarations upwards in a
program without changing its semantics.  We have implemented a tactic
\texttt{local up}%
\index{local up@\texttt{local up} (tactic)}%
\pagelabel{page:tactic:local-up}
in \texttt{qrhl-tool} that allows us to move all or
selected local variables upwards as far as possible. This is done by
repeatedly applying \autoref{lemma:move} repeatedly as long as
possible. In the footnotes in \autoref{lemma:move}, we describe how
the tactic instantiates the lemma (we call this the ``greedy
instantiation'') in order to move as many variables as possible.

\begin{proof}
  \eqref{lemma:move.block}
  is shown in Isabelle/HOL, as \verb|locals_up_block| in theory
  \verb|Locals_Up_Block.thy|. See \autoref{sec:isabelle-proofs} for
  remarks about our Isabelle/HOL development.

{\TODOQ{Proof sketch}}

  \bigskip

  We now show \eqref{lemma:move.if}.
  By definition,
  $\denotc{\langif e{\bc_\mathit{then}}{\bd_\mathit{else}}}
  = \denotc{\bc_\mathit{then}}\circ\restricte e + \denotc{\bd_\mathit{else}}\circ\restricte{\lnot e}$.
  Let $\bd_e,\bd_{\lnot e}$ be  arbitrary programs with
  $\denotc{\bd_e}=\restricte e$,   $\denotc{\bd_{\lnot e}}=\restricte{\lnot e}$ and $\fv{\bd_e},\fv{\bd_{\lnot e}}=\fv e$.\footnote{E.g.,
    $\bd_e:=\langif e{\textbf{halt}}{\Skip}$,
    $\bd_{\lnot e}:=\langif e\Skip{\textbf{halt}}$,
    $\textbf{halt} := \while\true\Skip$.}
  Then
  \begin{equation}
    \denotc{\langif e{\bc_\mathit{then}}{\bd_\mathit{else}}}
    = \denotc{\bd_e; \bc_\mathit{then}} + \denotc{\bd_{\lnot e}; \bc_\mathit{else}}
    \qquad
    \text{for any $\bc_\mathit{then},\bc_\mathit{else}$}.
    \label{eq:if.unroll}
  \end{equation}
  Thus
  \begin{align*}
    &\denotc{\langif e{{\localp{\VV_\mathit{then}}{\bc_\mathit{then}}}}
      {{\localp{\VV_\mathit{else}}{\bc_\mathit{else}}}}}
    \\
    &\eqrefrel{eq:if.unroll}=
      \denotc{\bd_e;\local{\VV_\mathit{then}}{\bc_\mathit{then}}}
      +
      \denotc{\bd_{\lnot e};\local{\VV_\mathit{else}}{\bc_\mathit{else}}}
    \\
    &\eqrefrel{lemma:move.block}=
      \denotc{\localpp{\VV^\uparrow}{\bd_e;\local{\VV^\downarrow_\mathit{then}}{\bc_\mathit{then}}}}
      +
      \denotc{\local{\VV^\uparrow}{\bd_{\lnot e};\local{\VV^\downarrow_\mathit{else}}{\bc_\mathit{else}}}}
    \\
    &\txtrel{\lemmarefshort{lemma:local.LOCAL}}=\quad
      \pb\LOCAL{\VV^\uparrow}{\denotc{\bd_{e};\local{\VV^\downarrow_\mathit{then}}{\bc_\mathit{then}}}}
      + \pb\LOCAL{\VV^\uparrow}{\denotc{\bd_{\lnot e};\local{\VV^\downarrow_\mathit{else}}{\bc_\mathit{else}}}}
    \\
    &\txtrel{\lemmarefshort{lemma:local.sum}}=\ \,
      \pb\LOCAL{\VV^\uparrow}{\denotc{\bd_{e};\local{\VV^\downarrow_\mathit{then}}{\bc_\mathit{then}}}}
      + \denotc{\bd_{\lnot e};\local{\VV^\downarrow_\mathit{else}}{\bc_\mathit{else}}}
    \\
    &\eqrefrel{eq:if.unroll}=
      \pb\LOCAL{\VV^\uparrow}{
      \denotc{\langif e{{\localp{\VV^\downarrow_\mathit{then}}{\bc_\mathit{then}}}}
      {{\localp{\VV^\downarrow_\mathit{else}}{\bc_\mathit{else}}}}}}
    \\
    &
      \txtrel{\lemmarefshort{lemma:local.LOCAL}}=\quad
      \denotc{
      \local{\VV^\uparrow}{
      \langif e{{\localp{\VV^\downarrow_\mathit{then}}{\bc_\mathit{then}}}}
      {{\localp{\VV^\downarrow_\mathit{else}}{\bc_\mathit{else}}}}}}.
  \end{align*}
  In this calculation, $\eqrefrel{lemma:move.block}=$ uses two
  invocations of \eqref{lemma:move.block}, one to convert each of the
  summands. The first invocation is instantiated with
  $\bc:={\bd_e;\local{\VV_\mathit{then}}{\bc_\mathit{then}}}$,
  $\bc_1:=\bd_e$, $\bc_2:=\bc_\mathit{then}$, $\VV_1:=\varnothing$,
  $\VV_2:=\VV_\mathit{then}$, $\VV_1^\downarrow:=\varnothing$,
  $\VV_2^\downarrow:=\VV^\downarrow_\mathit{then}$,
  $\VV_1^*,\VV_2^*:=\varnothing$, $\VV^\uparrow:=\VV^\uparrow$,
  $\WW_1:=\VV^\uparrow$, $\WW_2:=\VV^\uparrow\setminus\fv e$.  The
  premises of \eqref{lemma:move.block} then follow from the premises
  of \eqref{lemma:move.if} (using also that $\fv{\bd_e}=\fv e$ and
  $\fv{\bd_e;\local{\VV_\mathit{then}}{\bc_\mathit{then}}} = \fv
  e\cup(\fv{\bc_\mathit{then}}\setminus \VV_\mathit{then})$.  The
  second invocation is instantiated analogously, with $\mathit{else}$
  instead of $\mathtt{then}$, and with $\bc_1:=\bd_{\lnot e}$.
  
  This shows \eqref{lemma:move.if}.
    
  \bigskip

  We now show \eqref{lemma:move.while}. 
  
  By definition,
  $\denotc{\while e\bc}
  = \sum_{i=0}^\infty
  \restricte{\lnot e}
  \circ (\denotc\bc\circ\restricte e)^i$.
  (Convergence is pointwise with respect to the trace-norm.)
  Let $\bd_e,\bd_{\lnot e}$ be  arbitrary programs with
  $\denotc{\bd_e}=\restricte e$,   $\denotc{\bd_{\lnot e}}=\restricte{\lnot e}$ and $\fv{\bd_e},\fv{\bd_{\lnot e}}=\fv e$.
  Then
  \begin{equation}
    \denotc{\while e\bc}
  = \sum_{i=0}^\infty
  \denotc{\underbrace{\bd_e;\bc;\dots;\bd_e;\bc}_{i\text{ times}}
    ;\bd_{\lnot e}}
  \qquad
  \text{for any $\bc$}
  \label{eq:while.unroll}
\end{equation}
  From \eqref{lemma:move.block}, we have
  \begin{multline} \label{eq:move.while.i}
    {\bd_e;{\localp\VV\bc};\dots;\bd_e;{\localp\VV\bc};\bd_{\lnot e}}
    \\
    \deneq
    \localpp{\VV^\uparrow}{\underbrace{\bd_e;(\init\VV^\uparrow;\local{\VV^\downarrow}\bc);\dots;\bd_e;(\init\VV^\uparrow;\local{\VV^\downarrow}\bc);\bd_{\lnot e}}_{=:\bc_i^*}}
  \end{multline}
  Specifically, we instantiate \eqref{lemma:move.block} as follows:
  \begin{alignat*}8
    \bc_1,\dots,\bc_{2n+1} &\ :=\  \bd_e,\  &&\bc,\            &&\bd_e,\      &&\bc,\            &&\dots,\ &&\bd_e,\ &&\bc,\           &&\bd_{\lnot e} \\
    \bc_1',\dots,\bc_{2n+1}' &\ :=\  \bd_e,\ &&\scriptstyle (\init\VV^\uparrow;\local{\VV^\downarrow}\bc),\ &&\bd_e,       &&\scriptstyle (\init\VV^\uparrow;\local{\VV^\downarrow}\bc),\ &&\dots, &&\bd_e,    &&\scriptstyle (\init\VV^\uparrow;\local{\VV^\downarrow}\bc),\ &&\bd_{\lnot e} \\
    \VV_1,\dots\VV_{2n+1} &\ :=\  \varnothing,\ &&\VV,       &&\varnothing, &&\VV,              &&\dots, &&\varnothing, &&\VV,          && \varnothing \\
    \VV^*_1,\dots\VV^*_{2n+1} &\ :=\  \varnothing,\ &&\VV^\uparrow,       &&\varnothing, &&\VV^\uparrow,              &&\dots, &&\varnothing, &&\VV^\uparrow,          && \varnothing \\
    \VV^\downarrow_1,\dots\VV^\downarrow_{2n+1} &\ :=\  \varnothing,\ &&\VV^\downarrow,       &&\varnothing, &&\VV^\downarrow,              &&\dots, &&\varnothing, &&\VV^\downarrow,          && \varnothing
  \end{alignat*}
  As well as $ \VV^\uparrow := \VV^\uparrow$. (The values of
  $\WW_i$ do not matter in this application of
  \eqref{lemma:move.block}.) Then the premises of
  \eqref{lemma:move.block} are elementary to check (using
  also $\fv{\bd_e;{\localp\VV\bc};\dots;\bd_e;{\localp\VV\bc};\bd_{\lnot e}}=\fv e\cup(\fv\bc\setminus\VV)$
  and the premises $\VV^\uparrow\subseteq \VV\setminus\fv e$ and $\VV^\downarrow=\VV\setminus\VV^\uparrow$).
  \eqref{eq:move.while.i} follows.

  And from \eqref{eq:while.unroll}, we have
  \begin{equation}
    \label{eq:while.unroll2}
    \while e{\paren{\init\VV^\uparrow;\local{\VV^\downarrow}\bc}}
    \ \deneq\
    \sum_{i=0}^\infty
    \denotc{\bc_i^*}.
  \end{equation}
  
  This implies:
  \begin{align*}
    \denotc{\while e{\paren{\local\VV\bc}}}
    \ &\txtrel{\eqref{eq:while.unroll},\eqref{eq:move.while.i}}=\
    \sum_{i=0}^\infty \denotc {\local{\VV^\uparrow}{\bc_i^*}}
    =
    \sum_{i=0}^\infty
      \pb\LOCAL{\VV^\uparrow}{\denotc{\bc_i^*}}
    \\
    &\txtrel{\lemmarefshort{lemma:local.sum}}=\ \,
      \pB\LOCAL{\VV^\uparrow}{    \sum_{i=0}^\infty \denotc{\bc_i^*}}
    \eqrefrel{eq:while.unroll2}=
      \pB\LOCAL{\VV^\uparrow}{\denotc{\while e{\paren{\init\VV^\uparrow;\local{\VV^\downarrow}\bc}}}}
    \\
    &=
    \denotc{\local{\VV^\uparrow}{\while e{\paren{\init\VV^\uparrow;\local{\VV^\downarrow}\bc}}}} \\
    &\txtrel{\lemmarefshort{lemma:merge.local.fv}}=\quad
    \denotc{\local{\VV^\uparrow}{\while e{\paren{\localpp{\VV^\downarrow}{\init\VV^\uparrow;\bc}}}}}.
  \end{align*}
  This shows \eqref{lemma:move.while}.
\end{proof}

\section{Basic rules for variable renaming/removal}
\label{sec:basic.var.rules}

\ERULE{RenameQrhl1}{
    \noconflict\sigma\bc
    \\
    \sigma\text{ injective on } \fv\bc\cup \{\vv :
    \vv_1\in\fv\PA\cup\fv\PB\}
    \\
    \rhl{\PA\sigma_1}{\bc\sigma}\bc{\PB\sigma_1} }{ \rhl\PA\bc\bd\PB }
  
  Here $\sigma_1(\vv_1):=\ww_1$ whenever $\sigma(\vv)=\ww$ (and
  $\sigma_1(\vv_2):=\vv_2$ for all $\vv$).

  Analogously \texttt{RenameQrhl2} for renaming in $\bd$.

This is shown in Isabelle/HOL, as \verb|rename_qrhl_left| and
\verb|rename_qrhl_right| in theory \verb|Rename_Locals.thy|. See
\autoref{sec:isabelle-proofs} for remarks about our Isabelle/HOL
development.

We have implemented a tactic \texttt{rename}%
\index{rename@\texttt{rename} (tactic)}%
\pagelabel{page:tactic:rename}
in \texttt{qrhl-tool} that implements this rule.

{\TODOQ{Proof sketch}}

\ERULE{RemoveLocal1}{
  \fv\PA,\fv\PB\cap \VV_1=\varnothing
  \\
  \pb\rhl{\PA\cap
    \paren{\VV_1\eqstate\psiinit\VV}
  }\bc\bd\PB
}{
  \rhl\PA{\local\VV\bc}\bd\PB
}

Analogously \texttt{RemoveLocal2} for removing on the right side.

Note that the converse of this rule does not hold: If
    $\rhl\PA{\local\VV\bc}\bd\PB$, then we do not necessarily have
    $ \pb\rhl{\PA\cap \paren{\VV_1\eqstate\psiinit\VV} }\bc\bd\PB $.\footnote{Counterexample:
    Let $\bc:=\paren{\Qinit{\qq\rr}{\fsq\basis{}{00}
        + \fsq\basis{}{11}}}$,
    $\bd:=\paren{\sample\xx\bit}$ (here $\bit$ stands for the uniform distribution on $\bit$),
    $\PA:=\CL\true$,
    $\PB:=\paren{\qq_1\eqstate\basis{}{\xx_2}}$.
    Then
    $\rhl{\PA\cap\paren{\qq_1\eqstate\psiinit\qq}}\bc\bd\PB$ does not hold
    (this would imply that $\qq$ is unentangled with $\rr$ after $\bc$).
    Yet $\local\rr\bc$ is equivalent to assigning
    randomly $\basis{}0$ or $\basis{}1$ to $\qq$, thus
    $\rhl\PA{\local\yy\bc}\bd\PB$ holds.
  }

  We have implemented this rule as the tactic \texttt{local remove}%
  \index{local remove@\texttt{local remove} (tactic)}%
  \pagelabel{page:tactic:remove} in
  \texttt{qrhl-tool} which allows us to remove selected (or all) local
  variable declarations from the top of the left or right program. The
  tactic is a little weaker in that it does not include
  ${\VV_1\eqstate\psiinit\VV}$ in the precondition of the new subgoal.

\begin{proof}
  We show the rule for $\VV=\vv$ (only one variable). The general case follows by induction.
  
  Let $\ww\in\VVall$ be an arbitrary variable such that
  $\ww_1\notin\fv{\PA,\PB}$, $\ww\notin\fv\bc$,
  $\abs{\typev\ww}\geq\abs{\typev\vv}$,
  $\abs{\typev\ww}=\infty$.\footnote{Such $\ww$ exists because there
    are infinitely many $\ww$ with $\typev\ww=\typev\vv$ and
    infinitely many $\ww$ with $\typev\ww=\infty$,
    and $\fv\PA$, $\fv\PB$, $\fv\bc$ are finite (see
    preliminaries).}
  Then $\abs{\types{\vv'\ww}}=\abs{\types{\vv\ww}}=\abs{\typev\ww}$.
  Hence there is a bijection $\phi:\typev\ww\to\types{\vv'\ww}$.
  Thus $U:\basis{\ww}{x}\mapsto\basis{\vv'\ww}{\phi(x)}$ is a unitary from $\elltwov{\ww}$ to $\elltwov{\vv'\ww}$.

  Since $\ww\notin\fv\bc$, by \autoref{lemma:fv},
  $\denotc\bc=\calE\otimes\id_{\ww}$ for some $\calE$ on
  $\VVall\setminus\ww$.  In slight abuse of notation, we also write
  $\denotc\bc$ for that $\calE$.

  Since $\vv,\ww\notin\fv\PA,\fv\PB$, there are $\PA',\PB'$ such that
  $\PA=\PA'\otimes\elltwov{\vv\ww}$ and
  $\PB=\PB'\otimes\elltwov{\vv\ww}$.
  
  Consider the following circuit (but ignore the wavy lines with boxes on the bottom for now):
  \newcounter{vertline}
  \newcommand\vertline[2]{
    \refstepcounter{vertline}
    \draw[very thin, decorate, decoration={snake, segment length=1mm, amplitude=.1mm}]
    (\currentXPos |- \getWireCoord{above})
          node[inner ysep=0, inner xsep=1pt, anchor=west] {\tiny$\rho^{\scriptscriptstyle(\thevertline)}$}
    -- (\currentXPos |- \getWireCoord{below#1})
          node[draw, very thin, inner ysep=2pt, inner xsep=-4pt, anchor=north]
              {\footnotesize $\begin{array}{c}#2 \end{array}$};
  }
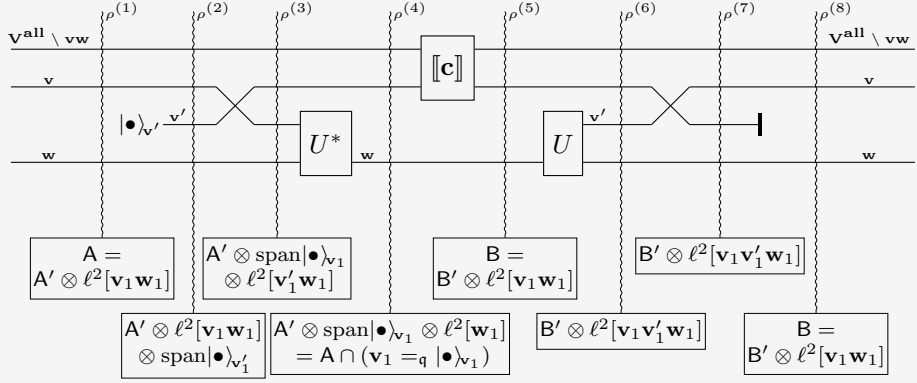
  \begin{center}
    %\label{eq:remove.local.circ}
    \begin{tikzpicture}
      \initializeCircuit;
      \newWires{above,rest,v,v2,w};
      \newWireRelative{below1}{w}{-1cm};
      \newWireRelative{below2}{below1}{-1cm};
      \stepForward{5mm};
      \labelWire[\tiny $\VVall\setminus\vv\ww$]{rest};
      \labelWire[\tiny $\vv$]{v};
      \labelWire[\tiny $\ww$]{w};
      \stepForward{7mm};
      \vertline1{\PA = \\ \PA'\otimes\elltwov{\vv_1\ww_1}};
      \label{at-start};
      \stepForward{8mm};
      \node[wireInput={v2}] (init-v2) {\footnotesize $\psiinit{\vv'}$};
      \stepForward{2mm};
      \labelWire[\tiny $\vv'$]{v2};
      \stepForward{2mm};
      \vertline2{\PA'\otimes\elltwov{\vv_1\ww_1}\\{}\otimes\SPAN\psiinit{\vv'_1}};
      \label{after-init};
      \stepForward{3mm};
      \drawWires{v,v2};
      \stepForward{5mm};
      \crossWire{v}{v2};
      \crossWire{v2}{v};
      \skipWires{v,v2};
      \stepForward{3mm};
      \vertline1{\PA'\otimes\SPAN\psiinit{\vv_1}\\{}\otimes\elltwov{\vv'_1\ww_1}};
      %\vertline1{\includegraphics[width=1.3cm]{/tmp/cat.png}};
      \label{after-swap};
      \stepForward{3mm};
      \node[gateAsy={v2,w}{w}] (U) {$\adj U$};
      \stepForward{2mm};
      \labelWire[\tiny $\ww$]{w};
      \stepForward{3mm};
      \vertline2{\PA'\otimes\SPAN\psiinit{\vv_1}\otimes\elltwov{\ww_1}
        \\{}=\PA\cap(\vv_1\eqstate\psiinit{\vv_1})};
      \label{after-U*};
      \stepForward{4mm};
      \node[gate={rest,v}] (c) {$\denotc\bc$};
      \stepForward{4mm};
      \vertline1{\PB=\\{}\PB'\otimes\elltwov{\vv_1\ww_1}};
      \label{after-c};
      \stepForward{5mm};
      \node[gateAsy={w}{v2,w}] (U) {$U$};
      \stepForward{2mm};
      \labelWire[\tiny $\vv'$]{v2};
      \stepForward{3mm};
      \vertline2{\PB'\otimes\elltwov{\vv_1\vv'_1\ww_1}};
      \label{after-U};
      \stepForward{4mm};
      \drawWires{v,v2};
      \stepForward{5mm};
      \crossWire{v}{v2};
      \crossWire{v2}{v};
      \skipWires{v,v2};
      \stepForward{4mm};
      \vertline1{\PB'\otimes\elltwov{\vv_1\vv'_1\ww_1}};
      \label{after-swap2};
      \stepForward{5mm};
      \node[killWire=v2] (kill-v2) {};
      \stepForward{7mm};
      \vertline2{\PB=\\{}\PB'\otimes\elltwov{\vv_1\ww_1}};
      \label{after-kill};
      \stepForward{7mm};
      \labelWire[\tiny $\VVall\setminus\vv\ww$]{rest};
      \labelWire[\tiny $\vv$]{v};
      \labelWire[\tiny $\ww$]{w};
      \stepForward{6mm};
      \drawWires{rest,v,w};
    \end{tikzpicture}
  \end{center}
  First, not that $U$ commutes with $\denotc\bc$. And since $U$ is
  unitary, $\adj U$ and $U$ are inverses. Thus $\adj U$ and $U$ cancel
  out in this circuit. The remaining circuit is by definition
  $\denotc{\local\vv\bc}$ (see \eqref{eq:semantics.local}).

  Abbreviating, we say ``$\rho_1,\rho_2$ satisfy $\PA$'' iff there exists
  a separable $\rho$ (the ``coupling'') such that $\partr{}2\rho=\rho_1$,
  $\partr{}1\rho=\rho_2$, and $\rho$ satisfies $\PA$.
  
  We need to show $\rhl\PA{\local\vv\bc}\bd\PB$. For this, fix
  cq-operators $\rho_1,\rho_2$ on $\VVall$ satisfying~$\PA$.  We need
  to show that $\denotc{\local\vv\bc}(\rho_1),\denotc\bd(\rho_2)$
  satisfy $\PB$.

  \newcommand\rhoref[1]{\rho^{\scriptscriptstyle\eqref{#1}}}
  \newcommand\Cref[1]{\PC^{\scriptscriptstyle\eqref{#1}}}

  Let $\rhoref{at-start},\dots,\rhoref{after-kill}$ be the
  states at the corresponding wavy lines when executing the above
  circuit with initial state $\rho_1$. In particular, $\rhoref{at-start}=\rho_1$
  and $\rhoref{after-kill}=\denotc{\local\vv\bc}(\rho_1)$. 
  Let $\PC^{(i)}$ denote the predicate given in the box under the wavy line
  for $\rho^{(i)}$.

  Since  $\rho_1,\rho_2$ satisfy~$\PA$, we have that
  $\rhoref{at-start},\rho_2$ satisfy $\Cref{at-start}$.

  Then $\rhoref{after-init},\rho_2$ satisfy $\Cref{after-init}$. (The
  coupling is $\rho\otimes\proj{\psiinit{\vv'_1}}$
  if the previous coupling was $\rho$.)

  Then $\rhoref{after-swap},\rho_2$ satisfy $\Cref{after-swap}$. (The
  coupling is  $(\toE{\swapcop\vv}\otimes\id)\rho$
  if the previous coupling was $\rho$.)

  Then $\rhoref{after-U*},\rho_2$ satisfy $\Cref{after-U*}$
  since $\adj U$ maps $\elltwov{\vv_1'\ww_1}$ to $\elltwov{\ww_1}$. (The
  coupling is $(\toE{\adj U}\otimes\id)\rho$ if the previous coupling
  was $\rho$.)

  By assumption,
  $ \pb\rhl{\PA\cap \paren{\vv_1\eqstate\psiinit{\vv_1}} }\bc\bd\PB $.
  And since $\rhoref{after-U*},\rho_2$ satisfy
  $\Cref{after-U*}=A\cap(\vv_1\eqstate\psiinit{\vv_1})$,
  we have that $\denotc\bc(\rhoref{after-U*}),\denotc\bd(\rho_2)$
  satisfy $\PB=\Cref{after-c}$.
  Since $\rhoref{after-c}=\denotc\bc(\rhoref{after-U*})$,
  we have that
  $\rhoref{after-c},\denotc\bd(\rho_2)$ satisfy $\Cref{after-c}$.

  Then $\rhoref{after-U},\denotc\bd(\rho_2)$ satisfy $\Cref{after-U}$
  since $U$ maps $\elltwov{\ww_1}$ to $\elltwov{\vv_1'\ww_1}$. (The
  coupling is $(\toE{U}\otimes\id)\rho$ if the previous coupling
  was $\rho$.)

  Then $\rhoref{after-swap2},\denotc\bd(\rho_2)$ satisfy $\Cref{after-swap2}$. (The
  coupling is  $(\toE{\swapcop\vv}\otimes\id)\rho$
  if the previous coupling was $\rho$.)

    Then $\rhoref{after-kill},\denotc\bd(\rho_2)$ satisfy $\Cref{after-kill}$. (The
  coupling is  $\partr{}{\vv'}\rho$
  if the previous coupling was $\rho$.)

  As mentioned above, $\rhoref{after-kill}=\denotc{\local\vv\bc}(\rho_1)$.
  And $\Cref{after-kill} = \PB$. Thus 
  $\denotc{\local\vv\bc}(\rho_1),\denotc\bd(\rho_2)$ satisfy $\PB$.

  This shows  $\rhl\PA{\local\vv\bc}\bd\PB$.
\end{proof}

\section{Two-sided initialization}
\label{sec:two-sided.init}

\ERULE{JointQInitEq}{
  \fv\PB\cap \QQ_1\QQ'_2=\varnothing
  \\\\
  \PA' := \PB \cap \pb\paren{(V\otimes U)\RR_1\QQ_1
    \quanteq (V'\otimes U')\RR'_2\QQ'_2}
  \cap \CL{\text{$U,U',V,V'$ are isometries}}
  \\\\
  \PB':={\PB \cap \paren{V\RR_1 \quanteq V'\RR'_2}}\cap
  \SPAN\basis{\QQ_1}{e_1} \cap \SPAN\basis{\QQ'_2}{e'_2}
}{
  \pb \rhl
  {\PA'}
  {\Qinit \QQ e}  {\Qinit{\QQ'}{e'}}
  {\PB'}
}

The following simple case is probably easier to understand at first reading.
We obtain it 
by setting $U,U',V,V' := \id$ and weakening the postcondition.
\ERULE{JointQInitEq0}{
  \fv\PB\cap \QQ_1\QQ'_2=\varnothing
}{
  \pB \rhl
  {\PB \cap \pb\paren{\RR_1\QQ_1 \quanteq\RR'_2\QQ'_2}}
  {\Qinit \QQ e}  {\Qinit{\QQ'}{e'}}
  {\PB \cap \paren{\RR_1 \quanteq \RR'_2}}
}

\begin{proof}[of \rulerefx{JointQInitEq}]
  By \qrhlautoref{lemma:pure}, it is sufficient to show that for all
  $m_1,m_2,\psi_1,\psi_2$ with normalized
  $\psi_1\otimes\psi_2\in\denotee{\PA'}{m_1m_2}$, there is a separable
  state $\rho'$ with $\partr1{}\rho'=\denotc{\Qinit\QQ e}\pb\paren{\proj{\basis{}{m_1}\otimes\psi_1}}$ and
  $\partr2{}\rho'=\denotc{\Qinit{\QQ'}{e'}}\pb\paren{\proj{\basis{}{m_2}\otimes\psi_2}}$ and $\rho'$ satisfies~$\PB'$.

  Since in the following proof, we will use the same $m_1,m_2$
  throughout, for ease of notation, we will simply write $\PA'$
  instead of $\denotee{\PA'}{m_1m_2}$, and analogously for all other
  expressions (e.g., $\PB,U_Q,e,e'$, etc.)
  
  Since
  $(V\otimes U)\RR_1\QQ_1 \quanteq (V'\otimes U')\RR'_2\QQ'_2
  \supseteq \PA'$, we have
  $\psi_1\otimes\psi_2\in (V\otimes U)\RR_1\QQ_1 \quanteq (V'\otimes
  U')\RR'_2\QQ'_2$.  Since $\CL{\text{$U,U',V,V'$ are isometries}}
  \supseteq \PA'$, we have that that $U,U',V,V'$ are
  isometries.

  By \qrhlautoref{lemma:quanteq}, this implies that there are normalized
  $\psi_1^{QR},\psi_1^E,\psi_2^{QR},\psi_2^E$ on $\QQ_1\RR_1$,
  $\VVall_1\setminus\QQ_1\RR_1$, $\QQ_2'\RR_2'$,
  $\VVall_2\setminus \QQ_2'\RR_2'$ such that:
  $\psi_1=\psi_1^{QR}\otimes\psi_1^E$
  and
  $\psi_2=\psi_2^{QR}\otimes\psi_2^E$
  and
  $(U\otimes V) \psi_1^{QR} = (U'\otimes V')\psi_2^{QR}$.

  Let $\{\phi_z\}_{z\in Z}$ be an orthonormal basis of $\im U\cap\im U'$.
  Then $\adj U\phi_z$ are orthonormal, and $\adj{U'}\phi_z$ are orthonormal.

  Let $\psi_{1z} := \paren{\proj{\adj U\phi_z}\otimes\id}\psi_1^{QR}\otimes \psi_1^E$
  and $\psi_{2z} := \paren{\proj{\adj{U'}\phi_z}\otimes\id}\psi_2^{QR}\otimes \psi_2^E$.

  Note that $\paren{\proj{\adj U\phi_z}\otimes\id}\psi_1^{QR}$
  is of the form $\adj U\phi_z\otimes \psi_{1z}^R$ for some (not necessarily normalized)
  $\psi_{1z}^R$. And similarly $\paren{\proj{\adj{U'}\phi_z}\otimes\id}\psi_2^{QR}=
  \adj{U'}\phi_z\otimes\psi_{2z}^R$ for some $\psi_{2z}^R$. We fix those 
  $\psi_{1z}^R$ and $\psi_{2z}^R$.

  Let
  \begin{equation}
    \label{eq:psi'z}
    \psi'_z := e_1\otimes  \psi_{1z}^R \otimes \psi_1^E
    \otimes e'_2 \otimes \psi_{2z}^R \otimes \psi_2^E
    \cdot \frac{1}{ \norm{\psi_{1z}^R}^2}
    \qquad\text{and}\qquad
    \rho' :=
    \sum_z \proj{\basis{}{m_1}\otimes\basis{}{m_2}\otimes\psi'_z}.
  \end{equation}
  (Note that $e_1,e'_2$ are normalized vectors on $\QQ_1,\QQ_2'$, respectively
  because we assume that the programs in the rule are well-typed.)

  \begin{claim}\label{claim:psi12.same}
    For all $z$, $\norm{\psi_{1z}^R}=\norm{\psi_{2z}^R}$ and $V\psi_{1z}^R=V'\psi_{2z}^R$.
  \end{claim}

  \begin{claimproof}
    Then
    \begin{align}
      &{(U\otimes V)\paren{\proj{\adj U\phi_z}\otimes\id}\psi_1^{QR}}
      \starrel=
      {(U\otimes V)\paren{\proj{\adj U\phi_z}\otimes\id}(\adj UU\otimes\id)\psi_1^{QR}}
      \notag\\&
      =(\proj{U\adj U\phi_z}\otimes\id)(U \otimes V)\psi_1^{QR}
      \starstarrel= 
      (\proj{\phi_z}\otimes\id)(U \otimes V)\psi_1^{QR}
      \notag\\&
      \tristarrel=
      (\proj{\phi_z}\otimes\id)(U' \otimes V')\psi_2^{QR}
      \starstarrel= 
      (\proj{U'\adj{U'}\phi_z}\otimes\id)(U' \otimes V')\psi_2^{QR}
      \notag\\&
      =
      {(U'\otimes V')\paren{\proj{\adj{U'}\phi_z}\otimes\id}(\adj{U'}{U'}\otimes\id)\psi_1^{QR}}
      \starrel=
      {(U'\otimes V')\paren{\proj{\adj{U'}\phi_z}\otimes\id}\psi_2^{QR}}.
      \label{eq:UVphiQR}
    \end{align}
    Here $(*)$ uses that $U,U'$ are isometries.
    And $(**)$ uses that $\psi_z\in\im U\cap\im U'$.
    And $(*\mathord**)$ uses that
    $(U\otimes V)\psi_1^{QR}=(U'\otimes V')\psi_2^{QR}$ (this was
    shown above).
    Then
    \begin{multline*}
      \norm{\psi_{1z}^R}
      \starrel=
      \norm{\adj U\phi_z \otimes \psi_{1z}^R}
      \starstarrel=
      \pb\norm{\paren{\proj{\adj U\phi_z}\otimes\id}\psi_1^{QR}}
      \tristarrel=
      \pb\norm{(U\otimes V)\paren{\proj{\adj U\phi_z}\otimes\id}\psi_1^{QR}}
    \end{multline*}
    Here $(*)$ uses that $\adj U\phi_z$ is normalized. And $(**)$ is by definition of $\psi_{1z}^R$.
    And $(*{*}*)$ uses that $U,V$ are isometries. Analogously, 
    $\norm{\psi_{2z}^R} = \pb\norm{(U'\otimes V')\paren{\proj{\adj{U'}\phi_z}\otimes\id}\psi_2^{QR}}$. By \eqref{eq:UVphiQR},
    this implies $\norm{\psi_{1z}^R} = \norm{\psi_{2z}^R}$.

    Furthermore,
    \begin{multline*}
      \phi_z \otimes V\psi_{1z}^R
      \starrel=
      U\adj U\phi_z \otimes V\psi_{1z}^R
      \starstarrel= (U\otimes V)\paren{\proj{\adj U\phi_z}\otimes\id}\psi_{1z}^{QR}
        \\
        \eqrefrel{eq:UVphiQR}=
        (U'\otimes V')\paren{\proj{\adj{U'}\phi_z}\otimes\id}\psi_{2z}^{QR}
        \starstarrel=
        U'\adj{U'}\phi'_z \otimes V'\psi_{2z}^R
        \starrel=
        \phi_z \otimes V'\psi_{2z}^R.
    \end{multline*}
    Here $(*)$ holds since $\phi_z\in\im U\cap\im U'$.
    Here $(**)$ holds by definition of $\psi_{1z}^R, \psi_{2z}^R$.
    Thus $\phi_z \otimes V\psi_{1z}^R
    = \phi_z \otimes V'\psi_{2z}^R$. Since $\phi_z\neq0$, this implies $V\psi_{1z}^R=V'\psi_{2z}^R$.
  \end{claimproof}

  \begin{claim}\label{claim:My}
    Let $\gamma\in\elltwov\QQ$ be orthogonal to all $\adj U\phi_z$. Then $\paren{\proj{\gamma} \otimes \id}\psi_1=0$.

    Let $\gamma'\in\elltwov{\QQ'}$ be orthogonal to all $\adj{U'}\phi_z$. Then $\paren{\proj{\gamma'} \otimes \id}\psi_2=0$.    
  \end{claim}

  \begin{claimproof}
    \newcommand\calH{\mathcal H}
    We have $(U\otimes V)\psi_{1}^{QR}\in \im U\otimes \calH$ where $\calH$ is the range of $V,V'$.
    We also have $(U'\otimes V')\psi_{2}^{QR}\in \im U'\otimes \calH$.
    Since $(U\otimes V)\psi_1^{QR}=(U'\otimes V')\psi_2^{QR}$,
    we have $(U\otimes V)\psi_{1}^{QR}\in (\im U\otimes \calH) \cap (\im U'\otimes \calH)
    = (\im U\cap\im U')\otimes \calH$.

    Since $\phi_z$ are a basis of $\im U\cap \im U'$, this implies that 
    $(U\otimes V)\psi_{1}^{QR} \in \SPAN\{\phi_z\}_z \otimes \calH$,
    and thus $\psi_1^{QR}
    = (\adj U\otimes\adj V)(U\otimes V)\psi_{1}^{QR}
    \in \SPAN\{\adj U \phi_z\} \otimes \elltwov\RR $.
    Since $\gamma$ is orthogonal to $\adj U\phi_z$,
    we then have $(\proj\gamma\otimes\id)\psi_1^{QR} = 0$.
    And since $\psi_1=\psi_1^{QR}\otimes\psi_1^E$,
    $(\proj\gamma\otimes\id)\psi_1 = 0$.

    This shows the first half of the claim. The second half is shown analogously.
  \end{claimproof}

  \begin{claim}\label{claim:rho'.B'}
    $\rho'$ satisfies $\PB'$.
  \end{claim}

  \begin{claimproof}
    We have that $\psi_1\otimes\psi_2\in\PA'\subseteq\PB$.  Since
    $\fv\PB\cap \QQ_1\QQ_2'=\varnothing$, and
    $\proj{\adj U\phi_z},\proj{\adj{U'}\phi_z}$ operate on
    $\QQ_1,\QQ_2'$ respectively, we have that
    $\psi_{1z}\otimes\psi_{2z} = 
    (\proj{\adj U\phi_z}\otimes\id\otimes\proj{\adj{U'}\phi_z}\otimes\id)
    (\psi_1\otimes\psi_2)\in\PB$.

    We thus have
    $\psi_{1z} \otimes \psi_{2z} = \adj U\phi_z \otimes \psi_{1z}^R \otimes
    \psi_1^E \otimes \adj {U'}\phi_z \otimes \psi_{2z}^R \otimes
    \psi_2^E\in\PB$. Since $\fv\PB\cap \QQ_1\QQ_2'=\varnothing$ and
    $\adj U\phi_z, \adj{U'}\phi_z$ are on $\QQ_1,\QQ_2'$ and nonzero, it follows that also
    $ \psi'_z \eqrefrel{eq:psi'z}= e_1\otimes \psi_{1z}^R \otimes \psi_1^E \otimes e'_2
    \otimes \psi_{2z}^R \otimes \psi_2^E \cdot \frac{1}{
      \norm{\psi_{1z}}^2} \in \PB $.  And obviously
    $\psi'_z\in\SPAN\basis{\QQ_1}{e_1} \cap
    \SPAN\basis{\QQ'_2}{e'_2}$.  (Because the only difference is in
    the tensor factors in $\QQ_1\QQ_2'$.)  Furthermore, by
    \autoref{claim:psi12.same}, $V\psi_{1z}^R=V'\psi_{2z}^R$.  Thus by
    \qrhlautoref{lemma:quanteq},
    $\psi_{1z}\otimes\psi_{2z}\in \SPAN\basis{\QQ_1}{e_1} \cap
    \SPAN\basis{\QQ'_2}{e'_2}$.
    
    Hence altogether $\psi'_z\in\PB'$.
    Thus $\rho'=\sum_z\proj{\psi_z'}$ satisfies $\PB'$.
  \end{claimproof}

  \begin{claim}\label{claim:rho'12}
    $\partr1{}\rho'=\denotc{\Qinit\QQ e}\pb\paren{\proj{\basis{}{m_1}\otimes\psi_1}}$ and
    $\partr2{}\rho'=\denotc{\Qinit{\QQ'}{e'}}\pb\paren{\proj{\basis{}{m_2}\otimes\psi_2}}$.
  \end{claim}

  \begin{claimproof}
    Let $\calE(\rho):=\partr{}{\QQ}\rho \otimes \proj{e_1}$, the
    operation of initializing $\QQ$ with the state $e$. With that
    definition,
    $\denotc{\Qinit\QQ e}\pb\paren{\proj{\basis{}{m_1}} \otimes \rho}
    = \proj{\basis{}{m_1}} \otimes \calE(\rho)$.

    Let
    $\calM_z(\rho):=\proj{\adj U\phi_z}\rho\adj{\proj{\adj U\phi_z}}$.
    Since $\{\phi_z\}_{z\in Z}$ are an orthonormal basis of
    $\im U\cap\im U'$, and $U$ is an isometry,
    $\{\adj{U}\phi_z\}_{z\in Z}$ are orthonormal. We can thus extend
    that set to an orthonormal basis
    $\{\adj{U}\phi_z\}_{z\in Z} \cup \{\gamma_y\}_{y\in Y}$.  Let
    $\calM_y(\rho):=\proj{\gamma_y}\rho\adj{\proj{\gamma_y}}$.  Then
    $\calM(\rho) := \sum_z \calM_z(\rho) + \sum_y \calM_y(\rho)$ is a
    CPTPM on $\QQ$.
  
    We have
    \begin{align}
      \partr1{} \proj{\psi_z'}
      &\eqrefrel{eq:psi'z}=
      \partr1{} \proj{
        e_1\otimes\psi_{1z}^R\otimes\psi_1^E \otimes
        e_2'\otimes\psi_{2z}^R\otimes\psi_2^E}
        \cdot \tfrac1{\norm{\psi_{1z}^R}^2}
      =
        \proj{e_1\otimes\psi_{1z}^R\otimes\psi_1^E}
        \cdot \tfrac{\norm{\psi_{2z}^R}^2}{\norm{\psi_{1z}^R}^2}
      \notag\\
      &\starrel=
        \proj{e_1\otimes\psi_{1z}^R\otimes\psi_1^E}
        \starstarrel=
        \calE \pb\paren{
        \proj{\adj U\phi_z\otimes\psi_{1z}^R\otimes\psi_1^E}}
        \tristarrel=
        \calE \pb\paren{
        \proj{(\proj{\adj U\phi_z}\otimes\id)\psi_1^{QR}\otimes\psi_1^E}}
      \notag\\
      &
        =
        \calE \circ (\calM_z\otimes \id) \pb\paren{
        \proj{\psi_1^{QR}\otimes\psi_1^E}}
      \fourstarrel=
        \calE \circ (\calM_z \otimes \id) \pb\paren{
        \proj{\psi_1}}.
        \label{eq:tr1.psi'z}
    \end{align}
    Here $(*)$ follows because the fraction is $1$ by
    \autoref{claim:psi12.same}.  And $(**)$ uses that
    $\norm{\adj U\phi_z}=1$ since $\phi_z$ is normalized and $U$ is an
    isometry and $\phi_z\in\im U\cap \im U'$.
    And $(*{*}*)$ 
    by definition of $\psi_{1z}^R$, and
    $(*{*}{*}*)$ by definition of 
    $\psi_1^{QR},\psi_1^E$.

    Since $\calM$ is a CPTPM on $\QQ$, and by definition of $\calE$, we have
    $\calE \circ (\calM\otimes\id) = \calE$.
    
    Then
    \begin{align*}
      \partr{}1\rho'
      &\eqrefrel{eq:psi'z}=
      \partr{}1\sum_z\proj{\basis{}{m_1}\otimes\basis{}{m_2}\otimes \psi_z'}
      \eqrefrel{eq:tr1.psi'z}=
      \sum_z \proj{\basis{}{m_1}}\otimes \calE\pb\paren{
        (\calM_z\otimes\id) \paren{\proj{\psi_1}}}
      \\
      &=
      \proj{\basis{}{m_1}}\otimes \calE\pB\paren{\sum_z 
        (\calM_z\otimes\id) \paren{\proj{\psi_1}}}
      \\
      &\starrel=
      \proj{\basis{}{m_1}}\otimes \calE\pB\paren{\sum_z 
        (\calM_z\otimes\id) \paren{\proj{\psi_1}}
        + \sum_y
        (\calM_y\otimes\id) \paren{\proj{\psi_1}}
      }
      \\
      &=
        \proj{\basis{}{m_1}}\otimes \calE\pb\paren{(\calM\otimes\id) \paren{\proj{\psi_1}}}
      \starstarrel=
        \proj{\basis{}{m_1}}\otimes \calE\paren{{\proj{\psi_1}}}
      \\&
      \tristarrel=
        \denotc{\Qinit\QQ e}\paren{
        \proj{\basis{}{m_1}}\otimes {{\proj{\psi_1}}}}
    \end{align*}
    Here $(*)$ follows from \autoref{claim:My}.
    And $(**)$ is because $\calE\circ(\calM\otimes\id)=\calE$.
    And $(*{*}*)$ was explained after the definition of $\calE$.

    Thus we have shown
    $\partr1{}\rho'=\denotc{\Qinit\QQ
      e}\pb\paren{\proj{\basis{}{m_1}\otimes\psi_1}}$.

    $\partr2{}\rho'=\denotc{\Qinit{\QQ'}{e'}}\pb\paren{\proj{\basis{}{m_2}\otimes\psi_2}}$
    is shown analogously. (With the sole exception that we do not need
    \autoref{claim:psi12.same} to simplify the fraction in
    \eqref{eq:tr1.psi'z} because the nominator and denominator are the
    same term in this case.)
  \end{claimproof}

  As mentioned in the first paragraph of this proof,
  \autoref{claim:rho'.B'} and \autoref{claim:rho'12} implies the
  conclusion of the rule. ($\rho'$ is separable by definition.)
\end{proof}

\ERULE{JointRemoveLocal}{
  \Tilde\QQ\subseteq\QQ
  \\
  \Tilde\QQ'\subseteq\QQ'
  \\
  \fv{\PA,\PB} \cap \QQ_1\QQ'_2\XX_1\XX_2' = \varnothing
  \\
    \PA' := {\PA \cap
      \paren{\QQ_1\eqstate\psiinit{}}
    \cap
    \paren{\QQ'_2\eqstate\psiinit{}}
    \cap
    \CL{\XX_1=\initial{} \land \XX'_2 = \initial{}}}
  \\
  \PB' :=   {
    \PB \cap \CL{U,U',V,V'\text{ are unitaries}}
      \cap \pb\paren{(U\otimes V)\Tilde\QQ_1\RR_1 \quanteq (U'\otimes V')\Tilde\QQ_2\RR'_2}
  }
  \\\\
  \pb\rhl
  {\PA'}
  \bc   \bd
  {\PB'}
}{
  \pB\rhl
  \PA
  {\local{\QQ\XX}\bc}  {\local{\QQ'\XX'}\bd}
  {
    \PB \cap \paren{V\RR_1 \quanteq V'\RR'_2}
  }
}

A simpler variant of this rule (that probably illustrates the core
ideas better) is the following:

\ERULE{JointRemoveLocal0}{
  \fv{\PA,\PB} \cap \QQ_1\QQ_2\XX_1\XX_2 = \varnothing
  \\\\
  \pb\rhl
{\PA \cap
    \paren{ \QQ_1\RR_1\quanteq\QQ_2\RR_2 }
    \cap
    \CL{\XX_1=\XX_2}}
  \bc   \bd
{
    \PB
      \cap \paren{\QQ_1\RR_1 \quanteq \QQ_2\RR_2}
  }
}{
  \pB\rhl
  {\PA  \cap   \paren{ \RR_1\quanteq\RR_2 }}
  {\local{\QQ\XX}\bc}  {\local{\QQ\XX}\bd}
  {
    \PB \cap \paren{\RR_1 \quanteq \RR_2}
  }
}

We have not implemented a tactic for either of these rules in \texttt{qrhl-tool}.
During the case study \cite{pqfo-verify}, we first implemented and used a
preliminary version\footnote{Preliminary means that it was not based
  on the proven rule but instead was based on intuition and thus not
  necessarily sound.}  of a \texttt{local remove joint}
tactic. \notanonymous However, after
further improvements of the \texttt{equal} tactic (which is
essentially an application of the \rulerefx{Adversary} rule, see
\autoref{sec:adv.rule}), and consequent rewriting of the affected proofs in the
case study, it turned out that all applications of \texttt{local
  remove joint} were gone. So it seems that a tactic implementing
\rulerefx{JointRemoveLocal} is not a high priority, and thus we did
not implement a final version of the \texttt{local remove joint}
tactic. However, the \ruleref{JointRemoveLocal} is still important
because it is used to prove the improved \rulerefx{Adversary} rule
which made it possible to remove the \texttt{local remove joint}
tactic in the first place.

\begin{proof}[of \rulerefx{JointRemoveLocal}]
  By \ruleref{JointQInitEq}, we have
  \[
    \pb\rhl
    {\PB'}
    {\init{\Tilde\QQ}}
    {\init{\Tilde\QQ'}}
    {\PB \cap \paren{V\RR_1 \quanteq V'\RR'_2}}.
  \]
  (The rule needs the premise $\fv\PB \cap \Tilde\QQ_1\Tilde\QQ'_2=\varnothing$ which
  follows from the premises of \ruleref{JointRemoveLocal}.)
  
  With the premise $\rhl{\PA'}\bc\bd{\PB'}$ and \ruleref{Seq}, we get
  \[
    \pb\rhl
    {\PA'}
    {\bc;\init{\Tilde\QQ}}
    {\bd;\init{\Tilde\QQ'}}
    {\PB \cap \paren{V\RR_1 \quanteq V'\RR'_2}}.
  \]

  By consecutive application of \ruleref{RemoveLocal1} and RemoveLocal2, we get
  \[
    \pb\rhl
    {\PA}
    {\local{\QQ\XX}{\bc;\init{\Tilde\QQ }}}
    {\local{\QQ\XX}{\bd;\init{\Tilde\QQ'}}}
    {\PB \cap \paren{V\RR_1 \quanteq V'\RR'_2}}.
  \]
  (This requires
  $ \fv{\PA,\PB} \cap \QQ_1\QQ'_2\XX_1\XX_2' = \varnothing$ which is
  provided in the premises of this proof.)

  By \lemmaref{lemma:add.init.end},
  ${\local{\QQ\XX}{\bc;\init{\Tilde\QQ }}} \deneq
  {\local{\QQ\XX}{\bc}}$ and
  $ {\local{\QQ\XX}{\bd;\init{\Tilde\QQ'}}} \deneq
  {\local{\QQ\XX}{\bd}}$. (Using the premises $\Tilde\QQ\subseteq\QQ$,
  $\Tilde\QQ'\subseteq\QQ'$.)

  Thus 
  \[
    \pb\rhl
    {\PA}
    {\local{\QQ\XX}{\bc}}
    {\local{\QQ\XX}{\bd}}
    {\PB \cap \paren{V\RR_1 \quanteq V'\RR'_2}}.
  \]
\end{proof}

\begin{proof}[of \rulerefx{JointRemoveLocal0}]
  We have
  $\paren{\QQ_1\eqstate\psiinit{}} \cap
  \paren{\QQ_2\eqstate\psiinit{}} \subseteq (\QQ_1\quanteq\QQ_2)$ by
  \qrhlautoref{lemma:qeq.span}.  And
  $(\QQ_1\quanteq\QQ_2) \cap (\RR_1\quanteq\RR_2) \subseteq
  (\QQ_1\RR_1\quanteq\QQ_2\RR_2)$
  since the lhs consists of states invariant under swapping
  $\QQ_1$ and $\QQ_2$ and invariant under swapping $\RR_1$ and $\RR_2$,
  while the rhs consists of states invariant doing both those swaps. And
  $\CL{\XX_1=\initial{}\land\XX_1=\initial{}} \subseteq
  \CL{\XX_1=\XX_2}$.

  Thus
  \begin{multline*}
    \PA' := {\PA \cap
      \paren{\QQ_1\eqstate\psiinit{}}
    \cap
    \paren{\QQ'_2\eqstate\psiinit{}}
    \cap
    \CL{\XX_1=\initial{} \land \XX'_2 = \initial{}}}
  \\
  \subseteq\quad
  {\PA \cap
    \paren{ \QQ_1\RR_1\quanteq\QQ_2\RR_2 }
    \cap
    \CL{\XX_1=\XX_2}}.
\end{multline*}
we get with the qRHL judgment from the premise and \ruleref{Seq}:
\[
\rhl{\PA'}\bc\bd
{
    \PB
      \cap \paren{\QQ_1\RR_1 \quanteq \QQ_2\RR_2}
    }
  \]

  By setting $U,U',V,V' := \id$, $\QQ',\Tilde\QQ,\Tilde\QQ':=\QQ$,
  $\RR':=\RR$ in \ruleref{JointRemoveLocal}, this implies the
  conclusion of \ruleref{JointRemoveLocal0}.
\end{proof}

The \ruleref{JointQInitEq} also gives us the opportunity to easily
strengthen the \rulerefx{QrhlElimEq} rule from \cite{qrhl}.  The
\ruleref{QrhlElimEq} allows us to derive a relationship between
probabilities (which is what we eventually care about in cryptographic
proofs) from a qRHL judgment. The new rule gives us more flexibility
in terms of the variables that have to be involved in the (quantum)
equalities in the pre-/postconditions of that qRHL judgment.

\ERULE{QrhlElimEqNew}{
  \rho\ \text{satisfies}\ \PA\\
  \QQ\supseteq \fv\bc \setminus \overwr\bc\\
  \QQ\supseteq \fv\bd \setminus \overwr\bd\\
  \QQ\supseteq \fv\PA \\
  \rhl{\CL{\XX_1=\XX_2} \cap \paren{\QQ_1\quanteq \QQ_2} \cap \PA_1 \cap \PA_2}
  \bc\bd{\CL{e_1 \implies f_2}}
}{
  \Pr[e:\bc(\rho)] \leq \Pr[f:\bd(\rho)]
}

The same holds with $\iff/=$ instead of $\implies/\leq$.

We rewrote the tactic \texttt{byqrhl}%
\index{byqrhl@\texttt{byqrhl} (tactic)}%
\pagelabel{page:tactic:byqrhl} in the \texttt{qrhl-tool} that
implements the \ruleref{QrhlElimEq} to require the more liberal
variable conditions from \rulerefx{QrhlElimEqNew}.

\begin{proof}
  We first show an auxiliary claim:
  \begin{claim}\label{claim:pr.init}
    For any program $\bc$ and expression $e$ and finite set of quantum
    variables~$\QQ$, $\Pr[e:(\bc;\init\QQ)(\rho)] = \Pr[e:\bc(\rho)]$.
  \end{claim}

  \begin{claimproof}
    We show that
    $\Pr[e:(\bc;\init\qq)(\rho)] = \Pr[e:\bc(\rho)]$.
    The general case $\Pr[e:(\bc;\init\QQ)(\rho)] = \Pr[e:\bc(\rho)]$
    follows by induction over $\QQ$.

    Let $\rho':=\denotc\bc(\rho)$.  If we write $\rho'$ as
    $\rho'=:\sum_m\proj{\basis{\cl\VV}{m}}\otimes \rho'_m$ (here $\cl\VV$
    is the set of all classical variables), by definition of
    $\Pr[e:\bc(\rho)]$ \qrhlautoref{def:prafter}, we have
    $\Pr[e:\bc(\rho)]=\sum_{m\text{ s.t.~}\denotee em=\true}\tr\rho'_m$.

    Let $\rho'':=\denotc{\bc;\init\qq}(\rho)$.
    Since $\init\qq = \Qinit\qq{\psiinit\qq}$ by definition,
    $\rho''=\denotc{\Qinit\qq{\psiinit\qq}}(\rho)$.
    By the semantics of the language,
    $\rho''=\sum_m \proj{\basis{\cl\VV}{m}}\otimes \rho''_m$
    where $\rho''_m=\partr{}{\QQ}\rho'_m \otimes \proj{\psiinit\qq}$.

    Note that $\tr\rho''_m=\tr\rho'_m$ since $\pb\norm{\psiinit\qq}=1$.

    Again by definition of $\Pr[\dots]$,
    we have 
    $\Pr[e:(\bc;\init\qq)(\rho)]=\sum_{m\text{ s.t.~}\denotee em=\true}\tr\rho''_m$.
    Since  $\tr\rho''_m=\tr\rho'_m$, it follows that
    $\Pr[e:(\bc;\init\qq)(\rho)] = \Pr[e:\bc(\rho)]$.
  \end{claimproof}
  
  Let $\Tilde\XX := \XX\cup\cl{\fv\bc}\cup\cl{\fv\bd}$ and
  $\Tilde\QQ := \QQ\cup\qu{\fv\bc}\cup\qu{\fv\bd}$ and
  $\tilde\bc := \init{\Tilde\QQ\setminus\QQ};\ \bc$.
  $\tilde\bd := \init{\Tilde\QQ\setminus\QQ};\ \bd$.

  We have
  \begin{align*}
    &\braces{\CL{\Tilde\XX_1=\Tilde\XX_2} \cap \paren{\Tilde\QQ_1\quanteq \Tilde\QQ_2}
    \cap \PA_1 \cap \PA_2}
    \\
    &\qquad \subseteq
      \braces{\CL{\XX_1=\XX_2} \cap \paren{\Tilde\QQ_1\quanteq \Tilde\QQ_2} \cap \PA_1 \cap \PA_2}
      && \text{(since $\XX\subseteq\Tilde\XX$)}
      \\
    &\qquad \init{\paren{\Tilde \QQ\setminus\QQ}}    \sim    \init{\paren{\Tilde \QQ\setminus\QQ}}
      \quad \braces{\CL{\XX_1=\XX_2} \cap \paren{\QQ_1\quanteq \QQ_2} \cap \PA_1 \cap \PA_2}
      && (\rulerefx{JointQInitEq0})
    \\
    &\qquad
    \bc \sim \bd
      \quad
      \braces{e_1\implies f_1}
    &&
       \text{(by assumption)}
  \end{align*}
  The application of \ruleref{JointQInitEq0} uses that
  $\fv\PA\subseteq\QQ$ by assumption and therefore
  $\fv{\PA_1\cap \PA_2} \cap
  \paren{\Tilde\QQ\setminus\QQ}_1\paren{\Tilde\QQ\setminus\QQ}_2 =
  \varnothing$.

  Thus with rules \rulerefx{Conseq} and \rulerefx{Seq}, and by
  definition of $\tilde\bc,\tilde\bd$, we have
  \[
    \rhl {\CL{\Tilde\XX_1=\Tilde\XX_2} \cap \paren{\Tilde\QQ_1\quanteq
      \Tilde\QQ_2} \cap \PA_1 \cap \PA_2} {\tilde\bc}{\tilde\bd}
  {e_1\implies f_1}.
  \]

  Furthermore, we have that $\rho$ satisfies $\PA$ (by assumption) and
  $\fv{\tilde\bc},\fv{\tilde\bd}\subseteq\Tilde\XX\Tilde\QQ$ and
  $\qu{\fv{\PA}}\subseteq\Tilde\QQ$ (by definition of
  $\tilde\bc,\tilde\bd,\Tilde\XX,\Tilde\QQ$).  Thus, by
  \ruleref{QrhlElimEq} from \cite{qrhl}, we have
  \begin{equation}
    \Pr[e:\tilde\bc(\rho)] \leq \Pr[f:\tilde\bd(\rho)].
    \label{eq:prtildebc}
  \end{equation}

  Furthermore,
  \begin{multline}
    \tilde\bc
    \deneq
    \init{\paren{\Tilde\QQ\setminus\QQ}\cap\overwr\bc};
    \init{\paren{\Tilde\QQ\setminus\QQ}\setminus\fv\bc};
    \bc
    \\
    \deneq
    \init{\paren{\Tilde\QQ\setminus\QQ}\cap\overwr\bc};
    \bc;
    \init{\underbrace{\paren{\Tilde\QQ\setminus\QQ}\setminus\fv\bc}_{=:\mathbf I}}
    \deneq
    \bc;
    \init{\mathbf{I}}
    \label{eq:tilde.bc}
  \end{multline}
  The first $\deneq$ follows since
  ${\paren{\Tilde\QQ\setminus\QQ}\cap\overwr\bc} \cup
  {\paren{\Tilde\QQ\setminus\QQ}\setminus\fv\bc} =
  {\Tilde\QQ\setminus\QQ}$.\footnote{The latter is shown in
    Isabelle/HOL, \texttt{Helping\_Lemmas.qrhlelimeq\_aux}.}  The
  second $\deneq$ follows by \lemmaref{lemma:swap} (and using that
  ${\paren{\Tilde\QQ\setminus\QQ}\setminus\fv\bc}$ and $\fv\bc$ are
  disjoint).  The third $\deneq$ follows by
  \lemmaref{lemma:init.overwr:init}
  (using that 
  $\pb\fv{\paren{\Tilde\QQ\setminus\QQ}\cap\overwr\bc}
  \subseteq\overwr\bc$).

  From \eqref{eq:tilde.bc}, we get
  $\Pr[e:\tilde\bc(\rho)]=\Pr[e:\paren{\bc;\init{\mathbf{I}}}(\rho)]$.
  Furthermore, by \autoref{claim:pr.init},
  $\Pr[e:\paren{\bc;\init{\mathbf{I}}}(\rho)]=\Pr[e:\bc(\rho)]$.
  
  Thus $\Pr[e:\tilde\bc(\rho)]
  = \Pr[e:\bc(\rho)]$.
  Analogously, we have 
  $\Pr[f:\tilde\bd(\rho)]
  = \Pr[f:\bd(\rho)]$.

  With \eqref{eq:prtildebc}, this implies
  $\Pr[e:\bc(\rho)] \leq \Pr[f:\bd(\rho)]$ and concludes the proof.

  (The case using
  $\iff/=$ instead of $\implies/\leq$
  is shown analogously, using the corresponding variant of
  \ruleref{QrhlElimEq}.)
\end{proof}

\section{Variable change in quantum equality}
\label{sec:varchange}

\ERULE{EqVarChange}{
  \typel\QQ = \typel{\QQ'}
  \\
  \typel{\Tilde\QQ} = \typel{\Tilde\QQ'}
  \\
  \abs{\typel\QQ} \leq \abs{\typel{\Tilde\QQ}}
  \text{ or }
  \abs{\typel{\Tilde\QQ}} = \infty
  \\
  \fv\PA,\fv\PB\cap\QQ_1\QQ'_2\Tilde\QQ_1\Tilde\QQ'_2 = \varnothing
  \\
  \fv\bc\cap\QQ\Tilde\QQ = \varnothing
  \\
  \fv\bd\cap\QQ'\Tilde\QQ' = \varnothing
  \\\\
  \pB \rhl
  {\PA \cap (U_S\otimes \id)\SSS_1\Tilde\QQ_1 \quanteq (U'_S\otimes\id)\SSS'_2\Tilde\QQ'_2}
  \bc \bd
  {\PB \cap (U_R\otimes \id)\RR_1\Tilde\QQ_1 \quanteq (U'_R\otimes\id)\RR'_2\Tilde\QQ'_2}
}{
  \pB \rhl
  {\PA \cap (U_S\otimes \id)\SSS_1\QQ_1 \quanteq (U'_S\otimes\id)\SSS'_2\QQ'_2}
  \bc \bd
  {\PB \cap (U_R\otimes \id)\RR_1\QQ_1 \quanteq (U'_R\otimes\id)\RR'_2\QQ'_2}
}

The above rule is implemented in \texttt{qrhl-tool} by the
\texttt{conseq qrhl}%
\index{conseq qrhl@\texttt{conseq qrhl} (tactic)}%
\pagelabel{page:tactic:conseq}
tactic, which combines the \ruleref{Seq} with
\ruleref{EqVarChange}. That is, it allows to take a given, already
proven qRHL judgment, (optionally) replaces specified variables in the
quantum equality in the pre/postcondition by new variables, and then
uses the resulting qRHL judgment $X$ to solve the current subgoal $Y$
using the \rulerefx{Seq} rule. (This means the current subgoal $Y$
must be a qRHL judgment as well with the same programs as $X$, and new
subgoals are created to show the relationship between the
pre/postconditions of $X$ and $Y$.)

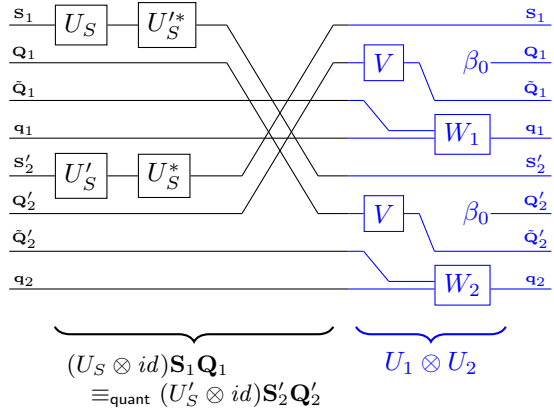
\begin{figure}[t]
  \centering
  \begin{tikzpicture}
    \initializeCircuit;
    \newWires{S1,Q1,Q1t,q1,S2,Q2,Q2t,q2,underbrace};
    \stepForward{2mm};
    \labelWire[\tiny $\SSS_1$]{S1};
    \labelWire[\tiny $\QQ_1$]{Q1};
    \labelWire[\tiny $\Tilde\QQ_1$]{Q1t};
    \labelWire[\tiny $\qq_1$]{q1};
    \labelWire[\tiny $\SSS'_2$]{S2};
    \labelWire[\tiny $\QQ'_2$]{Q2};
    \labelWire[\tiny $\Tilde\QQ'_2$]{Q2t};
    \labelWire[\tiny $\qq_2$]{q2};
    \stepForward{4mm};
    \coordinate (underbrace1-start) at (\getWireCoord{underbrace} -| \currentXPos);
    \node[gate=S1] (US) {$U_S$};
    \node[gate=S2] (US') {$U_S'$};
    \stepForward{4mm};
    \node[gate=S1] (US) {$\adj{U_S'}$};
    \node[gate=S2] (US') {$\adj{U_S}$};
    \stepForward{4mm};
    \drawWires{Q1,S1};
    \stepForward{2mm};
    \drawWires{Q2,S2};
    \stepForward{10mm};
    \crossWire{Q1}{Q2}; 
    \crossWire{S1}{S2};
    \stepForward{2mm};
    \crossWire{Q2}{Q1};
    \crossWire{S2}{S1};
    \skipWires{Q1,S1};
    \stepForward{-2mm};
    \skipWires{Q2,S2};
    \coordinate (underbrace1-stop) at (\getWireCoord{underbrace} -| \currentXPos);
    \draw[thick, decoration={brace,mirror,amplitude=2mm}, decorate]
    (underbrace1-start) -- ($(underbrace1-stop) + (2mm,0)$)
    node [pos=0.5,anchor=north,yshift=-2mm] {\small$\begin{array}{l} (U_S\otimes \id)\SSS_1\QQ_1 \\\quad \quanteq (U'_S\otimes\id)\SSS'_2\QQ'_2 \end{array}  $};
    \stepForward{2mm};
    \drawWires{S1,Q1,Q1t,q1,S2,Q2,Q2t,q2};
    \tikzset{every path/.style={color=blue}};
    \stepForward{2mm};
    \coordinate (underbrace2-start) at (\getWireCoord{underbrace} -| \currentXPos);
    \node[gate=Q1] (V1) {$V$};
    \node[gate=Q2] (V2) {$V$};
    \drawWires{Q1t,Q2t};
    \stepForward{-2mm};
    \newWireRelative{W1}{q1}{1mm};
    \newWireRelative{W2}{q2}{1mm};
    \crossWire{Q1t}{W1};
    \crossWire{Q2t}{W2};
    \stepForward{2mm};
    \drawWires{Q1,Q2};
    \stepForward{2mm};
    \crossWire{Q1}{Q1t};
    \crossWire{Q2}{Q2t};
    \skipWires{Q1t,Q2t};
    \stepForward{0mm};
    \node[gate={W1,q1}] (W1) {$W_1$};
    \node[gate={W2,q2}] (W2) {$W_2$};
    \stepForward{0mm};
    \node[wireInput={Q1}] (beta01) {$\beta_0$};
    \node[wireInput={Q2}] (beta02) {$\beta_0$};
    \coordinate (underbrace2-stop) at (\getWireCoord{underbrace} -| \currentXPos);
    \stepForward{6mm};
    \labelWire[\tiny $\SSS_1$]{S1};
    \labelWire[\tiny $\QQ_1$]{Q1};
    \labelWire[\tiny $\Tilde\QQ_1$]{Q1t};
    \labelWire[\tiny $\qq_1$]{q1};
    \labelWire[\tiny $\SSS'_2$]{S2};
    \labelWire[\tiny $\QQ'_2$]{Q2};
    \labelWire[\tiny $\Tilde\QQ'_2$]{Q2t};
    \labelWire[\tiny $\qq_2$]{q2};
    \stepForward{2mm};
    \drawWires{S1,Q1,Q1t,q1,S2,Q2,Q2t,q2};
    \draw[thick, decoration={brace,mirror,amplitude=2mm}, decorate]
    ($(underbrace2-start) - (1mm,0)$) -- ($(underbrace2-stop) + (2mm,0)$)
    node [pos=0.5,anchor=north,yshift=-2mm] {$U_1\otimes U_2$};
  \end{tikzpicture}
  
  \caption{Circuit: Quantum equality swap and $U_1\otimes U_2$}
  \label{fig:circuit-qeq-u}
\end{figure}

\begin{figure}[t]
  \centering
\begin{tikzpicture}
    \initializeCircuit;
    \tikzset{every path/.style={color=blue}};
    \newWires{S1,Q1,Q1t,q1,S2,Q2,Q2t,q2,underbrace};
    \stepForward{2mm};
    \labelWire[\tiny $\SSS_1$]{S1};
    \labelWire[\tiny $\QQ_1$]{Q1};
    \labelWire[\tiny $\Tilde\QQ_1$]{Q1t};
    \labelWire[\tiny $\qq_1$]{q1};
    \labelWire[\tiny $\SSS'_2$]{S2};
    \labelWire[\tiny $\QQ'_2$]{Q2};
    \labelWire[\tiny $\Tilde\QQ'_2$]{Q2t};
    \labelWire[\tiny $\qq_2$]{q2};
    \stepForward{4mm};
    \coordinate (underbrace2-start) at (\getWireCoord{underbrace} -| \currentXPos);
    \node[gate=Q1] (V1) {$V$};
    \node[gate=Q2] (V2) {$V$};
    \drawWires{Q1t,Q2t};
    \stepForward{-2mm};
    \newWireRelative{W1}{q1}{1mm};
    \newWireRelative{W2}{q2}{1mm};
    \crossWire{Q1t}{W1};
    \crossWire{Q2t}{W2};
    \stepForward{2mm};
    \drawWires{Q1,Q2};
    \stepForward{2mm};
    \crossWire{Q1}{Q1t};
    \crossWire{Q2}{Q2t};
    \skipWires{Q1t,Q2t};
    \stepForward{0mm};
    \node[gate={W1,q1}] (W1) {$W_1$};
    \node[gate={W2,q2}] (W2) {$W_2$};
    \stepForward{0mm};
    \node[wireInput={Q1}] (beta01) {$\beta_0$};
    \node[wireInput={Q2}] (beta02) {$\beta_0$};
    \coordinate (underbrace2-stop) at (\getWireCoord{underbrace} -| \currentXPos);
    \draw[thick, decoration={brace,mirror,amplitude=2mm}, decorate]
    (underbrace2-start) -- ($(underbrace2-stop) + (2mm,0)$)
    node [pos=0.5,anchor=north,yshift=-2mm] {$U_1\otimes U_2$};
    \stepForward{2mm};
    \tikzset{every path/.style={color=black}};
    \stepForward{2mm};
    \coordinate (underbrace1-start) at (\getWireCoord{underbrace} -| \currentXPos);
    \node[gate=S1] (US) {$U_S$};
    \node[gate=S2] (US') {$U_S'$};
    \stepForward{4mm};
    \node[gate=S1] (US) {$\adj{U_S'}$};
    \node[gate=S2] (US') {$\adj{U_S}$};
    \stepForward{4mm};
    \drawWires{Q1t,S1};
    \stepForward{2mm};
    \drawWires{Q2t,S2};
    \stepForward{10mm};
    \crossWire{Q1t}{Q2t}; 
    \crossWire{S1}{S2};
    \stepForward{2mm};
    \crossWire{Q2t}{Q1t};
    \crossWire{S2}{S1};
    \skipWires{Q1t,S1};
    \stepForward{-2mm};
    \skipWires{Q2t,S2};
    \coordinate (underbrace1-stop) at (\getWireCoord{underbrace} -| \currentXPos);
    \draw[thick, decoration={brace,mirror,amplitude=2mm}, decorate]
    (underbrace1-start) -- ($(underbrace1-stop) + (2mm,0)$)
    node [pos=0.5,anchor=north,yshift=-2mm]
    {\small$\begin{array}{l} (U_S\otimes \id)\SSS_1\Tilde\QQ_1 \\\quad \quanteq (U'_S\otimes\id)\SSS'_2\Tilde\QQ'_2 \end{array}  $};
    \stepForward{4mm};
    \labelWire[\tiny $\SSS_1$]{S1};
    \labelWire[\tiny $\QQ_1$]{Q1};
    \labelWire[\tiny $\Tilde\QQ_1$]{Q1t};
    \labelWire[\tiny $\qq_1$]{q1};
    \labelWire[\tiny $\SSS'_2$]{S2};
    \labelWire[\tiny $\QQ'_2$]{Q2};
    \labelWire[\tiny $\Tilde\QQ'_2$]{Q2t};
    \labelWire[\tiny $\qq_2$]{q2};
    \stepForward{2mm};
    \drawWires{S1,Q1,Q1t,q1,S2,Q2,Q2t,q2};
  \end{tikzpicture}
  
  \caption{Circuit: $U_1\otimes U_2$ and quantum equality swap}
  \label{fig:circuit-u-qeq}
\end{figure}
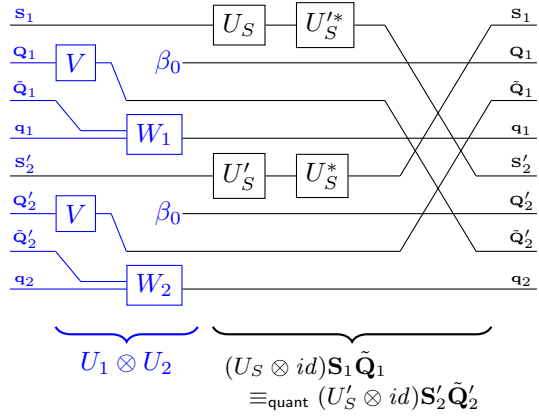

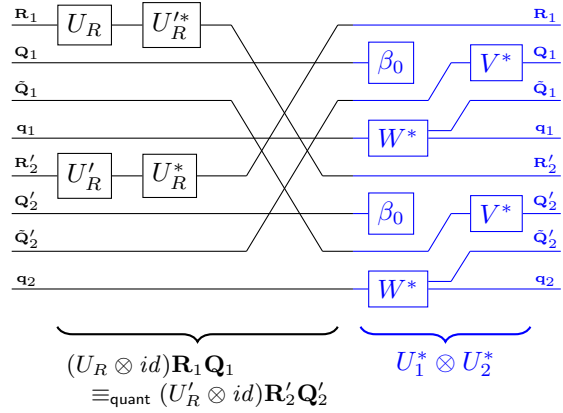
\begin{figure}[t]
  \centering
  \begin{tikzpicture}
    \initializeCircuit;
    \newWires{S1,Q1,Q1t,q1,S2,Q2,Q2t,q2,underbrace};
    \stepForward{2mm};
    \labelWire[\tiny $\RR_1$]{S1};
    \labelWire[\tiny $\QQ_1$]{Q1};
    \labelWire[\tiny $\Tilde\QQ_1$]{Q1t};
    \labelWire[\tiny $\qq_1$]{q1};
    \labelWire[\tiny $\RR'_2$]{S2};
    \labelWire[\tiny $\QQ'_2$]{Q2};
    \labelWire[\tiny $\Tilde\QQ'_2$]{Q2t};
    \labelWire[\tiny $\qq_2$]{q2};
    \stepForward{4mm};
    \coordinate (underbrace1-start) at (\getWireCoord{underbrace} -| \currentXPos);
    \node[gate=S1] (US) {$U_R$};
    \node[gate=S2] (US') {$U_R'$};
    \stepForward{4mm};
    \node[gate=S1] (US) {$\adj{U_R'}$};
    \node[gate=S2] (US') {$\adj{U_R}$};
    \stepForward{4mm};
    \drawWires{Q1t,S1};
    \stepForward{2mm};
    \drawWires{Q2t,S2};
    \stepForward{10mm};
    \crossWire{Q1t}{Q2t}; 
    \crossWire{S1}{S2};
    \stepForward{2mm};
    \crossWire{Q2t}{Q1t};
    \crossWire{S2}{S1};
    \skipWires{Q1t,S1};
    \stepForward{-2mm};
    \skipWires{Q2t,S2};
    \coordinate (underbrace1-stop) at (\getWireCoord{underbrace} -| \currentXPos);
    \draw[thick, decoration={brace,mirror,amplitude=2mm}, decorate]
    (underbrace1-start) -- ($(underbrace1-stop) + (2mm,0)$)
    node [pos=0.5,anchor=north,yshift=-2mm] {\small$\begin{array}{l} (U_R\otimes \id)\RR_1\QQ_1 \\\quad \quanteq (U'_R\otimes\id)\RR'_2\QQ'_2 \end{array}  $};
    \stepForward{2mm};
    \drawWires{S1,Q1,Q1t,q1,S2,Q2,Q2t,q2};
    \tikzset{every path/.style={color=blue}};
    \stepForward{2mm};
    \coordinate (underbrace2-start) at (\getWireCoord{underbrace} -| \currentXPos);
    \node[gate=Q1] (beta1) {$\beta_0$};
    \node[gate=Q2] (beta2) {$\beta_0$};
    \node[gate=q1] (W1) {$W^*$};
    \node[gate=q2] (W2) {$W^*$};
    \stepForward{0mm};
    \drawWires{Q1t,Q2t};
    \skipWires{W1,W2};
    \stepForward{3mm};
    \crossWire{Q1t}{Q1};
    \crossWire{Q2t}{Q2};
    \skipWires{Q1,Q2};
    \drawWires{W1,W2};
    \stepForward{2.5mm};
    \crossWire{W1}{Q1t};
    \crossWire{W2}{Q2t};
    \skipWires{Q1t,Q2t};
    \stepForward{0mm};
    \node[gate=Q1] (V1) {$V^*$};
    \node[gate=Q2] (V2) {$V^*$};
    \stepForward{0mm};
    \coordinate (underbrace2-stop) at (\getWireCoord{underbrace} -| \currentXPos);
    \stepForward{3mm};
    \labelWire[\tiny $\RR_1$]{S1};
    \labelWire[\tiny $\QQ_1$]{Q1};
    \labelWire[\tiny $\Tilde\QQ_1$]{Q1t};
    \labelWire[\tiny $\qq_1$]{q1};
    \labelWire[\tiny $\RR'_2$]{S2};
    \labelWire[\tiny $\QQ'_2$]{Q2};
    \labelWire[\tiny $\Tilde\QQ'_2$]{Q2t};
    \labelWire[\tiny $\qq_2$]{q2};
    \stepForward{2mm};
    \drawWires{S1,Q1,Q1t,q1,S2,Q2,Q2t,q2};
    \draw[thick, decoration={brace,mirror,amplitude=2mm}, decorate]
    ($(underbrace2-start) - (1mm,0)$) -- ($(underbrace2-stop) + (1mm,0)$)
    node [pos=0.5,anchor=north,yshift=-2mm] {$\adj{U_1}\otimes \adj{U_2}$};
  \end{tikzpicture}

  \caption[Circuit: Quantum equality swap and $\adj{U_1}\otimes \adj{U_2}$.]
  {Circuit: Quantum equality swap and $\adj{U_1}\otimes \adj{U_2}$.
    {} $\begin{tikzpicture}[baseline={(0,-1mm)}]
      \initializeCircuit;
      \newWire{Q};
      \stepForward{3mm};
      \node[gate=Q] (b0) {$\beta_0$};
    \end{tikzpicture}$
    stands for the projector $\proj{\beta_0}$.
  }
  \label{fig:circuit-qeq-ustar}
\end{figure}

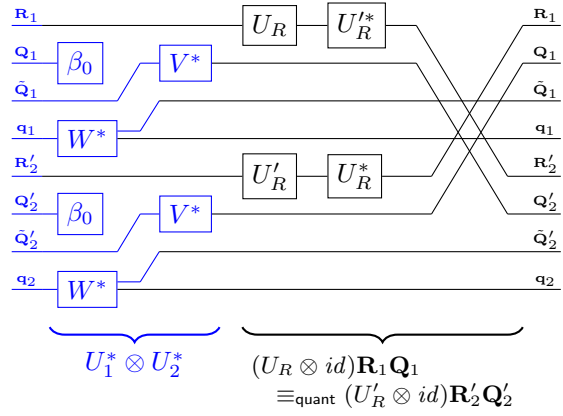
\begin{figure}[t]
  \centering
  \begin{tikzpicture}
    \initializeCircuit;
    \newWires{S1,Q1,Q1t,q1,S2,Q2,Q2t,q2,underbrace};
    \tikzset{every path/.style={color=blue}};
    \stepForward{2mm};
    \labelWire[\tiny $\RR_1$]{S1};
    \labelWire[\tiny $\QQ_1$]{Q1};
    \labelWire[\tiny $\Tilde\QQ_1$]{Q1t};
    \labelWire[\tiny $\qq_1$]{q1};
    \labelWire[\tiny $\RR'_2$]{S2};
    \labelWire[\tiny $\QQ'_2$]{Q2};
    \labelWire[\tiny $\Tilde\QQ'_2$]{Q2t};
    \labelWire[\tiny $\qq_2$]{q2};
    \stepForward{2mm};
    \drawWires{S1,Q1,Q1t,q1,S2,Q2,Q2t,q2};
    \tikzset{every path/.style={color=blue}};
    \stepForward{2mm};
    \coordinate (underbrace2-start) at (\getWireCoord{underbrace} -| \currentXPos);
    \node[gate=Q1] (beta1) {$\beta_0$};
    \node[gate=Q2] (beta2) {$\beta_0$};
    \node[gate=q1] (W1) {$W^*$};
    \node[gate=q2] (W2) {$W^*$};
    \stepForward{0mm};
    \drawWires{Q1t,Q2t};
    \skipWires{W1,W2};
    \stepForward{3mm};
    \crossWire{Q1t}{Q1};
    \crossWire{Q2t}{Q2};
    \skipWires{Q1,Q2};
    \drawWires{W1,W2};
    \stepForward{2.5mm};
    \crossWire{W1}{Q1t};
    \crossWire{W2}{Q2t};
    \skipWires{Q1t,Q2t};
    \stepForward{0mm};
    \node[gate=Q1] (V1) {$V^*$};
    \node[gate=Q2] (V2) {$V^*$};
    \stepForward{0mm};
    \coordinate (underbrace2-stop) at (\getWireCoord{underbrace} -| \currentXPos);
    \draw[thick, decoration={brace,mirror,amplitude=2mm}, decorate]
    ($(underbrace2-start) - (1mm,0)$) -- ($(underbrace2-stop) + (1mm,0)$)
    node [pos=0.5,anchor=north,yshift=-2mm] {$\adj{U_1}\otimes \adj{U_2}$};
    \tikzset{every path/.style={color=black}};
        \stepForward{4mm};
    \coordinate (underbrace1-start) at (\getWireCoord{underbrace} -| \currentXPos);
    \node[gate=S1] (US) {$U_R$};
    \node[gate=S2] (US') {$U_R'$};
    \stepForward{4mm};
    \node[gate=S1] (US) {$\adj{U_R'}$};
    \node[gate=S2] (US') {$\adj{U_R}$};
    \stepForward{4mm};
    \drawWires{Q1,S1};
    \stepForward{2mm};
    \drawWires{Q2,S2};
    \stepForward{10mm};
    \crossWire{Q1}{Q2}; 
    \crossWire{S1}{S2};
    \stepForward{2mm};
    \crossWire{Q2}{Q1};
    \crossWire{S2}{S1};
    \skipWires{Q1,S1};
    \stepForward{-2mm};
    \skipWires{Q2,S2};
    \coordinate (underbrace1-stop) at (\getWireCoord{underbrace} -| \currentXPos);
    \draw[thick, decoration={brace,mirror,amplitude=2mm}, decorate]
    (underbrace1-start) -- ($(underbrace1-stop) + (2mm,0)$)
    node [pos=0.5,anchor=north,yshift=-2mm] {\small$\begin{array}{l} (U_R\otimes \id)\RR_1\QQ_1 \\\quad \quanteq (U'_R\otimes\id)\RR'_2\QQ'_2 \end{array}  $};
    \stepForward{3mm};
    \labelWire[\tiny $\RR_1$]{S1};
    \labelWire[\tiny $\QQ_1$]{Q1};
    \labelWire[\tiny $\Tilde\QQ_1$]{Q1t};
    \labelWire[\tiny $\qq_1$]{q1};
    \labelWire[\tiny $\RR'_2$]{S2};
    \labelWire[\tiny $\QQ'_2$]{Q2};
    \labelWire[\tiny $\Tilde\QQ'_2$]{Q2t};
    \labelWire[\tiny $\qq_2$]{q2};
    \stepForward{2mm};
    \drawWires{S1,Q1,Q1t,q1,S2,Q2,Q2t,q2};
  \end{tikzpicture}

  \caption[Circuit: $\adj{U_1}\otimes \adj{U_2}$ and quantum equality swap.]
  {Circuit: $\adj{U_1}\otimes \adj{U_2}$ and quantum equality swap.
    {} $\begin{tikzpicture}[baseline={(0,-1mm)}]
      \initializeCircuit;
      \newWire{Q};
      \stepForward{3mm};
      \node[gate=Q] (b0) {$\beta_0$};
    \end{tikzpicture}$
    stands for the projector $\proj{\beta_0}$.}
  \label{fig:circuit-ustar-qeq}
\end{figure}

\begin{proof}
  Let $\PA^*,\PB^*$ denote the pre-/postcondition of the conclusion,
  and $\Tilde\PA,\Tilde\PB$ denote the pre-/postcondition of the premise.

  Fix $\psi_1\otimes\psi_2\in\PA^*$.
  We need to show that there is a separable $\rho'$ such that $\tr_2\rho'=\denotc\bc\pb\paren{\proj{\psi_1}}$
  and $\tr_1\rho'=\denotc\bd\pb\paren{\proj{\psi_2}}$ and
  $\suppo\rho'\subseteq\PB^*$.

  Let
  $S_1 := \suppo\partr{\QQ}{}\proj{\psi_1} \subseteq \elltwov{\QQ}$
  and
  $S_2 := \suppo\partr{\QQ}{}\proj{\psi_1} \subseteq \elltwov{\QQ'}$
  and $S := S_1+S_2$. ($S_1+S_2$ is
  meaningful because $\typel{\QQ}=\typel{\QQ'}$ and hence
  $\elltwov\QQ=\elltwov{\QQ'}$.)
  Since every density operator has countably dimensional support,
  $\dim S_1,\dim S_2\leq\aleph_0$
  and thus $\dim S\leq\aleph_0$.
  Furthermore, $\dim S\leq\dim\elltwov\QQ\leq\abs{\typel\QQ}$.
  Since $\abs{\typel\QQ}\leq\abs{\typel{\Tilde\QQ}}$
  or $\aleph_0\leq\abs{\typel{\tilde\QQ}}$
  by assumption of the rule, we have
  $\dim S\leq\abs{\typel{\tilde\QQ}}$.
  (Dimensions are Hilbert space dimensions, not vector space
  dimensions.)
  
  Let $\paren{\beta_i}_{i\in I}$
  be a orthonormal basis of $S$,
  and let $\paren{\beta_i}_{i\in B}$
  be an extension of that basis to the whole space $\elltwov{\QQ}$.
  We have $\abs I=\dim S\leq\abs{\typel{\Tilde Q}}$.
  Thus there exists an injection $\iota:I\to \typel{\Tilde Q}$.
  Fix such an $\iota$.

  Let $\qq$ be a fresh variable (i.e., $\qq\notin\fv{\bc,\bd,\PA^*,\PB^*,{\Tilde\PA},{\Tilde\PB}}$)
  with $\typev\qq = (\typel{\Tilde\QQ})^*$.\footnote{%
    Strictly speaking, we only require
    $\abs{\typev\qq} = \abs{(\typel{\Tilde\QQ})^*}$,
    then we can identify $\typev\qq$ with $(\typel{\Tilde\QQ})^*$.
    
    Such a $\qq$ always exists for the following reason:
    If $\typel{\Tilde\QQ}$ is infinite,
    then $\abs{\typel{\Tilde\QQ}}=\abs{\typev{\qq'}}$ for some $\qq'\in\Tilde\QQ$.
    And $\abs{\paren{\typel{\Tilde\QQ}}^*}=\abs{\typel{\Tilde\QQ}}$. Thus we need $\qq$ with $\typev\qq=\typev{\qq'}$.
    Since we assume that there are infinitely many variables of each type (see preliminaries),
    and since $\fv{\bc,\bd,\PA^*,\PB^*,{\Tilde\PA},{\Tilde\PB}}$ is finite (see preliminaries), $\qq$ exists.

    If $\typel{\Tilde\QQ}$ is finite, then $\abs{\paren{\typel{\Tilde\QQ}}^*}=\aleph_0$. Since we assume that there
    is a quantum variable of cardinality $\aleph_0$ (see preliminaries),
    and of each type there are infinitely variables, and 
     $\fv{\bc,\bd,\PA^*,\PB^*,{\Tilde\PA},{\Tilde\PB}}$ is finite, $\qq$ exists.}
  Define the bounded operator $U:\elltwov{\QQ\Tilde\QQ\qq}\to\elltwov{\QQ\Tilde\QQ\qq}$
  with
  \[
    U \pb\paren{\beta_a\otimes\basis{\Tilde\QQ}{b}\otimes{\basis\qq c}} =
    \begin{cases}
      \beta_0\otimes\basis{\Tilde\QQ}{\iota(a)}\otimes\pb\basis{\qq}{b\Vert c}
      &
      (a\in I,\ b\in\typel{\Tilde\QQ}) \\
      0
      &
      (a\notin I,\ b\in\typel{\Tilde\QQ}).
    \end{cases}
  \]
  (Here $0$ in the index of $\beta_0$ refers to an arbitrary but fixed element of $I$.)
  Note that $U$ is also an operator $\elltwov{\QQ'\Tilde\QQ'\qq}\to\elltwov{\QQ'\Tilde\QQ'\qq}$.
  In slight abuse of notation, we will also write $U$ for $U\otimes\id_{\QQ\Tilde\QQ\qq^\complement}$ or $U\otimes\id_{\QQ'\Tilde\QQ'\qq^\complement}$
  We write $U_1,U_2$ if $U$ is applied to the left/right memory, respectively.

  \begin{claim}\label{claim:in.tildeA}
    $U_1 \psi_1 \otimes U_2\psi_2
    \in \Tilde \PA$.
  \end{claim}

  \begin{claimproof}
    Since $\psi_1\otimes\psi_2\in\PA^*$, we have $\psi_1\otimes\psi_2\in\PA$.
    Since $\fv\PA\cap\QQ_1\QQ'_2\Tilde\QQ_1\Tilde\QQ'_2\qq_1\qq_2=\varnothing$ (by assumption of the rule
    ad definition of $\qq$),
    and $U_1,U_2$ operate on $\QQ_1\QQ'_2\Tilde\QQ_1\Tilde\QQ'_2\qq_1\qq_2$,
    we have $U_1 \psi_1 \otimes U_2\psi_2\in\PA$.

    Since $\psi_1\otimes\psi_2\in\PA^*$,
    we have
    $\psi_1\otimes\psi_2\in (U_S\otimes \id)\SSS_1\QQ_1 \quanteq
    (U'_S\otimes\id)\SSS'_2\QQ'_2$.  By definition of the quantum
    equality, this means that $\psi_1\otimes\psi_2$
    is invariant under the quantum circuit depicted in the black (left) part of
    \autoref{fig:circuit-qeq-u}.

    Furthermore, the blue (right) part of \autoref{fig:circuit-qeq-u} is an
    application of $U_1\otimes U_2$ if we define $V$, $W$ as follows:
    $V\beta_a := \basis{}{\iota(a)}$ for $a\in T$, $V\beta_a:=0$ otherwise.
    $W\pb\paren{\basis{\Tilde\QQ}b\otimes\basis{\qq} c} := \pb\basis{\qq}{b\Vert c}$.

    Thus the result of applying \autoref{fig:circuit-qeq-u} to
    $\psi_1\otimes\psi_2$ is $U_1\psi_1\otimes U_2\psi_2$.

    Furthermore, the blue (left) part of \autoref{fig:circuit-u-qeq}
    is $U_1\otimes U_2$. And by definition of quantum equality,
    a state is in $ (U_S\otimes \id)\SSS_1\Tilde\QQ_1 \quanteq
    (U'_S\otimes\id)\SSS'_2\Tilde\QQ'_2$ iff
    it is invariant under the black (right) part of \autoref{fig:circuit-u-qeq}.

    Thus,  $U_1\psi_1\otimes U_2\psi_2 \in 
    (U_S\otimes \id)\SSS_1\Tilde\QQ_1 \quanteq
    (U'_S\otimes\id)\SSS'_2\Tilde\QQ'_2$
    iff the result of applying 
    \autoref{fig:circuit-u-qeq} to $\psi_1\otimes\psi_2$
    is   $U_1\psi_1\otimes U_2\psi_2$.

    Furthermore, note that the circuits in
    Figures~\ref{fig:circuit-qeq-u} and~\ref{fig:circuit-u-qeq}
    compute the same function (they are identical as networks of
    linear operations). Since the result of
    \autoref{fig:circuit-qeq-u} is $U_1\psi_1\otimes U_2\psi_2$,
    so is that of \autoref{fig:circuit-u-qeq}.
    Thus   $U_1\psi_1\otimes U_2\psi_2 \in 
    (U_S\otimes \id)\SSS_1\Tilde\QQ_1 \quanteq
    (U'_S\otimes\id)\SSS'_2\Tilde\QQ'_2$
    and hence $U_1\psi_1\otimes U_2\psi_2\in \Tilde\PA$.
  \end{claimproof}

  Since $\rhl{\Tilde\PA}\bc\bd{\Tilde\PB}$
  by assumption, \autoref{claim:in.tildeA} implies that there is a
  separable state $\rho''$
  such that $\partr{1}{}\rho'' = \denotc\bc\paren{\proj{U\psi_1}}$
  and $\partr{2}{}\rho'' = \denotc\bd\paren{\proj{U\psi_2}}$
  and $\suppo\rho''\subseteq\Tilde\PB$.
  Let
  $\rho':=\toE{\adj{U_1}\otimes\adj{U_2}}\pb\paren{\rho''}$.
  (Here $*$ denotes an arbitrary but fixed element of $\typev\qq$.)

  \begin{claim}\label{claim:rho'.sep}
    $\rho'$ is separable.
  \end{claim}

  \begin{claimproof}
    $\rho''$
    is separable by definition. Thus
    $\rho'=\toE{\adj{U_1}\otimes\adj{U_2}}\pb\paren{\rho''}$ is separable.
  \end{claimproof}

  \begin{claim}\label{claim:rho'.marginals}
    $\partr{1}{}\rho' = \denotc\bc\pb\paren{\proj{\psi_1}}$
    and
    $\partr{2}{}\rho' = \denotc\bd\pb\paren{\proj{\psi_2}}$.
  \end{claim}

  \begin{claimproof}
    Since $\fv{\bc}\cap\QQ\Tilde\QQ\qq=\varnothing$ (by assumption of the rule
    ad definition of $\qq$), $\denotc\bc$ and $\toE{U_1}$ commute.
    Thus
    \begin{equation}\label{eq:rho''.U1}
      \partr1{} \rho'' = \denotc\bc\circ\toE{U_1}(\proj{\psi_1})
      = \toE{U_1}\circ\denotc\bc(\proj{\psi_1}).
    \end{equation}
    Thus $\suppo\partr1{}\rho'' \subseteq \im U_1$.
    Thus $\suppo\rho''\subseteq\im U_1\otimes\elltwo{\VVall_2}$.
    
    Analogously,  $\suppo\rho''\subseteq\im \elltwo{\VVall_1}\otimes\im U_2$.

    Let $\tilde\rho :=\toE{\id\otimes\adj{U_2}}\pb\paren{\rho''}$. Since $\suppo\rho''\subseteq \elltwov{\VVall_1}\otimes\im U_2$,
    and $U_2^*$ is an isometry on $\im U_2$ (since $U_2$ is an isometry), $\partr1{}\tilde\rho=
    \partr1{}\rho''$.
    Then
    \begin{align*}
      \partr1{}\rho'
      &=
        \partr1{} \toE{\adj{U_1}\otimes\id}(\tilde\rho)
        =
        \toE{\adj{U_1}} \circ \partr1{} \tilde \rho
        =
        \toE{\adj{U_1}} \circ \partr1{} \rho''
      \\&
        \eqrefrel{eq:rho''.U1}=
        \toE{\adj{U_1}} \circ \toE{U_1} \circ \denotc\bc\paren{\proj{\psi_1}}
        \starrel=
        \denotc\bc\paren{\proj{\psi_1}}.
    \end{align*}
    Here $(*)$ uses that $U_1$ is an isometry.

    $\adj{U_2}$, restricted to $\elltwov{\VVall_1\VVall_2}\otimes\SPAN{\basis{\qq_2}*}$, is an isometry.
    Thus $\partr{1}{}\tilde\rho = \partr{1}{}\rho''\otimes\proj{\basis{\qq_1}{*}}$.
    By definition of $\rho''$, we have $\partr{1}{}\rho'' = \denotc\bc\proj{U_1\psi_1}$.
    Thus $\partr{1}{}\tilde\rho = \denotc\bc\proj{U_1\psi_1} \otimes\proj{\basis{\qq_1}{*}}$

    $\partr{2}{}\rho' = \denotc\bd\pb\paren{\proj{\psi_2}}$ is shown analogously.
  \end{claimproof}

  \begin{claim}\label{claim:Ustar.TildeB}
    $\paren{\adj{U_1}\otimes\adj{U_2}}\Tilde\PB \subseteq \PB^*$.
  \end{claim}

  \begin{claimproof}
    Fix $\psi\in\Tilde\PB$.
    Then  $\psi\in\PB$.
    Since $\fv\PB\cap\QQ_1\QQ'_2\Tilde\QQ_1\Tilde\QQ'_2\qq_1\qq_2=\varnothing$  (by assumption of the rule
    ad definition of $\qq$),
    and $U_1,U_2$ operate on $\QQ_1\QQ_2\Tilde\QQ'_1\Tilde\QQ'_2\qq_1\qq_2$,
    we have $\paren{\adj{U_1} \otimes \adj{U_2}}\psi\in\PB$.

    Since $\psi\in\Tilde\PB$,
    we have
    $\psi\in (U_R\otimes \id)\RR_1\Tilde\QQ_1 \quanteq
    (U'_R\otimes\id)\RR'_2\Tilde\QQ'_2$.
    By definition of the quantum
    equality, this means that $\psi$
    is invariant under the quantum circuit depicted in the black (left) part of
    \autoref{fig:circuit-qeq-ustar}.

    Furthermore, the blue (right) part of \autoref{fig:circuit-qeq-ustar} is an
    application of $\adj{U_1}\otimes \adj{U_2}$ if we define $V$, $W$ as in
    the proof of \autoref{claim:in.tildeA}.

    Thus the result of applying \autoref{fig:circuit-qeq-ustar} to
    $\psi$ is $\paren{\adj{U_1}\otimes \adj{U_2}}\psi$.

    Furthermore, the blue (left) part of \autoref{fig:circuit-u-qeq}
    is $\adj{U_1}\otimes\adj{ U_2}$. And by definition of quantum equality,
    a state is in $ (U_R\otimes \id)\RR_1\QQ_1 \quanteq
    (U'_R\otimes\id)\RR'_2\QQ'_2$ iff
    it is invariant under the black (right) part of \autoref{fig:circuit-ustar-qeq}.

    Thus,  $\paren{\adj{U_1}\otimes\adj{U_2}}\psi \in 
    (U_R\otimes \id)\RR_1\QQ_1 \quanteq
    (U'_R\otimes\id)\RR'_2\QQ'_2$
    iff the result of applying 
    \autoref{fig:circuit-ustar-qeq} to $\psi$
    is $\paren{\adj{U_1}\otimes\adj{U_2}}\psi$.

    Furthermore, note that the circuits in
    Figures~\ref{fig:circuit-qeq-ustar} and~\ref{fig:circuit-ustar-qeq}
    compute the same function (they are identical as networks of
    linear operations). Since the result of
    \autoref{fig:circuit-qeq-u} is $\paren{\adj{U_1}\otimes\adj{U_2}}\psi$,
    so is that of \autoref{fig:circuit-ustar-qeq}.
    Thus   $\paren{\adj{U_1}\otimes\adj{U_2}}\psi \in 
    (U_R\otimes \id)\RR_1\QQ_1 \quanteq
    (U'_R\otimes\id)\RR'_2\QQ'_2$
    and hence $\paren{\adj{U_1}\otimes\adj{U_2}}\psi\in \PB^*$.
  \end{claimproof}
  
  \begin{claim}\label{claim:rho'.supp}
    $\suppo\rho'\subseteq\PB^*$.
  \end{claim}

  \begin{claimproof}
    By definition of $\rho''$, we have $\suppo\rho''\subseteq\Tilde\PB$. 
    Furthermore, $\rho'=\toE{\adj{U_1}\otimes\adj{U_2}}\paren{\rho''}$ by definition.
    Thus
    \[
      \suppo\rho' \subseteq \paren{\adj{U_1}\otimes\adj{U_2}} \suppo\rho''
      \subseteq \paren{\adj{U_1}\otimes\adj{U_2}} \Tilde\PB
      \starrel\subseteq \PB^*.
    \]
    Here $(*)$ is by \autoref{claim:Ustar.TildeB}.
  \end{claimproof}
  
  Since $\psi_1\otimes\psi_2\in\PA^*$ was arbitary,
  from Claims~\ref{claim:rho'.sep}, \ref{claim:rho'.marginals} and~\ref{claim:rho'.supp},
  we immediately get $\rhl{\PA^*}\bc\bd{\PB^*}$.
\end{proof}

\section{Adversary rule}
\label{sec:adv.rule}

\newcommand\aux{\qq_\mathbf{aux}}
\newcommand\Vmid{\VV\!_\mathit{mid}}
\newcommand\Vin{\VV\!_\mathit{in}}
\newcommand\Vout{\VV\!_\mathit{out}}
\newcommand\VR{\VV\!_R}
\newcommand\Eq[1]{{\equiv}#1}

Auxiliary notation: $\equiv\VV$ stands for $\qu{\VV_1}\quanteq\qu{\VV_2}\cap
\CL{\cl{\VV_1}=\cl{\VV_2}}$.

\ERULE{Adversary}{
  % From locale
  R\text{ is a predicate}
  \\
  s_i,s_i'\text{ are programs}
  \\
  C\text{ is a multi-hole context}
  \\
  \Vin,\Vmid,\Vout\text{ are finite sets of variables}
  \\
  \aux\text{ is a quantum variable}
  \\
  \VR := \{\xx: \xx_1\in R \vee \xx_2\in R\}
  \\
  \abs{\typev\aux}=\infty
  \\
  \aux\notin\VR
  \\
  \forall i. \aux\notin \fv{s_i}
  \\
  \forall i. \aux\notin \fv{s'_i}
  %
  % from lemma
  \\
  \aux \in \Vmid
  \\
  \inner C \subseteq \Vmid
  \\
  \Vout \subseteq \Vmid
  \\
  \Vout \setminus \overwr C \subseteq \Vin
  \\
  \paren{\Vout \setminus \Vin} \cap \VR = \varnothing
  \\
  \qu{\Vin} \subseteq \Vout \cup \overwr C
  \\
  \qu{\paren{\Vin \setminus \Vout}} \cap \VR = \varnothing
  \\
  \Vmid \cap \VR \subseteq \Vin \cup \covered C
  \\
  \qu{\paren{\Vmid\cap\VR}} \subseteq \Vout \cup \covered C
  \\
  \forall i. \Vmid \cap \pb\paren{\fv{s_i}\cup\fv{s'_i}} \subseteq
  \Vin \cup \covered C \cup \cl{\pb\paren{\overwr{s_i}\cap\overwr{s'_i}}}
  \\
  \forall i. \qu\Vmid \cap \pb\paren{\fv{s_i}\cup\fv{s'_i}} \subseteq
  \Vout \cup \covered C
  \\
  \fv C \subseteq \Vmid
  \\
  \fv C \subseteq \Vin \cup \overwr C
  \\
  \fv C \cap \VR \subseteq \Vin
  \\
  \qu{\fv C} \subseteq \Vout
  \\
  \VR \cap \inner C = \varnothing
  \\
  \VR \cap \written C = \varnothing
  \\
  \forall i. \rhl{R \cap \Eq{\VV_\mathit{mid}}}
  {s_i}{s'_i}
  {R \cap \Eq{\VV_\mathit{mid}}}
}{
\rhl{R \cap \Eq{\Vin}}
  {C[s_1,\dots,s_n]}
  {C[s'_1,\dots,s'_n]}
  {R \cap \Eq{\Vout}}
}

This is shown in Isabelle/HOL, as \verb|adversary_rule| in theory
\verb|Adversary_Rule.thy|. See \autoref{sec:isabelle-proofs} for
remarks about our Isabelle/HOL development.

{\TODOQ{Proof sketch}}

We have implemented this rule in the \texttt{equal}%
\index{equal@\texttt{equal} (tactic)}%
\pagelabel{page:tactic:equal}
tactic in
\texttt{qrhl-tool}.\footnote{The tactic is called \texttt{equal} because
  in its most basic form, it reasons about two programs that are
  identical (e.g., the same adversary invocation in the left and right
  program). In more advanced situations, the two programs can differ
  but the idea is still that this tactic can be applied when the last
  statement on the left/right side are ``mostly equal''.}  The tactic
allows to apply this rule to the last statement on the left/right side
(or suffix consisting of several statements, if the user chooses). The
sets $\VV_\mathit{in},\VV_\mathit{mid},\VV_\mathit{out}$ can be used
specified but the tactic makes a best effort attempt to find the
minimum sets $\VV_\mathit{in},\VV_\mathit{mid},\VV_\mathit{out}$ that
make all the premises of the rule true. The tactic also automatically
rewrites the postcondition into the form $R \cap \Eq{\Vout}$ in a way
that makes $R$ as weak as possible. For more details about the
\texttt{equal} tactic, see the manual of~\cite{qrhl-tool}.

In many situations, the tactic makes it very easy to apply the
\rulerefx{Adversary} rule. A manual application would be very
inconvenient because there are many technical conditions that need to
be checked. Yet, in our case study \cite{pqfo-verify}, in many
situations no arguments to \texttt{equal} are needed at all, and when
they were needed, it was only to specify the quantum variables in
$\VV_\mathit{mid}$ to control which variables occur (and in what
order) in the subgoals
$\rhl{R \cap \Eq{\VV_\mathit{mid}}} {s_i}{s'_i} {R \cap
  \Eq{\VV_\mathit{mid}}}$ for the mismatches $s_i,s_i'$.

\begin{comment}

\ERULE{Adversary0}{
  % From locale
  R\text{ is a predicate}
  \\
  s_i,s_i'\text{ are programs}
  \\
  C\text{ is a multi-hole context}
  \\
  \VV,\Vmid\text{ are finite sets of variables}
  \\
  \aux\text{ is a quantum variable}
  \\
  \VR := \{\xx: \xx_1\in R \vee \xx_2\in R\}
  \\
  \abs{\typev\aux}=\infty
  \\
  \aux\notin\VR
  \\
  \forall i. \aux\notin \fv{s_i}
  \\
  \forall i. \aux\notin \fv{s'_i}
  %
  % from lemma
  \\
  \aux \in \Vmid
  \\
  \Vmid \supseteq V \cup \inner C
  \\
  \Vmid \cap \VR \subseteq \VV \cup \covered C
  \\
  \forall i. \Vmid \cap \pb\paren{\fv{s_i} \cup \fv{s'_i}}
  \subseteq \cl{\paren{\overwr{s_i} \cap \overwr{s'_i}}}
  \cup \covered C \cup \VV
  \\
  \fv C \subseteq \VV
  \\
  \VR \cap \inner C = \varnothing
  \\
  \VR \cap \written C = \varnothing
  \\
  %
  %
  %
  \forall i. \rhl{R \cap \Eq{\VV_\mathit{mid}}}
  {s_i}{s'_i}
  {R \cap \Eq{\VV_\mathit{mid}}}
}{
\rhl{R \cap \Eq{\VV}}
  {C[s_1,\dots,s_n]}
  {C[s'_1,\dots,s'_n]}
  {R \cap \Eq{\VV}}
}

{\TODOQ{proof sketch, mention proof in Isabelle}}

\end{comment}

\TODOQ{Update reference to pqfo-paper}

\printbibliography

\renewcommand\symbolindexentry[4]{
  \noindent\hbox{\hbox to 2in{$#2$\hfill}\parbox[t]{3.5in}{#3}\hbox to 1cm{\hfill #4}}\\[2pt]}

\printsymbolindex

\printindex

\end{document}